\documentclass[times,sort&compress,3p]{elsarticle}
\journal{Journal of Multivariate Analysis}
\usepackage[labelfont=bf]{caption}

\usepackage{amsmath,amsfonts,amssymb,amsthm,booktabs,color,epsfig,graphicx,hyperref,url}
\usepackage{natbib}
\usepackage{smile}
\usepackage{algorithm}
\usepackage{algorithmic}

\theoremstyle{plain}% Theorem-like structures provided by amsthm.sty
\newtheorem{theorem}{Theorem}

\newtheorem{lemma}{Lemma}
\newtheorem{corollary}{Corollary}

\theoremstyle{definition}

\newtheorem{assumption}{Assumption}

\numberwithin{equation}{section}
\numberwithin{theorem}{section}
\numberwithin{corollary}{section}
\counterwithout{assumption}{section}
\numberwithin{definition}{section}

\begin{document}

\begin{frontmatter}

\title{Statistical inference for large-dimensional tensor factor model by iterative projections}

\author[1]{Matteo Barigozzi}
\author[2]{Yong He}
\author[3]{Lingxiao Li}
\author[4]{Lorenzo Trapani\corref{mycorrespondingauthor}}

\address[1]{Department of Economics, University of Bologna}
\address[2]{Institute for Financial Studies, Shandong University}
\address[3]{Department of Data Science and Artificial Intelligence, the Hong Kong Polytechnic University}
\address[4]{University of Leicester, and Universita' di Pavia}

\cortext[mycorrespondingauthor]{Corresponding author. Email address: \url{lt285@leicester.ac.uk}}

\begin{abstract}
	Tensor Factor Models (TFM) are appealing dimension reduction tools for high-order large-dimensional tensor time series, and have
	wide  applications in economics, finance and medical imaging.
	% Two  types of TFM have been proposed in the literature, essentially based on the Tucker or CP decomposition of tensors.
	In this paper, we propose a projection estimator for  the Tucker-decomposition based TFM, and provide its least-square interpretation which parallels to
	the least-square interpretation of the Principal Component Analysis  (PCA) for the vector factor model. The projection technique simultaneously reduces the dimensionality of the signal component and the magnitudes of the idiosyncratic component tensor, thus leading to an increase of the signal-to-noise ratio.
	We derive a  convergence rate
	of the projection estimator of the loadings and the common factor tensor which are faster than that of the naive PCA-based estimator. Our results are obtained under mild conditions which allow the
	idiosyncratic components to be weakly cross- and auto- correlated. We also provide a novel iterative  procedure based on the eigenvalue-ratio principle to determine the factor numbers.
	%Further motivated by the least-squares interpretation, we propose a robust version by utilizing a Huber-loss function, which leads
	%to an  iterative weighted projection technique.
	Extensive
	numerical studies are conducted to investigate the empirical performance of the proposed projection estimators relative to the
	state-of-the-art ones.
\end{abstract}

\begin{keyword} %alphabetical order
	Principal Component Analysis \sep
	Least squares \sep
	Tensor factor model \sep
	Projection Estimation.
	\MSC[2020] Primary 62H25 \sep
	Secondary 62F12
\end{keyword}

\end{frontmatter}

\section{Introduction\label{intro}}

Tensors, or high-order arrays, emerge ubiquitously in all applied sciences,
such as macroeconomic data (\cite{Chen2021Factor,BCM22}) and financial
data (\cite{Han2021Rank,He2021Matrix}) where typical
applications consider {higher-order realized moments %
	\citep{neuberger2012realized,buckle2016realised,bae2021realized,jondeau2018moment}%
}. They appear also in computer vision data (\cite{panagakis2021tensor});
neuroimaging and functional MRI data (see e.g. \cite{zhou2013tensor,Ji2021Brain,Chen2021Simultaneous}), typically
consisting of hundreds of thousands of voxels observed over time;
recommender systems (see e.g. \cite{entezari2021tensor}, and the various
references therein); data arising in psychometrics (%
\cite{carroll1970analysis,douglas1980candelinc}), and
chemometrics (see \cite{tomasi2005parafac}, and also the references in %
\cite{acar2011scalable}). For a review of further applications we also refer
to \citet{kolda2009tensor} and \citet{bi2021tensors}.

Hence, inference in the context of tensor-valued time series has now become
one of the most active research areas in statistics and machine learning.
Given that the analysis of tensor-valued data poses significant
computational challenges due to high-dimensionality, a natural approach is
to model tensor-valued time series through a low-dimensional projection on a
space of common factors. Indeed, factor models are one of the most powerful
dimension reduction tools in time series for extracting comovements among
multi-dimensional features, and, in the past decades, factor models applied
to vector-valued time series have been extensively studied - we refer,
\textit{inter alia}, to the contributions by \cite{forni2000generalized,stock2002forecasting,bai2003inferential,doz2012quasi,fan2013large,Sahalia2017Using}. In recent years, the
literature has extended the tools developed for the analysis of
vector-valued data to the context of matrix-valued data. In particular, we
refer to the seminal contribution by \cite{wang2019factor}, who proposed
Matrix Factor Model (MFM) exploiting the double low-rank structure of
matrix-valued observations, and to the subsequent paper by \citet{fan2021},
who propose the so-called $\alpha $-PCA estimation method. Furthermore, %
\citet{Yu2021Projected} study a projected version of PCA; \citet{Gao2021A}
study the Maximum Likelihood approach; and \citet{zhang2022modeling} study a
hierarchical CP product MFM. In related contributions, \cite{hkty} propose a
strong rule to determine whether there is a factor structure of matrix time
series; \citet{chen2019constrained} study the MFM under linear constraints;
and \citet{liu2019helping} and \citet{he2021online} study non-linearities in
the form of a threshold regression, or a changepoint model, respectively.

In contrast with this plethora of contributions, the statistical analysis of
tensor factor models (TFM) is still in its infancy, and most of the recent
contributions focus on \textit{i.i.d.} data, or on models with \textit{i.i.d.%
} noise, thus limiting the applicability to tensor-valued time series. In
such a context, a typical approach would be to decompose an observed $K$%
-fold tensor time series $\{\cX_{t}\in \RR^{p_{1}\times p_{2}\times \dots
	\times p_{K}},\;t\in \mathbb{Z}\}$ into a sum of a signal, or common, plus a
noise, or idiosyncratic, component, denoted as $\{\cS_{t},\;t\in \mathbb{Z}%
\} $ and $\{\cE_{t},\;t\in \mathbb{Z}\}$, respectively:
\begin{equation}
	\cX_{t}=\cS_{t}+\cE_{t}.  \label{TFM1}
\end{equation}%
However, whereas in the vector case the possible structure of the signal
component is uniquely defined as a product of a loadings matrix times a
vector of factors, in the context of a tensor factor model like (\ref{TFM1})
two possible decompositions exist: the \textit{Tucker decomposition}, and
the \textit{CP decomposition} (also known as PARAFAC or CANDECOMP
decomposition), and we refer to e.g. \citet{lettau2022high} for a recent
discussion and a comparison. The choice between the two relies mainly on the
kind of data one is dealing with and the problem one is interested in
studying. The Tucker decomposition is easier to estimate, less restrictive
and more flexible than the CP decomposition. This flexibility comes at a
cost, however, since Tucker tensor factor models, as the vector factors
models, are not uniquely identified, while CP factor models, if they exist,
are uniquely identified. An advantage of Tucker tensor factor models is that
they allow for different numbers of factors for each mode, while the CP
tensor factor model restricts the number of factors to be the same across
modes.

\subsection{Contributions of the paper\label{contribution}}

In this paper we focus on estimation of the tensor factor model based on the
Tucker-decomposition, which is defined as:
\begin{equation}
	\cS_{t}=\cF_{t}\times _{1}\Ab_{1}\times _{2}\cdots \times _{K}\Ab%
	_{K}=\sum_{j_{1}=1}^{r_{1}}\cdots \sum_{j_{K}=1}^{r_{K}}f_{j_{1},\ldots
		,j_{K},t}\left( \mathbf{a}_{1j_{1}}\circ \cdots \circ \mathbf{a}%
	_{Kj_{K}}\right) ,  \label{tucker}
\end{equation}%
where $\Ab_{k}=(\mathbf{a}_{k1}\ldots \mathbf{a}_{kr_{k}})\in \RR%
^{p_{k}\times r_{k}}$ is the $k$-th mode-wise loading matrix with rows $%
\mathbf{a}_{kj}^{\top }$, $\cF_{t}\in \RR^{r_{1}\times \cdots \times r_{K}}$
is the latent common factor tensor with entries $f_{j_{1},\ldots ,j_{K},t}$
and such that $r_{k}\ll p_{k}$, $(\times _{k})$ denotes the mode-$k$
product, and ($\circ $) denotes the tensor or outer product.

In particular, we consider the Tucker decomposition in \eqref{tucker},
adopting the definition of common factors which is typical of econometrics -
i.e., we assume that a common factor has a non-negligible impact on almost
all series, while the idiosyncratic components can be both weakly cross- and
auto-correlated.

%We first prove the equivalence between
%minimizing the least squares loss and PCA estimation under the given identifiability condition of orthonormal loadings. In other words,
We first provide the least squares interpretation of the iterative
projection for tensor factor model, which generalizes the least squares
interpretation of PCA for vector factor models \citep{bai2003inferential}
and of projection estimation for MFM \citep{Yu2021Projected,He2021Matrix}.
Second, we propose a naive estimator of the factor loadings matrices $%
\left\{ \mathbf{A}_{k},1\le k\le K\right\} $ based on the eigen-analysis of
the \textit{mode-wise sample covariance matrix}, referred to as initial
estimator in the following. Third, we provide a one-step projection method
for estimating the factor loadings matrices. Then, we provide the
theoretical convergence rates and asymptotic distributions for the initial
and the projection estimators of the loadings, and for
the tensor of common factors. These indicate that the projection technique
improves the estimation accuracy attributed to the increased signal-to-noise
ratio due to the use of information about the factors contained in all modes. Last, an iterative procedure is also provided to estimate the factor numbers consistently. In an extensive numerical study our approaches are compared with the
existing ones.

\subsection{Related literature}

\citet{zhang2022tucker} propose an iterative
projected mode-wise PCA (IPmoPCA) estimation approach. Although, potentially, our projected estimator of the loadings can also be cast into an iterative procedure (see our Algorithm \ref{alg1} in Section \ref{sec:2.1}), we show,
however, that
%, 
%at least under some restrictions on the relative rate of
%divergence of $T$, $p_{1}$,..., $p_{k}$ as they pass to infinity, 
our one-step estimator already attains the \textquotedblleft
optimal\textquotedblright\ rate of convergence, for the estimated $\mathbf{A}%
_{k}$, which would be obtained if $\left\{ \mathbf{A}_{j},1\leq j\leq
K,j\neq k\right\} $ were known in advance. This conclusion is also
reinforced by our Monte Carlo experiments, where we show that our one-step
iterative algorithm \ref{alg1} can get comparable results with the IPmoPCA approach, while being computationally more efficient. Furthermore, differently from \cite{zhang2022tucker}, we require weaker moment assumptions on the idiosyncratic components, we also thoroughly analyze the asymptotic distribution of the estimated tensor factor, and we provide a theoretical foundation of our estimator based on the least squares approach.

Estimation of (\ref{tucker}) has also been studied by \citet{Chen2021Factor} and %
\citet{Han2020Tensor,Han2021Rank} under the assumption that the idiosyncratic components
are weakly cross-correlated but \textquotedblleft white\textquotedblright ,
i.e., with no auto-correlation, and by \citet{lam2021rank} and \citet{chen2022rank}, who introduce a pre-averaging method before projection and work under the assumption that the idiosyncratic components are both weakly cross-correlated and also serially correlated. While their approach might be convenient in the case of weakly pervasive factors they do not derive any asymptotic distribution and, in the case of strong factors, their rates are comparable to ours (see our Theorem \ref{atilde}).
%Specifically, in the former case, \citet{Chen2021Factor} and %
%\citet{Han2020Tensor} consider estimation of the mode-$k$ loading space by eigen-analysis of the mode-$k$ lag-$m$ auto-covariance matrices
%for $m\neq 0$. In the latter case, \citet{chen2022rank}
%introduce a pre-averaging method before projection. 
%The determination of the
%latent ranks $r_{k}$ is also studied in \citet{Han2021Rank} and %
%\citet{lam2021rank}. 
For the CP or PARAFAC decomposition, which is a special
case of \eqref{tucker} having the factor tensor diagonal and not considered here, we refer to %
\citet{Han2021CP}, \citet{chang2021modelling}, and \citet{babiighysels2023}.
%and \citet{chang2021modelling} who study estimation under
%the assumption of white idiosyncratic components, and \citet{babiighysels2023} who allow the idiosyncratic components to be both weakly cross-correlated and also serially correlated.
%{\matt. In our paper, we discuss two estimators $\hat B_k$ (in which zhang only discussed this estimator) and $\hat B_k^*$. Indeed in zhang, they discussed that when $K>3$, they would not have asymptotic normality, but if we use $\hat B_k^*$, we still have asymptotic normality even $K>3$, while we also commented that to obtain $\hat B_k^*$, it's more time consuming. In zhang they used iterative projection while we only use one step iteration, however, we show that one-step is enough. We also give the least squares interpretation for this projection procedure while zhang did not. This is some aspects I can think of.}
\medskip

STRUCTURE. The rest of the article is organized as follows. In Section \ref{sec:2}, we
consider the Tucker decomposition and formulate the estimation of factor
loading matrices and the factor score tensor by minimizing the least squares
loss under the identifiability condition and give the KKT condition to the
optimization problem, from which it naturally leads to a projection
estimation algorithm. We also propose an iterative procedure to estimate the factor numbers. In Section \ref{sec:3}, we investigate the theoretical
properties of the initial estimator and the one-step estimator
under mild conditions. An analysis of world-trade data is in Section \ref{empirics}. We conclude in Section \ref{conclusions}. 
The proofs of the main theorems are in Section \ref{proofs}.
In  \ref{lemmas} we report all auxiliary technical results with their proofs. In \ref{sec:4} we report numerical studies based on simulated data.
% \href{https://arxiv.org/abs/2206.09800}{https://arxiv.org/abs/2206.09800}. 

\medskip
NOTATION. The following notation is used extensively henceforth. For a
tensor $\cX\in \RR^{p_{1}\times p_{2}\times \dots \times p_{K}}$, the mode-$%
k $ product with a matrix $\Ab\in \RR^{d\times p_{k}}$, denoting as $\cX%
\times _{k}\Ab$, is a tensor of size $p_{1}\times \cdots \times
p_{k-1}\times d\times p_{k+1}\times \cdots p_{K}$. Element-wise, $(\cX\times %
\Ab)_{i_{1},\ldots ,i_{k-1},j,i_{k+1},\ldots
	,i_{K}}=\sum_{i_{k}=1}^{p_{k}}x_{i_{i},\ldots ,i_{k},\ldots
	,i_{K}}a_{j,i_{k}}$. Let also $p=\prod_{k=1}^{K}p_{k}$ and $p_{-k}=p/p_{k}$.
The mode-$k$ unfolding matrix of $\cX$ is denoted by $\mathop{\mathrm{mat}}%
_{k}(\cX)$ and arranges all $p_{-k}$ mode-$k$ fibers of $\cX$ to be the
columns to get a $p_{k}\times p_{-k}$ matrix. For a matrix $\Ab$, $\Ab^{\top
}$ is the transpose of $\Ab$, $\mathop{\mathrm{Tr}}(\Ab)$ is the trace of $%
\Ab$, $\Vert \Ab\Vert $ is the spectral norm, $\Vert \Ab\Vert _{F}$ is the
Frobenious norm, and $\Vert \Ab\Vert _{\max }$ is the maximum entry of $\Ab$
in absolute value. $\Ab\otimes \Bb$ denotes the Kronecker product of
matrices $\Ab$ and $\Bb$. $\Ib_{k}$ represents a $k$-dimensional identity
matrix. For two sequences of random variables $X_{n}$ and $Y_{n}$, $%
X_{n}\lesssim Y_{n}$ means $X_{n}=O_{P}(Y_{n})$ and $X_{n}\gtrsim Y_{n}$
means $Y_{n}=O_{P}(X_{n})$. {When used, the short-hand notation $j\in
	\lbrack K]$ indicates $1\leq j\leq K$.}

\section{Methodology\label{sec:2}}

%In this section, we present the proposed approach to estimating the
%loadings, the factor tensor and the number of factors.

\subsection{Least Squares and Projection Estimation\label{sec:2.1}}

Let $\left\{ \mathcal{X}_{t},1\leq t\leq T\right\} $ be a $p_{1}\times
p_{2}\times \dots \times p_{K}$ sequence of tensor-valued random variables.
The corresponding factor model based on the Tucker decomposition is given by
\begin{equation}
	\mathcal{X}_{t}=\mathcal{F}_{t}\times _{k=1}^{K}\mathbf{A}_{k}+\mathcal{E}%
	_{t},  \label{model:TFM}
\end{equation}%
where: $\mathbf{A}_{k}$ is the $p_{k}\times r_{k}$ loading matrix for mode $%
k $, $\mathcal{F}_{t}$ is the $r_{1}\times \cdots \times r_{K}$ common
factor tensor, $\mathcal{E}_{t}$ is the $p_{1}\times \cdots \times p_{K}$
idiosyncratic component tensor, and $\mathcal{S}_{t}=\mathcal{F}_{t}\times
_{k=1}^{K}\mathbf{A}_{k}$ is the signal or common component tensor.

In order to estimate the loading matrices $\left\{ \mathbf{A}_{k},1\leq
k\leq K\right\} $ and the tensor-valued factor $\left\{ \mathcal{F}%
_{t},1\leq t\leq T\right\} $, we define the quadratic loss function
\begin{equation}
	L_{1}\left( \mathbf{A}_{1},\ldots ,\mathbf{A}_{K},\mathcal{F}_{1},\ldots ,%
	\mathcal{F}_{T}\right) :=\frac{1}{Tp}\sum_{t=1}^{T}\left\Vert \mathcal{X}%
	_{t}-\mathcal{F}_{t}\times _{k=1}^{K}\mathbf{A}_{k}\right\Vert _{F}^{2}.
	\label{equa:L1}
\end{equation}%
The least squares estimator of the loadings and of the common factor are
then defined as the solutions to the minimisation problem%
\begin{equation}
	\min_{\substack{ \left\{ \mathbf{A}_{k},1\leq k\leq K\right\}  \\ \left\{
			\mathcal{F}_{t},1\leq t\leq T\right\} }}L_{1}\left( \mathbf{A}_{1},\ldots ,%
	\mathbf{A}_{K},\mathcal{F}_{1},\ldots ,\mathcal{F}_{T}\right) .
	\label{ls-estimator}
\end{equation}
of this loss. %of the following problem:
%\begin{align}\label{equa:1}
%	\min_{\substack{\Ab_1,\ldots,\Ab_K\\\cF_1,\ldots,\cF_T}} L_{1}(\Ab_1,\ldots,\Ab_K, \cF_1,\ldots,\cF_T) ~
%	\text { s.t. } \frac{1}{p_{k}} \Ab_k^{\top} \Ab_k=\Ib_{r_k},\; 1\le k\le K,
%\end{align}
%where, without loss of generality,  we impose orthogonality of the columns of $\Ab_k$ for all $1\le k\le K$, in order to ensure the identifiability of the solution. {\matt in principle we should add other constraints to achieve identification, then these are automatically satisfied if we solve first for $\Ab_k$ and then get the factors by linear projection.}
%  it is assumed that $\Ab_k^{\top} \Ab_k/p_k=\Ib_k$ for $k=1,2,\cdots,K$. \textcolor{red}{Here $p_k$ are fixed since we are working in finite sample, then we are going to make the assumption $\Ab_k^{\top} \Ab_k/p_k\to \Ib_k$. This should be clarified.} {\YH I agree.}
The minimisation problem defined in (\ref{ls-estimator}) can be solved by
exploiting the following relation, valid for all $1\leq k\leq K$ and all $%
1\leq t\leq T$
\begin{equation*}
	\left\Vert \mathcal{X}_{t}-\mathcal{F}_{t}\times _{k=1}^{K}\mathbf{A}%
	_{k}\right\Vert _{F}^{2}
	%=\left\Vert \text{mat}_{k}\left( \mathcal{X}%
	%_{t}\right) -\text{mat}_{k}(\mathcal{F}_{t}\times _{k=1}^{K}\mathbf{A}%
	%_{k})\right\Vert _{F}^{2}
	=\left\Vert \text{mat}_{k}\left( \mathcal{X}%
	_{t}\right) -\mathbf{A}_{k}\text{mat}_{k}\left( \mathcal{F}_{t}\right)
	\left( \otimes _{j\in \lbrack K]\backslash \{k\}}\mathbf{A}_{j}\right)
	^{\top }\right\Vert _{F}^{2},
\end{equation*}%
where $\otimes _{j\in \lbrack K]\backslash \{k\}}\mathbf{A}_{j}=\mathbf{A}%
_{K}\otimes \cdots \otimes \mathbf{A}_{k+1}\otimes \mathbf{A}_{k-1}\otimes
\cdots \otimes \mathbf{A}_{1}$.
%Let $\Bb_k:=\otimes_{j \in [K]\backslash \{k\}}\Ab_j$, then  $\Bb_{k}$ must be such that $\Bb_k^{\top} \Bb_k/p_{-k}=\Ib_{r_{-k}}$ for all $k = 1,\ldots ,K$, where $p_{-k}=p/p_k$, $p=\prod_{k=1}^K p_k$, $r_{-k}=r/r_k$, $r=\prod_{k=1}^K r_k$.
For ease of notation, we will henceforth use the following quantities: $%
\mathbf{B}_{k}:=\otimes _{j\in \lbrack K]\backslash \{k\}}\mathbf{A}_{j}$; $%
\mathbf{X}_{k,t}:=\mathop{\mathrm{mat}}_{k}\left( \mathcal{X}_{t}\right) $;
and $\mathbf{F}_{k,t}:=\mathop{\mathrm{mat}}_{k}(\mathcal{F}_{t})$. Then,
for any given $1\le k\le K$, the loss (\ref{equa:L1}) can be
written as

\begin{align}
	L_{1}\left( \mathbf{A}_{1},\ldots ,\mathbf{A}_{K},\mathcal{F}_{t}\right) & =%
	\frac{1}{Tp}\sum_{t=1}^{T}\left\Vert \cX_{t}-\cF_{t}\times _{k=1}^{K}\Ab%
	_{k}\right\Vert _{F}^{2}  =\frac{1}{Tp}\sum_{t=1}^{T}\left\Vert \Xb_{k,t}-\Ab_{k}\Fb_{k,t}\Bb%
	_{k}^{\top }\right\Vert _{F}^{2}  \notag\\
	& =\frac{1}{Tp}\sum_{t=1}^{T}\Big[\mathop{\mathrm{Tr}}\left( \Xb_{k,t}^{\top
	}\Xb_{k,t}\right) -2\mathop{\mathrm{Tr}}\left( \Xb_{k,t}^{\top }\Ab_{k}\Fb%
	_{k,t}\Bb_{k}^{\top }\right) +p\mathop{\mathrm{Tr}}\left( \Fb_{k,t}^{\top }%
	\Fb_{k,t}\right) \Big]  =L_{1}\left( \Ab_{k},\Bb_{k},\Fb_{k,t}\right) .  \label{equa:L1bis} 
\end{align}

Hence, solving (\ref{ls-estimator}) is equivalent to finding $\Ab_{k}$, $\Bb%
_{k}$, and $\Fb_{k,t}$ for all $1\le k\le K$ and all $1\le t\le T$, such
that they minimize (\ref{equa:L1bis}). This problem can be solved as
follows. If $\Ab_{k}$ and $\Bb_{k}$ are given, for any $1\le t\le T$ 
%we have
the first order conditions
%\begin{equation*}
%\frac{\partial L_{1}\left( \Ab_{k},\Bb_{k},\Fb_{k,t}\right) }{\partial \Fb%
	%_{k,t}}=\mathbf{0.}
%\end{equation*}%
%From which 
give the following ordinary least squares solution:
\begin{equation}
	\Fb_{k,t}^{\ast }:=\frac{1}{p}\Ab_{k}^{\top }\Xb_{k,t}\Bb_{k},\quad 1\le
	t\le T.  \label{eq:Fstar}
\end{equation}%
Substituting $\Fb_{k,t}^{\ast }$ in \eqref{equa:L1bis} gives the
(concentrated) loss
\begin{equation}
	L_{1}\left( \Ab_{k},\Bb_{k}\right) :=\frac{1}{Tp}\sum_{t=1}^{T}\Big[%
	\mathop{\mathrm{Tr}}\left( \Xb_{k,t}^{\top }\Xb_{k,t}\right) -\frac{1}{p}%
	\mathop{\mathrm{Tr}}\left( \Xb_{k,t}^{\top }\Ab_{k}\Ab_{k}^{\top }\Xb_{k,t}%
	\Bb_{k}\Bb_{k}^{\top }\right) \Big].  \label{equa1:L1ter}
\end{equation}%
In order to identify the minimizers of \eqref{equa1:L1ter}, we impose the
following identifying constraints
\begin{equation}
	\frac{\Ab_{k}^{\top }\Ab_{k}}{p_{k}}=\Ib_{r_{k}}\;\text{ and }\;\frac{\Bb%
		_{k}^{\top }\Bb_{k}}{p_{-k}}=\Ib_{r_{-k}},  \label{eq:AABB}
\end{equation}%
which are $r(r+1)/2$ constraints, with $r=\prod_{k=1}^{K}r_{k}$. These are
in agreement of the idea that factors are pervasive along each mode (see
Assumption \ref{as-2}\textit{(i)} below).

Hence, minimizing \eqref{equa1:L1ter} subject to the constraint in %
\eqref{eq:AABB} is equivalent to minimizing the Lagrangian:
\begin{equation}
	\mathcal{L}_{1}\left( \Ab_{k},\Bb_{k},\bm\Theta ,\bm\Lambda \right)
	:=L_{1}\left( \Ab_{k},\Bb_{k}\right) +\mathop{\mathrm{Tr}}\left[ \bTheta%
	\left( \frac{1}{p_{k}}\Ab_{k}^{\top }\Ab_{k}-\Ib_{r_{k}}\right) \right] +%
	\mathop{\mathrm{Tr}}\left[ \bPhi\left( \frac{1}{p_{-k}}\Bb_{k}^{\top }\Bb%
	_{k}-\Ib_{r_{-k}}\right) \right] ,  \label{equa:L1quater}
\end{equation}%
where the Lagrange multipliers $\bTheta$ and $\bPhi$ are symmetric matrices
of dimensions $r_{k}\times r_{k}$ and $r_{-k}\times r_{-k}$, respectively.

Hence, the minimizers of \eqref{equa:L1quater} must solve the following
system of equations given by the Karush-Kuhn-Tucker (KKT) conditions:
\begin{equation*}
	\begin{aligned} &\frac{\partial \mathcal{L}_{1}}{\partial
			\Ab_k}=-\frac{1}{Tp} \sum_{t=1}^{T} \frac{2}{p} \Xb_{k,t} \Bb_k\Bb_k^{\top}
		\Xb_{k,t}^{\top} \Ab_k+\frac{2}{p_{k}} \Ab_k \bTheta=\mathbf 0, \qquad \frac{\partial \mathcal{L}_{1}}{\partial \Bb_k}=-\frac{1}{Tp}
		\sum_{t=1}^{T} \frac{2}{p} \Xb_{k,t}^{\top} \Ab_k \Ab_k^{\top} \Xb_{k,t}
		\Bb_k+\frac{2}{p_{-k}} \Bb_k \bPhi=\mathbf 0, \end{aligned}
\end{equation*}%
In turn, the KKT\ conditions are equivalent to solving the following linear
system:
\begin{equation}
	\left\{
	\begin{array}{l}
		\left( {\frac{1}{Tpp_{-k}}\sum_{t=1}^{\top }\Xb_{k,t}\Bb_{k}\Bb_{k}^{\top }%
			\Xb_{k,t}^{\top }}\right) {\Ab_{k}=\Ab_{k}\bTheta,} \\
		\left( {\frac{1}{Tpp_{k}}\sum_{t=1}^{\top }\Xb_{k,t}^{\top }\Ab_{k}\Ab%
			_{k}^{\top }\Xb_{k,t}}\right) {\Bb_{k}=\Bb_{k}\bPhi,}%
	\end{array}%
	\text{or}\quad \left\{
	\begin{array}{l}
		\Mb_{k}\Ab_{k}=\Ab_{k}\bTheta, \\
		\Mb_{-k}\Bb_{k}=\Bb_{k}\bPhi,%
	\end{array}%
	\right. \right.
\end{equation}%
where we have used the short-hand notation
\begin{equation*}
	\Mb_{k}=\frac{1}{Tpp_{-k}}\sum_{t=1}^{T}\Xb_{k,t}\Bb_{k}\Bb_{k}^{\top }\Xb%
	_{k,t}^{\top }\;\text{ and }\;\Mb_{-k}=\frac{1}{Tpp_{k}}\sum_{t=1}^{T}\Xb%
	_{k,t}^{\top }\Ab_{k}\Ab_{k}^{\top }\Xb_{k,t}.
\end{equation*}%
Denoting the largest $r_{k}$ eigenvalues of $\Mb_{k}$ in descending order as
$\lambda _{k,1},\ldots ,\lambda _{k,r_{k}}$ and collecting the corresponding
normalized eigenvectors into the $p_{k}\times r_{k}$ matrix $\Ub_{k}=\left( %
\bu_{k,1}\cdots \bu_{k,r_{k}}\right) $, we have the solutions $\bTheta^{\ast
}:=\diag\left( \lambda _{k,1},\ldots ,\lambda _{k,r_{k}}\right) $ and $\Ab%
_{k}^{\ast }:=\sqrt{p_{k}}\,\Ub_{k}$. Notice that, after replacing $\Ab%
_{k}^{\ast }$ into the definition of $\mathbf{F}_{k,t}^{\ast }$ in %
\eqref{eq:Fstar}, we obtain
\begin{equation*}
	\frac{1}{T}\sum_{t=1}^{T}\mathbf{F}_{k,t}^{\ast }\mathbf{F}_{k,t}^{\ast \top
	}=\frac{1}{Tp^2}\sum_{t=1}^{T}\Ab_{k}^{\ast \top }\Xb_{k,t}\Bb_{k}\Bb%
	_{k}^{\top }\Xb_{k,t}^{\top }\Ab_{k}^{\ast }=\frac{p_{-k}}p\Ab_{k}^{\ast
		\top }\Mb_{k}\Ab_{k}^{\ast }=\bm\Theta ^{\ast },
\end{equation*}%
which is diagonal by construction. In this way, for any $1\le k\le K$ we are
implicitly imposing $r_k(r_k-1)/2$ constraints, and, once this is repeated
for all $1\le k\le K$, we impose the remaining $r(r-1)/2$ constraints needed
to fully identify the model.

The estimators of $\Ab_{k}$ and $\mathbf{F}_{k,t}$ defined above require to
know $\Mb_{k}$ which, in turn, requires to know the unobservable projection
matrix $\Bb_{k}$, which depends on $\left\{ \Ab_{j},j\neq k\right\} $. A
natural solution is to replace each $\Ab_{j}$, $j\neq k$, with a consistent
initial estimator $\widehat{\Ab}_{j}$. While the construction of such
initial estimators is extensively discussed in Section \ref{sec:2.1.1}, here
we discuss how to proceed once we have such estimators. Defining for short
$\widehat{\Bb}_{k} :=\otimes _{j\in \lbrack K]\backslash \{k\}}\widehat{\Ab}%
_{j}$ and $\widehat{\Yb}_{k,t} :=\frac{1}{p_{-k}}\Xb_{k,t}\widehat{\Bb}%
_{k},$%
and letting
\begin{equation}
	\widetilde{\Mb}_{k}:=\frac{1}{Tp_{k}}\sum_{t=1}^{T}\widehat{\Yb}_{k,t}%
	\widehat{\Yb}_{k,t}^{\top },  \label{eq:Mtildek}
\end{equation}%
the projected estimator of the loadings $\widetilde{\Ab}_{k}$ can be
constructed as
\begin{equation}
	\widetilde{\Ab}_{k}:=\sqrt{p_{k}}\,\widetilde{\Ub}_{k},
	\label{eq:Atildekdef}
\end{equation}%
where $\widetilde{\Ub}_{k}=\left( \widetilde{\bu}_{k,1}\cdots \widetilde{\bu}%
_{k,r_{k}}\right) $ is the $p_{k}\times r_{k}$ matrix whose columns are the
normalized eigenvectors of $\widetilde{\Mb}_{k}$ corresponding to its
largest $r_{k}$ eigenvalues. Then, by iterating the procedure for all $1\le
k \le K$, we obtain the projected loadings estimators $\left\{ \widetilde{\Ab%
}_{k},\,1\le k\le K\right\} $.

We summarize the projection procedure in Algorithm \ref{alg1} below, which
extends, to the tensor case, the algorithm proposed in \cite{Yu2021Projected}
for the matrix factor model.

\begin{algorithm}[t!]
	\caption{Least squares method for estimating the loading spaces}
	\label{alg1}
	\begin{algorithmic}[1]
		
		\REQUIRE tensor data $\{\cX_t,1\le t\le T\}$, factor numbers $r_1,\ldots,r_K$;
		
		\ENSURE factor loading matrices $\{\tilde{\Ab}_k,1\le k\le K\}$;
		
		\STATE obtain the initial estimators $\{\widehat{\Ab}_k,1\le k\le K\}$;
		
		\STATE define $\widehat{\Yb}_{k,t}:=\frac{1}{p_{-k}}\Xb_{k,t}\widehat{\Bb}_k,\text{where~}\widehat{\Bb}_k:=\otimes_{j\in \lbrack K]\backslash \{k\}}\widehat{\Ab}_j$, $1\le k\le K$;
		
		\STATE given $\{\widehat{\Yb}_{k,t},1\le k\le K\}$, define $\widetilde{\Mb}_k:=(Tp_{k})^{-1}\sum_{t=1}^T \widehat{\Yb}_{k,t}(\widehat{\Yb}_{k,t})^{\top}$, set $\widetilde{\Ab}_k$ as $\sqrt{p_k}$ times the matrix with as columns the first $r_k$ normalized eigenvectors of $\widehat{\Mb}_k$;
		
		\STATE Output the projection estimators as $\{\widetilde{\Ab}_k,1\le k\le K\}$.
	\end{algorithmic} 	
\end{algorithm}

The projection method described above can be implemented recursively by
plugging in the newly estimated $\left\{ \widetilde{\Ab}_{k},\, 1\le k\le
K\right\} $ to replace $\left\{ \widehat{\Ab},1\le k\le K\right\} $ in {Step
	2,} and by iterating {Steps 2-4, in }Algorithm \ref{alg1}. The {theoretical
	analysis (and the computational side) of the recursive solution is
	challenging; however, simulations show that the projection estimators with
	one single iteration perform sufficiently well compared with the recursive
	method. Indeed, as we show in Corollary \ref{cor:1} below, when $T\asymp
	p_{1}\asymp \cdots \asymp p_{K}$, and for a suitable choice of $\left\{
	\widehat{\Ab},1\le k\le K\right\} $, the single-iteration, projected
	estimator $\widetilde{\Ab}_{k}$ converges to $\Ab_{k}$ (up to a rotation) at
	rate $O_{P}\left( 1/\sqrt{Tp_{k}}\right) $, which is the optimal rate one
	would obtain if all the other loading matrices were known in advance.}

Finally, an estimator of the common factor tensor is obtained by
linear projection as:
\begin{equation}
	\widetilde{\cF}_{t}:=\frac{1}{p}\cX_{t}\times _{k=1}^{K}\widetilde{\Ab}%
	_{k}^{\top },\quad 1\le t\le T.  \label{eq:Ftildeest}
\end{equation}%
Notice that, for any given $1\le k\le K$, we also have:
%\begin{equation*}
	$\widetilde{\Fb}_{k,t}=\text{mat}_{k}\left( \widetilde{\cF}_{t}\right) ={p}^{-1}\widetilde{\Ab}_{k}^{\top }\Xb_{k,t}\widetilde{\Bb}_{k},\quad 1\le t\le
	T,$
%\end{equation*}%
where $\widetilde{\Bb}_{k}:=\otimes _{j\in \lbrack K]\backslash \{k\}}%
\widetilde{\Ab}_{j}$. It is easy to see that this estimator is the least
squares estimator of the mode-$k$ matricization of the common factor tensor
in \eqref{eq:Fstar}, computed using the estimated loadings matrices.

%{\color{red} THIS SHOULD NOT GO HERE BUT MAYBE IN INTRODUCTION\newline
	%Theoretical analysis of the recursive solution is challenging, but the
	%simulation results show that the projection estimators with a single
	%iteration perform sufficiently well compared with the recursive method.
	%Actually, with $T\asymp p_{1}\asymp p_{2}\cdots \asymp p_{K}$ and $\left\{
	%\widehat{\Ab},1\le k\le K\right\} $ chosen suitably, it can be proved that
	%the projected estimator $\widetilde{\Ab}_{k}$ converges to $\Ab_{k}$ after
	%rotation at rate $O_{P}\left( 1/\sqrt{Tp_{k}}\right) $ in terms of the
	%averaged squared errors, which is the optimal rate one would obtain if all
	%the other loading matrices were known in advance.}

\subsection{Initial projection matrices\label{sec:2.1.1}}

We now discuss the choice of the initial projection matrices $\widehat{\Ab}%
_{k}$. For the sake of notational simplicity, let $\Eb_{k,t}:=%
\mathop{\mathrm{mat}}_{k}\left( \cE_{t}\right) $. Then, for any given $1\le
k\le K$, the mode-$k$ matricization of model (\ref{model:TFM}) is given by
\begin{equation}
	\Xb_{k,t}=\Ab_{k}\Fb_{k,t}\Bb_{k}^{\top }+\Eb_{k,t}.  \label{equ:matkma}
\end{equation}%
Furthermore, letting $\Fb_{k,t}=\left( {\bbf}_{k,t,\cdot 1}\cdots {\bbf}%
_{k,t,\cdot r_{-k}}\right) $ and $\Eb_{k,t}=\left( {\be}_{k,t,\cdot 1}\cdots
{\be}_{k,t,\cdot p_{-k}}\right) $, each of the $p_{-k}$ columns of $\Xb%
_{k,t} $ follows a vector factor model
\begin{equation}
	\bx_{k,t,\cdot j}=\Ab_{k}\Fb_{k,t}\Bb_{k,j\cdot }^{\top }+\be_{k,t,\cdot j}=%
	\Ab_{k}\overline{\bbf}_{k,t,\cdot j}+\be_{k,t,\cdot j},\quad 1\le t\le T,\ \
	1\le j\le p_{-k}.  \label{equ2.2}
\end{equation}%
Hence, in order to estimate $\Ab_{k}$, we can consider each column of $\Xb%
_{k,t}$ as an individual vector-valued time series, and apply the classical
PCA estimator for vector time series \citep{bai2003inferential}.

Whilst details are in the next sections, here we offer a heuristic preview
of the intuition behind PCA in this tensor setting. For any given $1\le k\le
K$ define the scaled mode-wise sample covariance matrix as
\begin{equation}
	\widehat{\Mb}_{k}:=\frac{1}{Tp}\sum_{t=1}^{T}\sum_{j=1}^{p_{-k}}\bx%
	_{k,t,\cdot j}\bx_{k,t,\cdot j}^{\top }=\frac{1}{Tp}\sum_{t=1}^{T}\Xb_{k,t}%
	\Xb_{k,t}^{\top }.  \label{eq:Mhatk}
\end{equation}%
Under the usual assumption of weak dependence between factors and
idiosyncratic components, it approximately holds that
% \textcolor{red}{I don't see the first term how the $p_{-k}$ disappears. Also this is beacuse we assume weak dep between factor and idio, so no cross product is present} {\YH note that $B_k^\top B_k/p_{-k}=I$}
\begin{eqnarray}
	\widehat{\Mb}_{k} &\approx &\frac{1}{p_{k}}\Ab_{k}\left( \frac{1}{Tp_{-k}}%
	\sum_{t=1}^{T}\sum_{j=1}^{p_{-k}}\overline{\bbf}_{k,t,\cdot j}\overline{\bbf}%
	_{k,t,\cdot j}^{\top }\right) \Ab_{k}^{\top }+\frac{1}{p_{k}}\frac{1}{Tp_{-k}%
	}\sum_{t=1}^{T}\sum_{j=1}^{p_{-k}}\be_{k,t,\cdot j}\be_{k,t,\cdot j}^{\top }
	\notag  \label{hatM1} \\
	&=&\frac{1}{p_{k}}\Ab_{k}\left( \frac{1}{T}\sum_{t=1}^{T}\sum_{j=1}^{r_{-k}}%
	\Fb_{k,t,\cdot j}\Fb_{k,t,\cdot j}^{\top }\right) \Ab_{k}^{\top }+\frac{1}{%
		Tp_{k}}\bigg(\frac{1}{p_{-k}}\sum_{t=1}^{T}\sum_{j=1}^{p_{-k}}\be_{k,t,\cdot
		j}\be_{k,t,\cdot j}^{\top }\bigg).
\end{eqnarray}%
Now, on one hand, the first term on the right-hand-side of \eqref{hatM1}
converges to a matrix of rank $r_{k}$, since the term $T^{-1}\sum_{t=1}^{T}%
\sum_{j=1}^{r_{-k}}\bbf_{k,t,\cdot j}\bbf_{k,t,\cdot j}^{\top }$ converges
to a positive definite $r_{k}\times r_{k}$ matrix, and since by means of the
identifying constraint \eqref{eq:AABB} we imposed pervasiveness of the
factors as $p_{k}\rightarrow \infty $. On the other hand, under the usual
assumption of weakly cross-correlated idiosyncratic components, the second
term on the right-hand-side of of \eqref{hatM1} is asymptotically
negligible, as $\min \left\{ T,p_{1},\ldots ,p_{K}\right\} \rightarrow
\infty $. As a result, the leading $r_{k}$ eigenvalues of $\widehat{\Mb}_{k}$
dominate over the remaining $p_{k}-r_{k}$, and, by Davis-Kahan's $\sin
(\Theta )$ theorem (\cite{davis1970rotation,Yu2015useful}%
), the $r_{k}$ leading eigenvectors of $\widehat{\Mb}_{k}$ asymptotically
span the same column space as that of the columns of $\Ab_{k}$. Therefore,
letting $\widehat{\Ub}_{k}$ be the $p_{k}\times r_{k}$ matrix having as
columns the $r_{k}$ leading normalized eigenvectors of $\widehat{\Mb}_{k}$,
we define an initial estimator the loadings matrix $\Ab_{k}$ as
\begin{equation}
	\widehat{\Ab}_{k}:=\sqrt{p_{k}}\,\widehat{\Ub}_{k}.  \label{eq:Ahatkdef}
\end{equation}%
Similarly, the estimator for $\Bb_{k}$ can be naturally chosen as $\widehat{%
	\Bb}_{k}:=\otimes _{j=1,j\neq k}^{K}\widehat{\Ab}_{j}$.

\subsubsection{Alternative initial projection matrices\label{sec:2.1.2}}

Some other choices of initial estimates of $\Ab_{k}$ are admissible as long
as two sufficient conditions stated below in (\ref{equ:3.1}) are fulfilled.
In particular, instead of setting the estimator of $\Bb_{k}$ as $\widehat{\Bb%
}_{k}$ defined above, we can choose another estimator $\widehat{\Bb}%
_{k}^{\ast }$ as $\sqrt{p_{-k}}$ times the matrix with columns being the
leading $r_{-k}$ normalized eigenvectors of $\widehat{\Mb}%
_{-k}:=(Tp)^{-1}\sum_{t=1}^{T}\Xb_{k,t}^{\top }\Xb_{k,t}$, a choice which is
again motivated by the matrix factor model form in (\ref{equ:matkma}). We
denote the resulted projection estimators, obtained using $\widehat{\Bb}%
_{k}^{\ast }$, as $\widetilde{\Ab}_{k}^{\ast }$. The detailed procedure is
given in Algorithm \ref{alg3}.

\begin{algorithm}[htbp]
	\caption{Projected method for estimating the loading spaces}
	\label{alg3}
	\begin{algorithmic}[1]
		
		\REQUIRE tensor data $\{\cX_t,1\le t\le T\}$, factor numbers $r_1,\ldots,r_K$;
		
		\ENSURE factor loading matrices $\{\widetilde{\Ab}^*_k,1\le k\le K\}$;
		
		\STATE obtain the initial estimators $\{\widehat{\Ab}_k,1\le k\le K\}$ and $\{\widehat{\Bb}^*_k,1\le k\le K\}$ by $\sqrt{p_{-k}}$ times the matrix with columns being the first $r_{-k}$ normalized eigenvectors of $\widehat{\Mb}_{-k}:=(Tp)^{-1}\sum_{t=1}^T\Xb_{k,t}^{\top}\Xb_{k,t}$;
		
		\STATE define $\widehat{\Yb}^*_{k,t}:=p_{-k}^{-1}\Xb_{k,t}\widehat{\Bb}^*_k,1\le k\le K$;
		
		\STATE given $\{\widehat{\Yb}^*_{k,t},1\le k\le K\}$, define $\widetilde{\Mb}^*_k:=(Tp_{k})^{-1}\sum_{t=1}^T \widehat{\Yb}^*_{k,t}\widehat{\Yb}_{k,t}^{*\top}$, obtain $\widetilde{\Ab}^*_k$ as $\sqrt{p_k}$ times the matrix with columns being the first $r_k$ normalized eigenvectors of $\widetilde{\Mb}^*_k$.
	\end{algorithmic} 	
\end{algorithm}

\subsection{Estimation of Factor Numbers\label{sec:2.3}}

We discuss two types of estimators of the numbers of common factors $r_{k}$,$%
\ 1\leq k\leq K$, both based on the eigenvalue-ratio principle (%
\cite{lam2012factor,ahnhorenstein13}).

A first estimator can be based on the simple mode-wise sample covariance
matrix $\widehat{\mathbf{M}}_{k}$, whereas a second one can instead be based
on the mode-wise sample covariance matrix of the projected data, i.e., $%
\widetilde{\mathbf{M}}_{k}$; indeed, the results in \citet{hkty} for
matrix-valued time series seem to suggest that using projected data results
in better estimators of the number of common factors, especially in small
samples.

In particular, we will consider the following family of modified
eigenvalue-ratio, criteria
\begin{equation}
	\hat{r}_{k}^{\,\text{IE-ER}}:=\argmax_{1\le j\leq r_{\max }}\frac{%
		\lambda _{j}\left( \widehat{\mathbf{M}}_{k}\right) }{\lambda _{j+1}\left(
		\widehat{\mathbf{M}}_{k}\right) +\widehat{c}\delta _{p_{1},...,p_{K},T}}%
	,\qquad \hat{r}_{k}^{\,\text{PE-ER}}:=\argmax_{1\le j\leq r_{\max }}%
	\frac{\lambda _{j}\left( \widetilde{\mathbf{M}}_{k}\right) }{\lambda
		_{j+1}\left( \widetilde{\mathbf{M}}_{k}\right) +\widetilde{c}\delta
		_{p_{1},...,p_{K},T}},  \label{equ:2}
\end{equation}%
defined for $1\leq k\leq K$, where $r_{\max }$ is a predetermined positive
constant such that $\max_{1\leq k\leq K}r_{k}<r_{\max }<\min \left\{
\min_{1\leq k\leq K}p_{k},T\right\} $, the constants $\widehat{c},\widetilde{%
	c}\in(0,\infty) $ are user-chosen and
\begin{equation}
	\delta _{p_{1},...,p_{K},T}=\frac{1}{\sqrt{Tp_{-k}}}+\frac{1}{p_{k}}.
	\label{delta}
\end{equation}%
Under our assumption of pervasive factors (see Assumption \ref{as-2}\textit{%
	(i)} below) and weakly correlated idiosyncratic components, the eigen-gap
between the first $r_{k}$ eigenvalues of $\widetilde{\mathbf{M}}_{k}$ ($%
\widehat{\mathbf{M}}_{k}$) and the remaining ones widens as $%
p_{k}\rightarrow \infty $. Thus, (\ref{equ:2}) will take the maximum value
at $j=r_{k}$. The rationale for the extra $\widehat{c}\delta
_{p_{1},...,p_{K},T}$ (respectively $\widetilde{c}\delta
_{p_{1},...,p_{K},T} $) term at the denominator of $\hat{r}_{k}^{\,\text{%
		IE-ER}}$ (respectively $\hat{r}_{k}^{\,\text{PE-ER}}$) is that there is
no theoretical guarantee that, for some $j>r_{k}$, $\lambda _{j+1}\left(
\widehat{\mathbf{M}}_{k}\right) $ (and, similarly, $\lambda _{j+1}\left(
\widetilde{\mathbf{M}}_{k}\right) $), will not be very small, thus
artificially inflating the ratio $\lambda _{j}\left( \widehat{\mathbf{M}}%
_{k}\right) /\lambda _{j+1}\left( \widehat{\mathbf{M}}_{k}\right) $. The
presence of $\delta _{p_{1},...,p_{K},T}$ (which is of the same order of
magnitude as the upper bound for $\lambda _{j+1}\left( \widehat{\mathbf{M}}%
_{k}\right) $ when $j\geq r_{k}$) serves the purpose of \textquotedblleft
weighing down\textquotedblright\ the eigenvalue ratio and avoid such
degeneracy.

As far as the choice of $\widehat{c}$ and $\widetilde{c}$ is concerned,
theoretically any positive, finite number is acceptable. In particular, two
approaches are possible. On the one hand, one can choose $\widehat{c}$ and $%
\widetilde{c}$ adaptively, using different subsamples and choosing the
values of $\widehat{c}$ and $\widetilde{c}$ which offer stable estimates
across such subsamples, in a similar spirit to \citet{Hallin2007Determining}
and \citet{ABC10}. Alternatively, following a similar proposal as in %
\citet{trapani2018randomized} and \citet{BT2}, one can use%
\begin{equation}
	\widehat{c}=\sum_{j=1}^{p_{k}}\lambda _{j}\left( \widehat{\mathbf{M}}%
	_{k}\right) ,\qquad \widetilde{c}=\sum_{j=1}^{p_{k}}\lambda _{j}\left(
	\widetilde{\mathbf{M}}_{k}\right) .  \label{bt2}
\end{equation}%
Further, operationally, in order to compute $\widetilde{\mathbf{M}}_{k}$ we
need to compute $\widehat{\mathbf{B}}_{k}$ first; in turn, this requires
some prior knowlegde of $r_{-k}$ (or, equivalently, of $r_{j}$, $j\neq k$).
In order to address this, we propose to determine the numbers of factors by
the following Algorithm \ref{alg2}.
\begin{algorithm}[htbp]
	\caption{Projected estimation of the numbers of factor}
	\label{alg2}
	\begin{algorithmic}[1]	
		
		\REQUIRE tensor data $\{\cX_t,1\le t\le T\}$, maximum number $r_{\max}$, maximum iterative step $m$
		
		\ENSURE factor numbers $\{\hat{r}_k^{\text{PE-ER}},1\le k\le K\}$
		
		\STATE initialize: $r_k^{(0)}=r_{\max},1\le k\le K$;
		
		\STATE given $r_k^{(0)}$, obtain the initial estimators  $\{\widehat{\Ab}_k,1\le k\le K\}$ and set $\widehat{\Ab}_k^{(0)}=\widehat{\Ab}_k$;
		
		\STATE for $1\le s\le m$, compute $\widehat{\Bb}_k^{(s)}=\otimes_{j\in \lbrack K]\backslash \{k\}}\widehat{\Ab}_j^{(s-1)}$;
		
		\STATE compute $\widetilde{\Mb}_k^{(s)}$, obtain $\hat{r}_k^{(s)}$ for $1\le k\le K$;
		
		\STATE update $\widehat{\Ab}_k^{(s)}$ as $\sqrt{p_k}$ times the matrix with as columns the  first $\hat{r}_k^{(s)}$ normalized eigenvectors  of $\widehat{\Mb}_k$;
		
		\STATE Repeat steps 3 to 5 until $\hat{r}_k^{(s)}=\hat{r}_k^{(s-1)}$, for all $1\le k\le K$, or reach the maximum number of iterations;\vskip .1cm
		
		\STATE Output the last step estimator as $\{\hat{r}_k^{\text{PE-ER}},1\le k\le K\}$.
	\end{algorithmic} 	
\end{algorithm}

\section{Theoretical results\label{sec:3}}

In this section, we present the main assumptions (Section \ref{assumptions}%
), and we then present the asymptotics of the estimated loadings (Sections %
\ref{loadings-1} and \ref{loadings-2}), of the estimated common factors and
common components (Section \ref{factors}), and of the estimators of the
numbers of common factors (Section \ref{nfactors}).

\subsection{Assumptions\label{assumptions}}

The following assumptions are required for our theory, and may be viewed as
higher-order extensions of those adopted for large-dimensional matrix factor
model by \citet{fan2021} and \citet{Yu2021Projected}.

\begin{assumption}
	\textit{\label{as-1}It holds that: (i) for all $1\leq t\leq T$, $\mathbb{E}%
		\left( \text{\upshape Vec}\left( \mathcal{F}_{t}\right) \right) =0$ and $%
		\mathbb{E}\left( \left\Vert \text{\upshape Vec}\left( \mathcal{F}_{t}\right)
		\right\Vert ^{4}\right) \leq c$ for some $c<\infty $ independent of $t$;
		(ii) for all $1\leq k\leq K$, as $\min \left\{ T,p_{1},...,p_{K}\right\}
		\rightarrow \infty $, 
		$
		\frac{1}{T}\sum_{t=1}^{T}\mathbf{F}_{k,t}\mathbf{F}_{k,t}^{\top }\overset{P}{%
			\rightarrow }\mathbf{\Sigma }_{k},
		$
		where $\mathbf{\Sigma }_{k}$ is an $r_{k}\times r_{k}$ positive definite
		matrix with finite, distinct eigenvalues and spectral decomposition $\mathbf{%
			\Sigma }_{k}=\mathbf{\Gamma }_{k}\mathbf{\Lambda }_{k}\mathbf{\Gamma }%
		_{k}^{\top }$; (iii) the factor numbers $\left\{ r_{k},1\leq k\leq K\right\}
		$ are fixed integers as $\min \left\{ T,p_{1},...,p_{K}\right\} \rightarrow
		\infty $.}
\end{assumption}

\begin{assumption}
	\label{as-2} \textit{It holds that: (i) for all $1\leq k\leq K$, $\left\Vert
		\mathbf{A}_{k}\right\Vert _{\max }\leq \overline{a}_{k}$ for some $\overline{%
			a}_{k}<\infty $ independent of $p_{k}$; (ii) $\left\Vert p_{k}^{-1}\mathbf{A}%
		_{k}^{\top }\mathbf{A}_{k}-\mathbf{I}_{k}\right\Vert _{F}\rightarrow 0$ as $%
		\min \left\{ T,p_{1},...,p_{K}\right\} \rightarrow \infty $.}
\end{assumption}

Assumption \ref{as-1} imposes finite fourth moments of the mode-$k$
unfolding factor matrices for all $1\leq k\leq K$, and it ensures that the
second-order sample moments converge to positive definite matrices $\mathbf{%
	\Sigma }_{k}$, which are also assumed to have distinct eigenvalues to ensure
the identifiability of eigenvectors. This assumption is typical in vector
factor models, and it can be compared e.g. with Assumption A in %
\citet{bai2003inferential}.

As far as Assumption \ref{as-2} is concerned, this is also a standard
requirement in the context of vector (and matrix) factor models, and we
refer to e.g. Assumption B in \citet{bai2003inferential} for comparison. In
the assumption, the common factors are assumed to be \textquotedblleft
strong\textquotedblright\ or pervasive - see, e.g., the recent contribution
%by \citet{uematsu2022estimation} and %
by \citet{bai2023approximate} for a treatment, and useful insights, on the
notion of strong versus weak common factors and the case of vector data. We
would like to point out that extensions to the case of \textquotedblleft
weak\textquotedblright\ common factors - where $\left\Vert \mathbf{A}%
_{k}\right\Vert _{F}^{2}=c_{0}p_{k}^{\alpha _{k}}$\ for some $0<\alpha
_{k}<1 $\ - are in principle possible, at the price of more complicated
algebra, even in the context of tensor-valued time series; \citet{hkty}
offer a comprehensive treatment for the case of matrix-valued time series.

\begin{assumption}
	\label{as-3}\textit{It holds that: (i) for all $1\leq t\leq T$, $1\leq
		i_{k}\leq p_{k}$ and $1\leq k\leq K$, $\mathbb{E}\left(
		e_{t,i_{1},...,i_{K}}\right) =0 $ and $\mathbb{E}\left\vert
		e_{t,i_{1},...,i_{K}}\right\vert ^{4}\leq c$ for some $c<\infty $
		independent of $t$ and $i_{k}$; (ii) for all $1\leq k\leq K$, and all $p_{k}$
		and $T$, it holds that%
		\begin{equation*}
			\frac{1}{Tp}\sum_{t,s=1}^{T}\sum_{i,l=1}^{p_{k}}\sum_{j,h=1}^{p_{-k}}\left%
			\vert \mathbb{E}\left( e_{t,k,lj}e_{s,k,ih}\right) \right\vert \leq c,
		\end{equation*}%
		for some $c<\infty $ independent of $k$, $p_{k}$ and $T$; (iii) for all 
		$1\le i,l_1\le p_k$, $1\le j,h_1\le p_{-k}$,
		$1\leq k\leq K$, and all $p_{k}$ and $T$, it holds that%
		\begin{align*}
			&\sum_{s=1}^{T}\sum_{l_{2}=1}^{p_{k}}\sum_{h_{2}=1}^{p_{-k}}\left\vert
			\text{\upshape Cov}\left( e_{t,k,ij}e_{t,k,l_{1}j},e_{s,k,ih_{2}}e_{s,k,l_{2}h_{2}}\right)
			\right\vert \leq c, \; \sum_{s=1}^{T}\sum_{l_{2}=1}^{p_{k}}\sum_{h_{2}=1}^{p_{-k}}\left\vert
			\text{\upshape Cov}\left( e_{t,k,ij}e_{t,k,ih_{1}},e_{s,k,l_{2}j}e_{s,k,l_{2}h_{2}}\right)
			\right\vert \leq c, \\
			&\sum_{s=1}^{T}\sum_{l_2=1}^{p_{k}}\sum_{h_2=1}^{p_{-k}}\left\vert \text{\upshape Cov}\left(
			e_{t,k,ij}e_{t,k,l_{1}h_{1}},e_{s,k,ij}e_{s,k,l_{2}h_{2}}\right) \right\vert
			+\left\vert \text{\upshape Cov}\left(
			e_{t,k,l_{1}j}e_{t,k,ih_{1}},e_{s,k,l_{2}j}e_{s,k,ih_{2}}\right) \right\vert
			\leq c,
		\end{align*}%
		for some $c<\infty $ independent of $k$, $p_{k}$ and $T$.}
\end{assumption}

According to Assumption \ref{as-3}, weak serial and cross-sectional
correlation are allowed along each mode; in particular, the summability
conditions in part \textit{(ii)} of the assumptions can be verified under
more primitive assumptions of weak dependence, at least across $t$; %
\citet{hkty} study the case of stationary causal processes approximable by
an $m$-dependent sequence, in the case of matrix-valued time series.
Essentially the same extension can be studied in this context. We note that
the following equations are nested within the assumption
%{\color{red} not sure these are correct we need some rescaling in front of sums don't we?}
\begin{align}
	&\sum_{s=1}^{T}\sum_{l=1}^{p_{k}}\sum_{h=1}^{p_{-k}}\left\vert \mathbb{E}%
	\left( e_{t,k,ij}e_{s,k,lh}\right) \right\vert \leq c,  \quad \sum_{l=1}^{p_{k}}\sum_{h=1}^{p_{-k}}\left\vert \mathbb{E}\left(
	e_{t,k,lj}e_{t,k,ih}\right) \right\vert \leq c,  \label{weak-dep-1} \\
	& \frac{1}{p}\sum_{i,l=1}^{p_{k}}\sum_{j,h=1}^{p_{-k}}\left\vert \mathbb{E}%
	\left( e_{t,k,lj}e_{t,k,ih}\right) \right\vert \leq c, \quad \sum_{l=1}^{p_{k}}\left\vert \mathbb{E}\left( e_{t,k,lj}e_{t,k,ij}\right)
	\right\vert \leq c,  \label{weak-dep-4}
\end{align}%
for all $1\leq i\leq p_{k}$, $1\leq j\leq p_{-k}$ and $1\leq t\leq T$. Part
\textit{(iii)} of the assumption controls the second-order correlation among
the elements of mode-$k$ unfolding matrices of the noise tensors, and it
implies existence and summability of all $4$-th order cumulants of the
process $\left\{ \mathbf{e}_{k,t},t\in \mathbb{Z}\right\} $ - in turn, this
is a necessary and sufficient condition for%
$
\frac{1}{Tp_{-k}}\sum_{t=1}^{T}\mathbf{e}_{k,t}\mathbf{e}_{k,t}^{\top }%
\overset{P}{\rightarrow }\frac{1}{p_{-k}}\mathbb{E}\left( \mathbf{e}_{k,t}%
\mathbf{e}_{k,t}^{\top }\right) ,
$ as $T\rightarrow \infty $ (see e.g. \cite{hannan1970}, pp. 209-211).
Essentially the same set of assumptions is required in the context of matrix
factor models (see \cite{Yu2021Projected} and \cite{hkty})

\begin{assumption}
	\label{as-4}\textit{It holds that: (i) for all $1\leq t\leq T$, $1\leq k\leq
		K$ and any couple of deterministic vectors $\mathbf{v}$ of dimension $p_{k}$
		and $\mathbf{w}$ of dimension $p_{-k}$ such that $\left\Vert \mathbf{v}%
		\right\Vert =1$ and $\left\Vert \mathbf{w}\right\Vert =1$,
		$
		\mathbb{E}\left\Vert \frac 1{\sqrt T}\sum_{t=1}^{T}\mathbf{F}_{k,t}\left(
		\mathbf{v}^{\mathbf{\top }}\mathbf{E}_{k,t}\mathbf{w}\right) \right\Vert
		_{F}^{2}\leq c,
		$ for some $c<\infty $ independent of $k$ and $T$; (ii) letting $\mathbf{\zeta
		}_{i_{1},...,i_{K}}=\text{\upshape Vec}\left( T^{-1/2}\sum_{t=1}^{T}\mathcal{%
			F}_{t}e_{t,i_{1},...,i_{K}}\right) $, then for all $1\leq k\leq K$, $1\leq
		i_{k}\leq p_{k}$ and all $p_{k}$%
		\begin{align*}
			&\left\Vert \sum_{h=1,h\neq k}^{K}\sum_{i_{h}^{\prime }=1}^{p_{h}}\mathbb{E}%
			\left( \mathbf{\zeta }_{i_{1},...,i_{K}}\otimes \mathbf{\zeta }%
			_{i_{1}^{\prime },...,i_{k-1}^{\prime },i_{k},i_{k+1}^{\prime
				},...,i_{K}^{\prime }}\right) \right\Vert _{\max }\!\!\!\!\!\!\! \leq c, \\
			&\left\Vert \sum_{h=1,h\neq k}^{K}\sum_{i_{h}^{\prime
				}=1}^{p_{h}}\sum_{l=1,l\neq \ell }^{K}\sum_{j_{l}^{\prime
				}=1}^{p_{l}}Cov\left( \mathbf{\zeta }_{i_{1},...,i_{K}}\otimes \mathbf{\zeta
			}_{j_{1},...,j_{K}},\mathbf{\zeta }_{i_{1}^{\prime },...,i_{k-1}^{\prime
				},i_{k},i_{k+1}^{\prime },...,i_{K}^{\prime }}\otimes \mathbf{\zeta }%
			_{j_{1}^{\prime },...,j_{\ell -1}^{\prime },j_{\ell },j_{\ell +1}^{\prime
				},...,j_{K}^{\prime }}\right) \right\Vert _{\max } \!\!\!\!\!\!\!\leq c,
		\end{align*}%
		for some $c<\infty $ independent of $k$, $p_{k}$, $p_{\ell }$ and $i_{k}$, $%
		j_{\ell }$.}
\end{assumption}

Assumption \ref{as-4} is automatically satisfied if the idiosyncratic
component and the factor tensor processes are mutually independent, but it
allows for possible correlation between the common factors and the
idiosyncratic component. Indeed, given that $\mathbf{v}^{\mathbf{\top }}%
\mathbf{E}_{k,t}\mathbf{w}$ is a random variable with zero mean and bounded
variance (because of Assumption \ref{as-1}), part \textit{(i)} of Assumption %
\ref{as-4} states that that the correlation between $\left\{ \mathbf{F}%
_{k,t}\right\} $ and $\left\{ \mathbf{v}^{\mathbf{\top }}\mathbf{E}_{k,t}%
\mathbf{w}\right\} $ is also weak. By the same token, part \textit{(ii)}
controls the high-order dependence between the common factors and the
idiosyncratic components.

\begin{assumption}
	\label{as-5}\textit{It holds that: (i) for all $1\leq i\leq p_{k}$ and all $%
		1\leq k\leq K$, as $\min \left\{ T, p_{1},\ldots,p_{K}\right\} \rightarrow
		\infty$
		\begin{equation*}
			\frac{1}{\sqrt{Tp_{-k}}}\sum_{t=1}^{T}\mathbf{F}_{k,t}\mathbf{B}_{k}^{%
				\mathbf{\top }}\mathbf{e}_{k,t,i\cdot }^{\mathbf{\top }}\overset{D}{%
				\rightarrow }\mathcal{N}\left( \mathbf{0},\mathbf{V}_{ki}\right) ,
		\end{equation*}%
		where
		$
		\mathbf{V}_{ki}=\lim_{\min \left\{ T,p_{1},\ldots,p_{K}\right\} \rightarrow
			\infty }\frac{1}{Tp_{-k}}\sum_{t,s=1}^{T}\mathbb{E}\left( \mathbf{F}_{k,t}%
		\mathbf{B}_{k}^{\mathbf{\top }}\mathbf{e}_{k,t,i\cdot }^{\mathbf{\top }}%
		\mathbf{e}_{k,s,i\cdot }\mathbf{B}_{k}\mathbf{F}_{k,s}^{\mathbf{\top }%
		}\right)
		$
		is a positive definite $r_{k}\times r_{k}$ matrix with $\left\Vert \mathbf{V}%
		_{ki}\right\Vert _{F}<\infty $; (ii) for all $1\le t\le T$ and all $1\le
		k\le K$, as $\min \left\{ T, p_{1},\ldots,p_{K}\right\} \rightarrow \infty$
		\begin{equation*}
			\frac 1{\sqrt p} \left(\mathbf{B}_k\otimes \mathbf{A}_k\right)^\top \text{%
				\upshape Vec}\left(\mathbf{E}_{k,t}\right)\overset{D}{\rightarrow }\mathcal{N%
			}\left( \mathbf{0},\mathbf{W}_{kt}\right),
		\end{equation*}
		where
		$
		\mathbf{W}_{kt} = \lim_{\min \left\{ T,p_{1},\ldots,p_{K}\right\}
			\rightarrow \infty }\frac 1p \left(\mathbf{B}_k\otimes \mathbf{A}%
		_k\right)^\top \mathbb{E}\left\{ \text{\upshape Vec}\left(\mathbf{E}%
		_{k,t}\right) \text{\upshape Vec}\left(\mathbf{E}_{k,t}\right)^\top \right\}
		\left(\mathbf{B}_k\otimes \mathbf{A}_k\right)
		$
		is a positive definite $r\times r$ matrix with $\left\Vert \mathbf{W}%
		_{kt}\right\Vert _{F}<\infty $. }
\end{assumption}

Assumption \ref{as-5} is required only to derive the limiting laws of the
estimated loadings and common factors, without being required to derive
rates. In principle, it could be derived from more primitive assumptions on
the dependence structure of the data, as, e.g., strong mixing.

We now study the asymptotic properties of the \textquotedblleft
initial\textquotedblright\ estimator of the loadings, $\widehat{\mathbf{A}}%
_{k}$, defined in (\ref{eq:Ahatkdef}), of the projection based one $%
\widetilde{\mathbf{A}}_{k}$, defined in (\ref{eq:Atildekdef}), and of the
estimated factor tensor defined in \eqref{eq:Ftildeest}. We also study the
consistency of the estimators of the numbers of factors defined in Section %
\ref{sec:2.3}.

\subsection{Asymptotic properties of the initial estimator of the loadings\label{loadings-1}}

The following theorem establishes the convergence rate of the initial
projection estimators discussed in Section \ref{sec:2.1.1} and defined in %
\eqref{eq:Ahatkdef}. Henceforth, we use the notation%
\begin{equation*}
	w_{k}:=\frac{1}{p_{k}^{2}}+\frac{1}{Tp_{-k}}.
\end{equation*}

\begin{theorem}
	\label{th:1}{\it We assume that Assumptions \ref{as-1}-\ref{as-4} are satisfied.
		Then, for any given $1\leq k\leq K$, there exists an $r_{k}\times r_{k}$
		invertible matrix $\widehat{\mathbf{H}}_{k}$ such that, as $\min \left\{
		T,p_{1},...,p_{K}\right\} \rightarrow \infty $, it holds that
		$\widehat{\mathbf{H}}_{k}\widehat{\mathbf{H}}_{k}^{\mathbf{\top }}\overset{P}{%
			\rightarrow }\mathbf{I}_{r_{k}}$, and
		$
		\frac{1}{p_{k}}\left\Vert \widehat{\mathbf{A}}_{k}-\mathbf{A}_{k}\widehat{%
			\mathbf{H}}_{k}\right\Vert _{F}^{2}=O_{P}\left( w_{k}\right) .
		$
	}
\end{theorem}

%\textcolor{red}{The fact that $\widehat{\Hb}_{k}$ is in the limit an orthogonal matrix is a consequence of orthonormality of the columns of the loadings.}

The theoretical convergence rates for the initial estimators incorporates
the results for matrix-valued case, studied in \cite{Yu2021Projected}, as a
special case. We note that essentially the same result is obtained when
considering a vector factor model, and the result in the theorem can be
contrasted with Theorem 2 in \citet{bai2003inferential}, where results are
reported for the unit-specific estimates of the loadings - in such a case, $%
p_{k}$ is the only cross-sectional dimension (denoted as $N$ in %
\cite{bai2003inferential}), and the other dimensions $p_{-k}$ are equal
to $1$.

The next theorem presents the asymptotic distributions of the initial
estimators of the loadings, and again it can be read in conjunction with
Theorem 2 in \citet{bai2003inferential}.

\begin{theorem}
	\label{th:7}{\it We assume that Assumptions \ref{as-1}-\ref{as-5} are satisfied.
		Then, for any given $1\leq i\leq p_{k}$ and $1\leq k\leq K$, as $\min \left\{
		T,p_{1},...,p_{K}\right\} \rightarrow \infty $
		
		(i) if $Tp_{-k}=o\left( p_{k}^{2}\right)$, then it holds that%
		\begin{equation}
			\sqrt{Tp_{-k}}\left( \widehat{\mathbf{A}}_{k,i\cdot }^{\mathbf{\top }}-%
			\widehat{\mathbf{H}}_{k}^{\mathbf{\top }}\mathbf{A}_{k,i\cdot }^{\mathbf{%
					\top }}\right) \overset{D}{\rightarrow }\mathcal{N}\left( \mathbf 0,\mathbf{\Lambda }%
			_{k}^{-1}\mathbf{\Gamma }_{k}^{\top }\mathbf{V}_{ki}\mathbf{\Gamma }_{k}%
			\mathbf{\Lambda }_{k}^{-1}\right) ,  \label{th7-part1}
		\end{equation}%
		where $\mathbf{\Gamma }_{k}$ and $\mathbf{\Lambda }_{k}$ are defined in
		Assumption \ref{as-1}\textit{(iii)} and $\mathbf{V}_{ki}$ is defined in
		Assumption \ref{as-5};
		
		(ii) if
		$Tp_{-k}\gtrsim p_{k}^{2}$,
		then it holds that%
		$\left\Vert \widehat{\mathbf{A}}_{k,i\cdot }^{\mathbf{\top }}-\widehat{%
			\mathbf{H}}_{k}^{\mathbf{\top }}\mathbf{A}_{k,i\cdot }^{\mathbf{\top }%
		}\right\Vert _{F}=O_{P}\left( p_k^{-1}\right) .  $
	}
\end{theorem}

\subsection{Asymptotic properties of the projection-based estimators\label%
	{loadings-2}}

We now turn to studying the projection-based estimator defined in (\ref%
{eq:Atildekdef}). As can be expected, the properties of $\widetilde{\mathbf{A%
}}_{k}$ are bound to depend on the properties of the initial projection
matrices.

We begin by stating a sufficient condition on the rates of convergence of
the initial estimators which guarantees the consistency of the
projection-based estimator.

\noindent \textbf{Sufficient Condition} \textit{Consider a generic initial
	estimators $\left\{ \widehat{\mathbf{A}}_{k}^{(0)},\,1\leq k\leq K\right\} $%
	, and define the corresponding $\widehat{\mathbf{B}}_{k}^{(0)}:=\otimes
	_{j\in \lbrack K]\backslash \{k\}}\widehat{\mathbf{A}}_{j}^{(0)}$. For any
	given $1\leq k\leq K$, there exist $r_{-k}\times r_{-k}$ matrices $\widehat{%
		\mathbf{H}}_{-k}$ such that as $\min \{T,p_{1},\ldots ,p_{k}\}\rightarrow
	\infty $, $\widehat{\mathbf{H}}_{-k}\widehat{\mathbf{H}}_{-k}^{\top }\overset%
	{p}{\rightarrow }\mathbf{I}_{r_{-k}}$ and%
	\begin{equation}
		\begin{aligned} &\text{(a) }
			\frac{1}{p_{-k}}\left\|\widehat{\Bb}_k^{(0)}-\Bb_k\widehat{\Hb}_{-k}\right%
			\|_{F}^{2}=O_{P}\left(w_{-k}\right),\\ &\text{(b) }
			\frac{1}{p_{k}}\left\|\frac{1}{T p_{-k}} \sum_{s=1}^{T}
			\mathbf{E}_{k,s}\left(\widehat{\Bb}_k^{(0)}-\Bb_k
			\widehat{\mathbf{H}}_{-k}\right)
			\mathbf{F}_{k,s}^{\top}\right\|_{F}^{2}=O_{P}\left(m_{-k}\right),
		\end{aligned}  \label{equ:3.1}
	\end{equation}%
	with $w_{-k},m_{-k}\rightarrow 0$.}

In essence, the sufficient condition presented above requires the initial
estimator to have an estimation error that vanishes \textquotedblleft
sufficiently fast\textquotedblright . We will see that any estimator which
satisfies (\ref{equ:3.1}) results in a projection-based estimate, $%
\widetilde{\mathbf{A}}_{k}$, which consistently estimates the space spanned
by $\mathbf{A}_{k}$. In particular, as shown in Lemma \ref{lem:1} below, the
PCA estimator defined in \eqref{eq:Ahatkdef} satisfies such a sufficient
condition.

The following theorem presents the rate of convergence of the projected
estimators $\widetilde{\mathbf{A}}_{k}$ defined in \eqref{eq:Atildekdef}
when computed using any set of initial estimators satisfying the sufficient
condition in (\ref{equ:3.1}).

Henceforth, we use the notation%
\begin{equation*}
	\widetilde{w}_{k}:=\frac{1}{Tp_{-k}}+\frac{1}{p^{2}}+w_{-k}\left( \frac{1}{%
		Tp_{k}}+\frac{1}{p_{k}^{2}}\right) +m_{-k}.
\end{equation*}

\begin{theorem}
	\label{th:2}{\it We assume that Assumptions \ref{as-1}-\ref{as-4} are satisfied,
		and that the projected estimators $\left\{ \widetilde{\mathbf{A}}_{k},1\leq
		k\leq K\right\} $ are computed using an initial estimator satisying (\ref%
		{equ:3.1}). Then, for any given $1\leq k\leq K$, there exists an $r_{k}\times r_{k}$
		invertible matrix $\widetilde{\mathbf{H}}_{k}$ such that, as $\min \left\{
		T,p_{1},...,p_{K}\right\} \rightarrow \infty $, it holds that
		$\widetilde{\mathbf{H}}_{k}^{\mathbf{\top }}\widetilde{\mathbf{H}}_{k}\overset%
		{P}{\rightarrow }\mathbf{I}_{r_{k}}$,
		and
		${p_{k}}^{-1}\left\Vert \widetilde{\mathbf{A}}_{k}-\mathbf{A}_{k}%
		\widetilde{\mathbf{H}}_{k}\right\Vert _{F}^{2}=O_{P}\left( \widetilde{w}%
		_{k}\right). $%.
	}
\end{theorem}

Theorem \ref{th:2} illustrates how the convergence rates of the projection
estimators depend on the convergence rates of the initial estimators, and it
is a general and rather abstract result. In order to make it more
practically useful, we now discuss the impact of using the PCA estimator$\
\widehat{\mathbf{A}}_{k}$ on the properties of $\widetilde{\mathbf{A}}_{k}$.

We begin by showing that, when using $\left\{ \widehat{\mathbf{A}}_{k},1\leq
k\leq K\right\} $ as the initial estimators, these satisfy the sufficient
condition (\ref{equ:3.1}).

\begin{lemma}
	\label{lem:1}{\it We assume that Assumptions \ref{as-1}-\ref{as-4} are satisfied.
		Then the initial estimators $\left\{ \widehat{\mathbf{A}}_{k},1\leq k\leq
		K\right\} $ defined in (\ref{eq:Ahatkdef}) satisfy equation (\ref{equ:3.1})
		with $\widehat{\mathbf{H}}_{-k}=\otimes _{j\in \lbrack K]\backslash \{k\}}\widehat{\mathbf{H}%
		}_{k}$ and
		\begin{eqnarray*}
			w_{-k} =\sum_{j=1,j\neq k}^{K}\left( \frac{1}{p_{j}^{2}}+\frac{1}{Tp_{-j}}%
			\right) , \quad m_{-k} =\sum_{j=1,j\neq k}^{K}\left( \frac{1}{Tp_{j}^{2}}+\frac{1}{%
				T^{2}p_{-j}^{2}}\right) .
		\end{eqnarray*}
	}
\end{lemma}

As a consequence of Lemma \ref{lem:1}, we can prove the following version of
Theorem \ref{th:2}.

\begin{corollary}
	\label{cor:1}{\it We assume that Assumptions \ref{as-1}-\ref{as-4} are satisfied,
		and that the and that the projected estimators $\left\{ \widetilde{\mathbf{A}%
		}_{k},1\leq k\leq K\right\} $ are computed using $\left\{ \widehat{\mathbf{A}%
		}_{k},1\leq k\leq K\right\} $ defined in (\ref{eq:Ahatkdef}). Then, $\left\{
		\widetilde{\mathbf{A}}_{k},1\leq k\leq K\right\} $ satisfy Theorem \ref{th:2}
		with%
		\begin{equation*}
			\widetilde{w}_{k}=\frac{1}{Tp_{-k}}+\frac{1}{p^{2}}+\sum_{j=1,j\neq
				k}^{K}\left( \frac{1}{Tp_{j}^{2}}+\frac{1}{T^{2}p_{-j}^{2}}+\frac{1}{%
				p_{k}^{2}p_{j}^{4}}\right) .
		\end{equation*}
	}
\end{corollary}

Corollary \ref{cor:1} states that if we choose - as initial estimators - the
PCA based $\widehat{\mathbf{A}}_{k}$ defined in (\ref{eq:Ahatkdef}), then
the one-step projection estimator $\widetilde{\mathbf{A}}_{k}$ might achieve
faster convergence rates than $\widehat{\mathbf{A}}_{k}$. In particular, by
comparing the convergence rates of $\widehat{\mathbf{A}}_{k}$ in Theorem \ref%
{th:1}, and the ones of $\widetilde{\mathbf{A}}_{k}$ in Corollary \ref{cor:1}%
, we see that the latter might improve on the former, especially when $p_{k}$
is \textquotedblleft small\textquotedblright .

Finally, the following theorem presents the asymptotic distributions of the
projected estimators of the loadings.

\begin{theorem}
	\label{th:5}{\it We assume that Assumptions \ref{as-1}-\ref{as-5} are satisfied,
		and that the and that the projected estimators $\left\{ \widetilde{\mathbf{A}%
		}_{k},1\leq k\leq K\right\} $ are computed using $\left\{ \widehat{\mathbf{A}%
		}_{k},1\leq k\leq K\right\} $ defined in (\ref{eq:Ahatkdef}).
		Then, for any given $1\leq i\leq p_{k}$ and $1\le k\le K$, as $\min \left\{
		T,p_{1},...,p_{K}\right\} \rightarrow \infty $
		
		(i) if $K\le 2$ and
		$Tp_{-k}=o\left( \min \left\{
		p^{2},T^{2}p_{-j}^{2},Tp_{j}^{2},p_{k}^{2}p_{j}^{2}\right\} \right)$, for  all $j\ne k$, then it holds that%
		\begin{equation}
			\sqrt{Tp_{-k}}\left( \widetilde{\mathbf{A}}_{k,i\cdot }^{\mathbf{\top }}-%
			\widetilde{\mathbf{H}}_{k}^{\mathbf{\top }}\mathbf{A}_{k,i\cdot }^{\mathbf{%
					\top }}\right) \overset{D}{\rightarrow }\mathcal{N}\left( \mathbf{0},\mathbf{%
				\Lambda }_{k}^{-1}\mathbf{\Gamma }_{k}^{\mathbf{\top }}\mathbf{V}_{ki}%
			\mathbf{\Gamma }_{k}\mathbf{\Lambda }_{k}^{-1}\right) ,  \label{th5-1}
		\end{equation}%
		with $\mathbf{\Gamma }_{k}$ and $\mathbf{\Lambda }_{k}$ defined in
		Assumption \ref{as-1}\textit{(iii)} and $\mathbf{V}_{ki}$ defined in
		Assumption \ref{as-5};
		
		(ii) if $K\le 2$ and $Tp_{-k}\gtrsim \min \left\{
		p^{2},T^{2}p_{-j}^{2},Tp_{j}^{2},p_{k}^{2}p_{j}^{2}\right\}$,
		for all $j\ne k$, or if $K\ge 3$, then it holds that $\left\Vert \widetilde{\mathbf{A}}_{k,i\cdot }^{\mathbf{\top }}-\widetilde{%
			\mathbf{H}}_{k}^{\mathbf{\top }}\mathbf{A}_{k,i\cdot }^{\mathbf{\top }%
		}\right\Vert =O_{P}\left( \frac{1}{p}+\sum_{j=1,j\neq k}^{K}\left( \frac{1}{%
			\sqrt{T}p_{j}}+\frac{1}{Tp_{-j}}+\frac{1}{p_{k}p_{j}^{2}}\right) \right) .
		$
	}
\end{theorem}

Comparing part (i) of the theorem with part (i) of Theorem \ref{th:7} (which
holds under the same assumptions and restrictions on the relative rate of
divergence of the dimensions $T$, $p_{1}$,..., $p_{K}$ as they pass to
infinity), it emerges that the projected estimator $\widetilde{\mathbf{A}}%
_{k,i\cdot }$ has the same consistency rate of the initial estimator $%
\widehat{\mathbf{A}}_{k,i\cdot }$ and it is equally efficient. Indeed, due
to its iterative nature, $\widetilde{\mathbf{A}}_{k,i\cdot }$ behaves
similarly to a classical one-step estimator (see, e.g.,
\cite[Theorem
4.3]{LC06}). {We would like to emphasize that this result holds only for
	vector- or matrix-valued time series - i.e. when $K\leq 2$. Indeed, it is
	easy to see that if $K>2$, then the required constraint between the rates of
	divergence of $T$ and $p_{k}$, $1\leq k\leq K$, is never satisfied. For the
	case of higher-order tensors, i.e. }$K\geq 3${, comparing the results in
	part (ii) of Theorem \ref{th:7} with part (ii) of Theorem \ref{th:5}, it
	emerges that $\widetilde{\mathbf{A}}_{k,i\cdot }$ has a faster rate of
	convergence - this result can be understood by noting that, unlike the
	initial estimator $\widehat{\mathbf{A}}_{k,i\cdot }$, the projected
	estimator $\widetilde{\mathbf{A}}_{k,i\cdot }$ it is built by using all the
	information contained in all other modes. In this case however no asymptotic
	distribution can be derived under the present set of assumptions.}

However, if we consider  
%We conclude by considering the use of 
$\widehat{\mathbf{B}}_{k}^{\ast }$ as
initial estimators (see Section \ref{sec:2.1.2}) we can derive also asymptotic normality. Let
\begin{equation*}
	\widetilde{w}_{k}^{\ast }=\frac{1}{Tp_{-k}}+\frac{1}{p^{2}}+\frac{1}{%
		T^{2}p_{k}^{2}}.
\end{equation*}%
The following theorem summarizes the asymptotics (rates and limiting
distribution) of $\widetilde{\mathbf{A}}_{k}^{\ast }$.

\begin{theorem}
	\label{atilde}{\it We assume that Assumptions \ref{as-1}-\ref{as-4} are
		satisfied. Then, as $\min \{T,p_{1},\ldots ,p_{k}\}\rightarrow
		\infty $, it holds that, for all $1\leq k\leq K$ $
		\frac{1}{p_{k}}\left\Vert \widetilde{\mathbf{A}}_{k}^{\ast }-\mathbf{A}_{k}%
		\widetilde{\mathbf{H}}_{k}^{\ast }\right\Vert _{F}^{2}=O_{P}\left(
		\widetilde{w}_{k}^{\ast }\right),$
		where $\widetilde{\mathbf{H}}_{k}^{\ast }$ is defined in (\ref{h-tilde-star}%
		). Further, if Assumption \ref{as-5} is also satisfied, and if $%
		Tp_{-k}=o\left( \min \left\{ p^{2},T^{2}p_{k}^{2}\right\} \right) $, then it
		holds that $\sqrt{Tp_{-k}}\left( \widetilde{\mathbf{A}}_{k,i}^{\ast \top }-\widetilde{%
			\mathbf{H}}_{k}^{\ast \top }\mathbf{A}_{k,i}^{\top }\right) \overset{%
			\mathcal{D}}{\rightarrow }\mathcal{N}\left( \mathbf{0},\mathbf{\Lambda }%
		_{k}^{-1}\mathbf{\Gamma }_{k}\mathbf{V}_{k,i}\mathbf{\Gamma }_{k}^{\top }%
		\mathbf{\Lambda }_{k}^{-1}\right) , $
		where $\mathbf{\Lambda }_{k}$, $\mathbf{\Gamma }_{k}$ and $\mathbf{V}_{k,i}$
		are defined in Theorem \ref{th:5}.}
\end{theorem}

In order to appreciate the differences
in the convergence rates between $\widetilde{\mathbf{A}}_{k}$ and $%
\widetilde{\mathbf{A}}_{k}^{\ast }$, let us assume for simplicity that that $%
p_{1}\asymp \cdots \asymp p_{K}\asymp \overline{p}$. When $K=2$, then $%
\widetilde{w}_{k}=O_{P}((\overline{p}^{-4}+(T\overline{p}^{-1})$ and $%
\widetilde{w}_{k}^{\ast }=\widetilde{w}_{k}$, since $\widehat{\mathbf{B}}%
_{k}^{\ast }=\widehat{\mathbf{B}}_{k}$ in this case. Similarly, when $K=3$, $%
\widetilde{w}_{k}=O_{P}\left( (\overline{p}^{-6}+(T\overline{p}%
^{2})^{-1}\right) $ and the equality $\widetilde{w}_{k}^{\ast }=\widetilde{w}%
_{k}$ still holds. When $K>3$, $\widetilde{w}_{k}=O_{P}\left( (\overline{p}%
^{-6}+(T\overline{p}^{2})^{-1}\right) $, and $\widetilde{w}_{k}^{\ast
}=O_{P}\left( (\overline{p}^{-2K}+(T\overline{p}^{-2}+(T\overline{p}%
^{(K-1)})^{-1}\right) $. This shows that when $K$ is large (and, also, when $%
T$ is large), the estimator $\widetilde{\mathbf{A}}_{k}^{\ast }$ will have
faster converge rate than $\widetilde{\mathbf{A}}_{k}$. Note however that,
in order to obtain $\widehat{\mathbf{B}}_{k}^{\ast }$, one needs to perform
the eigen-decomposition of a very large scale, $p_{-k}\times p_{-k}$,
covariance matrix, with computational complexity $O\left( p_{-k}^{3}\right) $%
. Conversely, $\widetilde{\mathbf{A}}_{k}$ requires calculating, as initial
estimators, the loading matrices $\widehat{\mathbf{A}}_{k}$ separately; this
requires performing the eigen-decompositions of $p_{k}\times p_{k}$
covariance matrices, and then calculating Kronecker product of $(K-1)$ small
scale matrices of dimension $p_{k}\times r_{k}$, which has (lower)
computational complexity $O\left( \sum_{j\neq k}p_{j}^{3}+p_{-k}\times
r_{-k}\right) $.

Finally, we notice that in the case of strong factors as considered in this paper, 
\citet{chen2022rank} have a (squared) consistency rate which is $\min(T^2,p^2,Tp_{-k})$ which is comparable to the rate $\min(T^2p_k^2,p^2,Tp_{-k})$ derived in Theorem \ref{atilde}.

\subsection{Asymptotic properties of the factors and common component\label{factors}}

As long as the loading matrices are determined, the factor $\cF_{t}$ can be
estimated easily by
$\widetilde{\cF}_t=\frac 1p\cF_t\times_{k=1}^K\widetilde{\Ab}_k^{\top},$
where the projected estimators $\left\{ \widetilde{\mathbf{A}}_{k},1\leq
k\leq K\right\} $ are computed using $\left\{ \widehat{\mathbf{A}}_{k},1\leq
k\leq K\right\} $ defined in (\ref{eq:Ahatkdef}).
%The common component is then given by
%$$
%\widetilde{\cS}_t=%
%\widetilde{\cF}_t\times_{k=1}^K\widetilde{\Ab}_k.
%$$
The next theorem provides the consistency of the asymptotic properties of
the estimated factors.

\begin{theorem}\label{th:6}{\it
		%Under Assumptions \ref{asmp:2} to \ref{asmp:5}, for any given $1\le t\le T$, $k=1,\ldots, K$, and  $i_k =1,\ldots, p_k$, as $\min \left\{T, p_{1},\ldots,p_K\right\} \rightarrow \infty$,
		We assume that Assumptions \ref{as-1}-\ref{as-5} are satisfied.
		Then, for any given $1\leq i\leq p_{k}$ and $1\le k\le K$, as $\min \left\{
		T,p_{1},...,p_{K}\right\} \rightarrow \infty$

		(i) if  $p=o\left(\min\left\{Tp_{-k}^2,Tp_k p_j^2,Tp_{-k} p_j^2,p_k^2p_j^4\right\}\right)$, for all $j\ne k$,
		then it holds that
		\[
		\sqrt p\left(\text{\upshape Vec}\left(\widetilde{\mathbf F}_{t,k}\right)-\left(\widetilde {\mathbf H}_{-k} \otimes \widetilde {\mathbf H}_{k}\right)^{-1} \text{\upshape Vec}\left({\mathbf F}_{t,k}\right)\right)\overset{D}{\rightarrow }\mathcal{N}\left( \mathbf{0},\left(\mathbf H_{-k}\otimes\mathbf H_{k} \right)^\top\mathbf{W}_{kt}\left(\mathbf H_{-k}\otimes\mathbf H_{k} \right)\right) ,  \label{th5-1}
		\]
		with
		$\mathbf H_{k} = {P\text{-}\lim}_{\min \left\{
			T,p_{1},...,p_{K}\right\} \rightarrow \infty} \widetilde{\mathbf H}_{k}$ and $\mathbf H_{-k}=\otimes_{j\in[K]\slash \{k\}} \mathbf H_{j}$, and $\mathbf{W}_{kt}$ defined in
		Assumption \ref{as-5};
		
		(ii) if 	$p\gtrsim \min\left\{Tp_{-k}^2,Tp_k p_j^2,Tp_{-k} p_j^2,p_k^2p_j^4\right\}$,
		for all $j\ne k$,
		then it holds that
		\[
		\left\|\widetilde{\cF}_{t}- \cF_{t}\times_1\widetilde{\mathbf{H}}_{1}^{-1}\times_2\cdots\times_K\widetilde{\mathbf{H}}_{K}^{-1}\ \right\|=O_{P}\left(\sum_{k=1}^{K} \left(\frac{1}{\sqrt{T} p_{-k}}+\sum_{j\neq k}\left(\frac{1}{\sqrt{T p_{k}} p_{j}}+\frac{1}{\sqrt{T p_{-k}} p_{j}}+\frac{1}{p_{k} p_{j}^{2}}\right)\right)\right).
		\]
	}
\end{theorem}

The theorem states the asymptotic distribution of the estimated factor tensor, which holds for any finite tensor order $K$ and it can be compared with
Theorem 1 in \citet{bai2003inferential} which holds for $K=1$, i.e. for the vector factor model.
Two comments are worth making. First, notice that an explicit expression for the matrix $\mathbf H_{k}$ can be derived from the definition of
$\widetilde{\mathbf H}_{k}$ given in \eqref{eq:Htildedefinition} in \ref{proofs}. However, we prefer not to write those expressions explicitly to avoid introducing further notation. Second, a similar result can be derived for the factor tensor estimators
$\widehat{\cF}_t=\frac 1p\cF_t\times_{k=1}^K\widehat{\Ab}_k^{\top}$ and $\widetilde{\cF}^*_t=\frac 1p\cF_t\times_{k=1}^K\widetilde{\Ab}_k^{*\top}$,  which are estimated
using the initial estimator of the loadings $\left\{ \widehat{\mathbf{A}}_{k},1\leq
k\leq K\right\} $ or the estimator of the loadings $\left\{ \widetilde{\mathbf{A}}_{k}^*,1\leq
k\leq K\right\} $, respectively. These results are omitted for
brevity.

Finally, the estimated common component is given by
$\widetilde{\cS}_t=\widetilde{\cF}_t\times_{k=1}^K\widetilde{\Ab}_k$.
The next theorem provides consistency of the estimated common component.

\begin{theorem}\label{th:6bis}{\it
		%Under Assumptions \ref{asmp:2} to \ref{asmp:5}, for any given $1\le t\le T$, $k=1,\ldots, K$, and  $i_k =1,\ldots, p_k$, as $\min \left\{T, p_{1},\ldots,p_K\right\} \rightarrow \infty$,
		We assume that Assumptions \ref{as-1}-\ref{as-5} are satisfied.
		Then, for any given $1\leq i_k\leq p_{k}$, $1\le k\le K$, and $1\le t\le T$, as $\min \left\{
		T,p_{1},...,p_{K}\right\} \rightarrow \infty $
		$$
		\left|\widetilde{S}_{t, i_1\cdots i_K}-S_{t, i_1\cdots i_K}\right|=O_{P}\left(\frac{1}{\sqrt{p}}+\sum_{k=1}^{K}\left(\frac{1}{\sqrt{T} p_{k}}+\frac{1}{\sqrt{T p_{-k}}}+\sum_{j\neq k} \frac{1}{p_{k} p_{j}^{2}}\right)\right).
		$$
	}
\end{theorem}

The consistency rate depends of the rates for $\{\widetilde{\Ab}_k,\,1\le k \le K\}$ in Theorem \ref{th:5} and for $\widetilde{\cF}_t$ in Theorem \ref{th:6}.
Two other estimators of the common component can be considered:
$\widehat{\cS}_t=\widehat{\cF}_t\times_{k=1}^K\widehat{\Ab}_k$ and $\widetilde{\cS}_t^*=\widetilde{\cF}_t^*\times_{k=1}^K\widetilde{\Ab}_k^*$,
and the proof of their consistency is straightforward and thus it is omitted.  The proof of asymptotic normality for $\widetilde{\cS}_t$, $\widehat{\cS}_t$, and $\widetilde{\cS}_t^*$ is also omitted for brevity. Indeed, it can be easily proved in a similar way as in %
\citet{bai2003inferential}. Notice that, while asymptotic normality of $\widehat{\cS}_t$ and $\widetilde{\cS}_t^*$ would hold for any finite tensor order $K$, for
$\widetilde{\cS}_t$ it would hold only if $K\le 2$, since only in that case
we can prove asymptotic normality of the corresponding estimated loadings.

\subsection{Asymptotic properties of the estimators for factor numbers\label%
	{nfactors}}

In this section, we establish the consistency of the proposed estimators of
factor numbers. We first focus on the estimators based on the simple
mode-wise sample covariance matrix $\widehat{\Mb}_{k}$. The following
theorem shows that the $\hat{r}_{k}^{\text{IE-ER}}$ defined in (\ref%
{equ:2}) are consistent.

\begin{theorem}
	\label{th:4}{\it We assume that Assumptions \ref{as-1}-\ref{as-4} are satisfied,
		and that $r_{k}\geq 1$, for all $1\leq k\leq K$, with $r_{\max }\geq
		\max_{1\leq k\leq K}r_{k}$. Then, for all $1\leq k\leq K$, as $\min \left\{ T,p_{1},...,p_{K}\right\}
		\rightarrow \infty $, it holds that
		$\mathbb{P}\left( \hat{r}_{k}^{\,\text{\upshape IE-ER}}=r_{k}\right) \to1$. }
\end{theorem}

In Algorithm \ref{alg2}, we choose to calculate eigenvalue-ratio of the
projection version $\widetilde{\Mb}_{k}$ rather than the initial version $%
\widehat{\Mb}_{k}$, as $\widetilde{\Mb}_{k}$ is more accurate. The
consistency of the estimator outputted by the iterative algorithm is
guaranteed by the following theorem.

\begin{theorem}
	\label{th:4-bis}{\it We assume that Assumptions \ref{as-1}-\ref{as-4} are
		satisfied, and that $r_{k}\geq 1$, for all $1\leq k\leq K$, with $r_{\max
		}\geq \max_{1\leq k\leq K}r_{k}$. If, for all $j\neq k$, it holds that
		$r_{j}\leq \hat{r}_{j}^{\,\left( s-1\right) }\leq r_{\max }$,
		for some $s$ in Algorithm \ref{alg2}, then, for all $1\leq k\leq K$, as $\min \left\{
		T,p_{1},...,p_{K}\right\} \rightarrow \infty $, it holds that
		$\mathbb{P}\left( \hat{r}_{k}^{\,\left( s\right) }=r_{k}\right) \to 1$ thereby, $\mathbb{P}\left( \hat{r}_{k}^{\,\text{\upshape PE-ER}}=r_{k}\right) \to 1.$ 
	}
\end{theorem}

The theorem states that as long as we choose $r_k^{(0)}$ larger than real $%
r_k$ at the beginning, the iterative method will yield consistent estimators
of the number of factors. The algorithm is computationally very fast because
it has a large probability to stop within few steps.

\section{Real data analysis\label{empirics}}

We analyzed multi-category import-export network data analyzed in \cite%
{Chen2021Factor}. The data contains total monthly imports and exports for 15
commodity categories between 22 European and American countries from January
2014 to December 2019, with missing values for the export from any country
to itself. More detailed information on data, countries and commodity
categories can be found in \cite{Chen2021Factor}. Similar to \cite%
{Chen2021Factor}, we set the volume of each country's exports to zero and
take a three-month rolling average of the data to eliminate the impact of
unusual large transactions or shipping delays. At this point, this data can
be viewed as a $22 \times 22 \times 15 \times 70$ fourth-order tensor. Each
element $x_{i,j,k,t}$ is the three-month average total export volume of
country $i$ to country $j$ at time $t$ for category $k$ goods. We will model
the multi-category import-export network data according to the following
model. Let $\cX_t=\cF_t\times_1\Ab_1\times_2\Ab_2\times_3\Ab_3+\cE_t$, where
$\cX_t\in \RR^{p_1\times p_2 \times p_3}$ is the observed tensor at time $t$%
, $\cF_t \in \RR^{r_1\times r_2\times r_3}$ is the factor tensor, $\{\Ab_k\in%
\RR^{p_k\times r_k},~k=1,2,3\}$ are the loading matrices, for $t=1,2,\cdots
,70$, $p_1=p_2=22,~p_3=15$, $\{r_k,k=1,2,3\}$ to be determined.

We first use a rolling-validation procedure as in \cite{wang2019factor} to
compare the methods mentioned in Section \ref{sec:4.2}. To implement the
rolling-validation procedure, we add additional November and December 2013
data before making the three-month moving average and normalize the final
data. For each year $t$ from 2017 to 2019, we repeatedly use the $n$
(bandwidth) year observations prior to $t$ to estimate the loading matrices
and then used loadings to estimate the observations and the corresponding
residuals for the 12 months of the year. Specifically, let $\cX_{t}^{i}$ and
$\widehat\cX_{t}^{i}$ be the observed and estimated import-export tensor of
month $i$ in year $t$, and further define the mean squared error as:
\begin{equation*}
	\mathrm{MSE}_{t}=\frac{1}{12 \times 22 \times 22 \times 15}
	\sum_{i=1}^{12}\left\|\widehat\cX_{t}^{i}-\cX_{t}^{i}\right\|_{F}^{2}.
\end{equation*}
%as the mean squared error.

\begin{table}[!h]
	\caption{ The averaged MSE of rolling validation for the three-month moving
		average data. $12n$ is the sample size of the training set. $r_1=r_2=r_3=r$
		is the number of factors. ``PE": projection estimation
		method.``IE": initial estimation method. ``TOPUP": Time series Outer-Product
		Unfolding Procedure with $h_0=1$. ``TIPUP": Time series Inner-Product
		Unfolding Procedure with $h_0=1$. ``iTOPUP": Iterative Time series
		Outer-Product Unfolding Procedure with $h_0=1$. ``iTIPUP": Iterative Time
		series Inner-Product Unfolding Procedure with $h_0=1$. ``IPmoPCA": iterative
		projected mode-wise PCA estimation.}
	
	\label{tab:4}\renewcommand{\arraystretch}{0.5} \centering
	\begin{tabular}{ccccccccc}
		\toprule[2pt] $n$ & $r$ & PE & IE & IPmoPCA & TIPUP & iTIPUP & TOPUP & iTOPUP \\
		\midrule 1 & 3 & 0.717165 & 0.719191 & 0.717929 & 0.719222 & 0.717905 &
		0.719204 & 0.717891 \\
		2 & 3 & 0.358891 & 0.359821 & 0.359287 & 0.359839 & 0.359309 & 0.359794 &
		0.359284 \\
		3 & 3 & 0.239331 & 0.239936 & 0.239582 & 0.239945 & 0.239595 & 0.239909 &
		0.239571 \\
		1 & 4 & 0.711102 & 0.712812 & 0.711038 & 0.712861 & 0.711067 & 0.712864 &
		0.711088 \\
		2 & 4 & 0.355766 & 0.356680 & 0.355675 & 0.356697 & 0.355685 & 0.356598 &
		0.355629 \\
		3 & 4 & 0.237183 & 0.237880 & 0.237054 & 0.237888 & 0.237066 & 0.237798 &
		0.237037 \\
		1 & 5 & 0.697179 & 0.699275 & 0.693501 & 0.698985 & 0.693601 & 0.698836 &
		0.693544 \\
		2 & 5 & 0.348440 & 0.349485 & 0.346284 & 0.349321 & 0.346368 & 0.349145 &
		0.346376 \\
		3 & 5 & 0.232304 & 0.233041 & 0.230282 & 0.232942 & 0.230299 & 0.232837 &
		0.230415 \\
		1 & 6 & 0.672986 & 0.675783 & 0.672979 & 0.675047 & 0.672926 & 0.673810 &
		0.672323 \\
		2 & 6 & 0.336668 & 0.338309 & 0.336630 & 0.338103 & 0.336587 & 0.337334 &
		0.335946 \\
		3 & 6 & 0.224361 & 0.225838 & 0.224242 & 0.225780 & 0.224165 & 0.225522 & 0.223770 \\
		\bottomrule[2pt] &  &  &  &  &  &  &  &
	\end{tabular}%
\end{table}

Table \ref{tab:4} compares the mean values of MSE of various estimation
methods for different combinations of bandwidth $n$ and factor number $%
r_1=r_2=r_3=r$. The estimation errors of PE, iTIOPUP and iTIPUP methods are
very close, and the estimation performance of projected methods is better
than that of non-prejected methods.

For processed $22 \times 22 \times 15 \times 70$ fourth-order tensor data,
our PE-ER method suggests $r_1=1,~r_2=3,~r_3=1$, while TCorTh suggests $%
r_1=3,~r_2=1,~r_3=3$, iTOP-ER suggests $r_1=3,~r_2=4,~r_3=1$, iTIP-ER
suggests $r_1=2,~r_2=2,~r_3=1$ and other methods all suggest $r_1=r_2=r_3=1$%
. For better illustration, we take $r_1=r_2=4,~r_3=6$ as same as in \cite%
{Chen2021Factor}.

\begin{table}[!h]
	\caption{ Estimated loading matrix $\Ab_3$ for category fiber.}
	\label{tab:5}\renewcommand{\arraystretch}{0.5}
	\resizebox{\textwidth}{!}{
		\centering
		\begin{tabular}{cccccccccccccccc}
			\toprule[2pt]
			&	Animal	&	Vegetable	&	Food	&	Mineral	&	Chemicals	&	Plastics	&	Leather	&	Wood	&	Textiles	&	Footwear	&	Stone	&	Metals	&	Machinery	&	Transport	&	Misc	\\
			\midrule
			1	&	0	&	-3	&	1	&	0	&	1	&	-2	&	0	&	2	&	-1	&	0	&	-1	&	0	&	\bf{-29}	&	0	&	-6	\\
			2	&	0	&	1	&	0	&	\bf{30}	&	0	&	0	&	0	&	3	&	-1	&	0	&	2	&	1	&	0	&	0	&	1	\\
			3	&	1	&	3	&	0	&	0	&	\bf{-29}	&	-1	&	0	&	0	&	1	&	0	&	2	&	-2	&	0	&	0	&	-8	\\
			4	&	1	&	0	&	4	&	0	&	-1	&	-3	&	0	&	2	&	0	&	0	&	-1	&	2	&	0	&	\bf{29}	&	2	\\
			5	&	6	&	5	&	5	&	0	&	1	&	19	&	1	&	5	&	5	&	1	&	-1	&	17	&	1	&	0	&	-9	\\
			6	&	2	&	2	&	4	&	-2	&	2	&	-2	&	1	&	4	&	2	&	1	&	\bf{29}	&	0	&	-1	&	0	&	0	\\
			\bottomrule[2pt]
	\end{tabular}}
\end{table}

\begin{table}[!h]
	\caption{ Estimated loading matrix $\Ab_1$ for export fiber.}
	\label{tab:6}\renewcommand{\arraystretch}{0.5}
	\resizebox{\textwidth}{!}{
		\centering
		\begin{tabular}{ccccccccccccccccccccccc}
			\toprule[2pt]
			&	BE	&	BU	&	CA	&	DK	&	FI	&	FR	&	DE	&	GR	&	HU	&	IS	&	IR	&	IT	&	MX	&	NO	&	PO	&	PT	&	ES	&	SE	&	CH	&	ER	&	US	&	UK	\\
			\midrule
			1	&	1	&	0	&	-1	&	0	&	0	&	-1	&	-4	&	0	&	0	&	0	&	2	&	-1	&	\bf{-30}	&	0	&	0	&	0	&	1	&	0	&	1	&	0	&	1	&	-1	\\
			2	&	-1	&	0	&	2	&	0	&	0	&	-3	&	1	&	0	&	-2	&	0	&	2	&	-3	&	-1	&	-1	&	-3	&	0	&	-1	&	-1	&	-2	&	-1	&	\bf{-29}	&	-1	\\
			3	&	-1	&	0	&	\bf{-30}	&	0	&	0	&	0	&	2	&	0	&	0	&	0	&	-1	&	0	&	0	&	-1	&	0	&	0	&	0	&	0	&	-1	&	0	&	-2	&	-1	\\
			4	&	6	&	0	&	1	&	2	&	1	&	8	&	24	&	0	&	1	&	0	&	8	&	6	&	-3	&	1	&	3	&	1	&	5	&	2	&	5	&	2	&	-1	&	5	\\
			\bottomrule[2pt]
		\end{tabular}
	}
\end{table}

\begin{table}[!h]
	\caption{ Estimated loading matrix $\Ab_2$ for import fiber.}
	\label{tab:7}\renewcommand{\arraystretch}{0.5}
	\resizebox{\textwidth}{!}{
		\centering
		\begin{tabular}{ccccccccccccccccccccccc}
			\toprule[2pt]
			&	BE	&	BU	&	CA	&	DK	&	FI	&	FR	&	DE	&	GR	&	HU	&	IS	&	IR	&	IT	&	MX	&	NO	&	PO	&	PT	&	ES	&	SE	&	CH	&	ER	&	US	&	UK	\\
			\midrule
			1	&	0	&	0	&	0	&	0	&	0	&	0	&	0	&	0	&	0	&	0	&	0	&	0	&	0	&	0	&	0	&	0	&	0	&	0	&	0	&	0	&	\bf{-30}	&	0	\\
			2	&	1	&	0	&	0	&	1	&	0	&	-1	&	-8	&	0	&	1	&	0	&	-1	&	2	&	\bf{-28}	&	0	&	2	&	0	&	1	&	1	&	4	&	0	&	0	&	-5	\\
			3	&	-3	&	-1	&	-1	&	-3	&	-2	&	-16	&	2	&	0	&	-5	&	0	&	1	&	-8	&	2	&	-2	&	-5	&	-2	&	-8	&	-5	&	1	&	-4	&	0	&	-19	\\
			4	&	9	&	0	&	0	&	0	&	0	&	1	&	20	&	1	&	-1	&	0	&	5	&	7	&	-2	&	0	&	3	&	0	&	2	&	-1	&	17	&	1	&	0	&	-4	\\
			\bottomrule[2pt]
	\end{tabular}}
\end{table}

Table \ref{tab:5} shows an estimator of $\Ab_3$ under tensor factor model
using PE procedure. For better interpretation, the loading matrix is rotated
using the varimax procedure and all numbers are multiplied by 30 and then
truncated to integers for clearer viewing.
%{\color{red} SOMETHING WRONG WITH NOTATION HERE?  }
Table \ref{tab:5} shows that there is a group structure of
these six factors, that each factor is called a ``condensed product group" in \cite{Chen2021Factor}.
Factor 1, 2, 3, 4 and 6 can be interpreted as Machinery \& Electrical
factor, Mineral Products factor, Chemicals \& Allied Industries factor,
Transportation factor and Stone \& Glass factor, because these factors are
mainly loaded on these commodity categories. Factor 5 can be viewed as a
mixing factor, with Plastics \& Rubbers and Metals as main load.

Table \ref{tab:6} and \ref{tab:7} show estimators of $\Ab_1$ and $\Ab_2$.
The loading matrices are also rotated via varimax procedure. \cite%
{Chen2021Factor} proposed that there are some virtual import hubs and export
hubs in the import-export data factor model. When exporting a commodity, the
exporting country first puts the commodity into the virtual export hub, then
the commodities are exported from the export hub to the virtual import hub,
and finally taken out of the import by the importing countries. Table \ref%
{tab:6} shows that Mexico, the United States of America and Canada heavily
load on virtual export hubs E1, E2 and E3, respectively. European countries
mainly load on export hub E4, and Germany occupies an important position.
This can be seen in Table \ref{tab:7} that the United States of America and
Mexico heavily load on virtual import hubs I1 and I2. European countries
mainly load on import hub I3 and I4, while I3 is mainly loaded by Western
European countries France and Britain and I4 is mainly loaded by Central
European countries Germany and Switzerland.
%%%%%%%%%%%%%%%%%%%%%%%%%%%%%%%%%%%%%%%%%%%%%%%%%%%
%%%%%%%%%%%%%%%%%%%%%%%%%%%%%%%%%%%%%%%%%%%%%%%%%%%
%%%%%%%%%%%%%%%%%%%%%%%%%%%%%%%%%%%%%%%%%%%%%%%%%%%
\section{Concluding remarks\label{conclusions}}

Tensor factor model is a powerful tool for dimension reduction of high-order
tensors and is drawing growing attention in the last few years. In this
paper, we propose a projection estimation method for Tucker-decomposition
based Tensor Factor Model (TFM). We also provide the least squares
interpretation of the iterative projection for TFM, which parallels to the
least squares interpretation of traditional PCA for vector factor models %
\citep{fan2013large} and of Projection Estimation for matrix factor models %
\citep{He2021Matrix}. We establish the theoretical properties for the
one-step iteration projection estimators, and faster convergence rates are
achieved by the projection technique compared with the naive tensor PCA
method.

\section{Technical details\label{proofs}}
%\section{Proofs of main results\label{proofs}}

We report all proofs assuming, for simplicity
and with no loss of generality, that $r_{k}=1$, $1\leq k\leq K$.

\begin{proof}[Proof of Theorem \ref{th:1}]
	We state some preliminary facts. Recall (\ref{dec-m-hat}), and define $%
	\widehat{\mathbf{\Lambda }}_{k}$ as the diagonal matrix containing the
	largest $r_{k}$ eigenvalues of $\widehat{\mathbf{M}}_{k}$, viz.%
	\begin{equation*}
		\widehat{\mathbf{\Lambda }}_{k}=\diag(\lambda _{1}(\widehat{\mathbf{M}}%
		_{k}),\ldots ,\lambda _{r_{k}}(\widehat{\mathbf{M}}_{k})).
	\end{equation*}%
	Also recall that by definition $\widehat{\mathbf{A}}_{k}=\sqrt{p_{k}}\,%
	\widehat{\mathbf{U}}_{k}$, where $\widehat{\mathbf{U}}_{k}$ has as columns
	the normalized eigenvectors corresponding to the $r_{k}$ largest eigenvalues
	of $\widehat{\mathbf{M}}_{k}$ - see (\ref{eq:Ahatkdef}). Then, by definition
%	\begin{equation*}
$		\widehat{\mathbf{A}}_{k}\widehat{\mathbf{\Lambda }}_{k}=\widehat{\mathbf{M}}%
		_{k}\widehat{\mathbf{A}}_{k}$.
%	\end{equation*}%
	Define now
	\begin{equation}
		\widehat{\mathbf{H}}_{k}=\frac{1}{Tp}\left( \sum_{t=1}^{T}\mathbf{F}%
		_{k,t}^{\top }\widehat{\mathbf{B}}_{k}^{\top }\widehat{\mathbf{B}}_{k}%
		\mathbf{F}_{k,t}^{\top }\right) \mathbf{A}_{k}^{\top }\widehat{\mathbf{A}}%
		_{k}\widehat{\boldsymbol{\Lambda }}_{k}^{-1}.  \label{eq:defHk}
	\end{equation}%
	Then it is easy to see that
	\begin{equation}
		\widehat{\mathbf{A}}_{k}-\mathbf{A}_{k}\widehat{\mathbf{H}}_{k}=(\mathcal{I}%
		\mathcal{I}+\mathcal{I}\mathcal{I}\mathcal{I}+\mathcal{I}\mathcal{V})%
		\widehat{\mathbf{A}}_{k}\widehat{\boldsymbol{\Lambda }}_{k}^{-1},
		\label{eq:AHATapp}
	\end{equation}%
	where $\mathcal{I}\mathcal{I}$, $\mathcal{I}\mathcal{I}\mathcal{I}$ and $%
	\mathcal{I}\mathcal{V}$ are defined in (\ref{dec-m-hat}). Recall that: by
	Assumption \ref{as-1}\textit{(ii)}%
%	\begin{equation*}
		$T^{-1}\sum_{t=1}^{T}\mathbf{F}_{k,t}\mathbf{F}_{k,t}^{\top }\overset{P}{%
			\rightarrow }\boldsymbol{\Sigma }_{k}$,
%	\end{equation*}%
	as$T\rightarrow \infty $; by Assumption \ref{as-2}\textit{(ii)} and the
	definition of $\mathbf{B}_{k}$ it readily follows that
%	\begin{equation*}
		$p_{-k}^{-1}\mathbf{B}_{k}^{\top }\mathbf{B}_{k}\rightarrow \mathbf{I}%
		_{r_{-k}},$
%	\end{equation*}%
	as $p_{-k}\rightarrow \infty $; and by Assumption \ref{as-2}\textit{(ii)}
	again, and the definition of $\widehat{\mathbf{A}}_{k}$,
%	\begin{equation*}
$		\Vert \mathbf{A}_{k}\Vert _{F}^{2}\asymp \Vert \widehat{\mathbf{A}}_{k}\Vert
		_{F}^{2}\asymp p_{k},$
%	\end{equation*}%
	as $p_{k}\rightarrow \infty $. Also, by Lemma \ref{lemma2}, as $\min
	\{T,p_{1},\ldots ,p_{k}\}\rightarrow \infty $, $\lambda _{j}(\widehat{%
		\mathbf{M}}_{k})$ converge to some positive constants for all $j\leq r_{k}$;
	hence
	\begin{equation}
		\left\Vert \widehat{\mathbf{\Lambda }}_{k}\right\Vert _{F}=O_{P}(1),\quad
		\left\Vert \widehat{\mathbf{\Lambda }}_{k}^{-1}\right\Vert _{F}=O_{P}(1).
		\label{eq:LambdaHat}
	\end{equation}%
	Putting all together, it follows that
	\begin{equation}
		\left\Vert \widehat{\mathbf{H}}_{k}\right\Vert _{F}=O_{P}(1).  \label{h-hat}
	\end{equation}%
	Using now (\ref{eq:AHATapp}) and (\ref{eq:LambdaHat}), it holds that
	\begin{align}
		\frac{1}{p_{k}}\left\Vert \widehat{\mathbf{A}}_{k}-\mathbf{A}_{k}\widehat{%
			\mathbf{H}}_{k}\right\Vert _{F}^{2}& =\frac{1}{p_{k}}\left\Vert (\mathcal{I}%
		\mathcal{I}+\mathcal{I}\mathcal{I}\mathcal{I}+\mathcal{I}\mathcal{V})%
		\widehat{\mathbf{A}}_{k}\widehat{\boldsymbol{\Lambda }}_{k}^{-1}\right\Vert
		_{F}^{2}   \lesssim \frac{1}{p_{k}}\left( \left\Vert \mathcal{I}\mathcal{I}\widehat{%
			\mathbf{A}}_{k}\right\Vert _{F}^{2}+\left\Vert \mathcal{I}\mathcal{I}%
		\mathcal{I}\widehat{\mathbf{A}}_{k}\right\Vert _{F}^{2}+\left\Vert \mathcal{I%
		}\mathcal{V}\widehat{\mathbf{A}}_{k}\right\Vert _{F}^{2}\right)  \notag \\
		& =O_{P}\left( \frac{1}{Tp_{-k}}+\frac{1}{p_{k}^{2}}\right) +o_{p}\left(
		\frac{1}{p_{k}}\left\Vert \widehat{\mathbf{A}}_{k}-\mathbf{A}_{k}\widehat{%
			\mathbf{H}}_{k}\right\Vert _{F}^{2}\right) .  \label{eq:statement31}
	\end{align}%
	We note that the last term in (\ref{eq:statement31}) can be neglected:
	indeed, $p_{k}^{-1}\left\Vert \widehat{\mathbf{A}}_{k}-\mathbf{A}_{k}%
	\widehat{\mathbf{H}}_{k}\right\Vert _{F}^{2}$ is bounded byt $O_{P}(1)$
	because $\Vert \mathbf{A}_{k}\Vert _{F}^{2}\asymp p_{k}$, $\Vert \widehat{%
		\mathbf{A}}_{k}\Vert _{F}^{2}\asymp p_{k}$, and $\Vert \widehat{\mathbf{H}}%
	_{k}\Vert _{F}=O_{P}(1)$ as shown in (\ref{h-hat}). In order to complete the
	proof, it remains to show that $\widehat{\mathbf{H}}_{k}^{\top }\widehat{%
		\mathbf{H}}_{k}\overset{P}{\rightarrow }\mathbf{I}_{r_{k}}.$ By
	Cauchy-Schwartz inequality and (\ref{eq:statement31}), it holds that%
	\begin{eqnarray*}
		\left\Vert \frac{1}{p_{k}}\mathbf{A}_{k}^{\top }\left( \widehat{\mathbf{A}}%
		_{k}-\mathbf{A}_{k}\widehat{\mathbf{H}}_{k}\right) \right\Vert _{F}^{2}
		&\leq &\frac{\left\Vert \mathbf{A}_{k}\right\Vert _{F}^{2}}{p_{k}}\frac{%
			\left\Vert \widehat{\mathbf{A}}_{k}-\mathbf{A}_{k}\widehat{\mathbf{H}}%
			_{k}\right\Vert _{F}^{2}}{p_{k}}=o_{P}\left( 1\right) , \\
		\left\Vert \frac{1}{p_{k}}\widehat{\mathbf{A}}_{k}^{\top }\left( \widehat{%
			\mathbf{A}}_{k}-\mathbf{A}_{k}\widehat{\mathbf{H}}_{k}\right) \right\Vert
		_{F}^{2} &\leq &\frac{\left\Vert \widehat{\mathbf{A}}_{k}\right\Vert _{F}^{2}%
		}{p_{k}}\frac{\left\Vert \widehat{\mathbf{A}}_{k}-\mathbf{A}_{k}\widehat{%
				\mathbf{H}}_{k}\right\Vert _{F}^{2}}{p_{k}}=o_{P}\left( 1\right) .
	\end{eqnarray*}%
	Note now that, by construction, $p_{k}^{-1}\widehat{\mathbf{A}}_{k}^{\top }%
	\widehat{\mathbf{A}}_{k}=\mathbf{I}_{r_{k}}$; also, by Assumption \ref{as-2}%
	\textit{(ii)}, $p_{k}^{-1}\mathbf{A}_{k}^{\top }\mathbf{A}_{k}\rightarrow
	\mathbf{I}_{r_{k}}$. Hence, by (\ref{eq:statement31})
	\begin{equation*}
		\mathbf{I}_{r_{k}}=\frac{1}{p_{k}}\widehat{\mathbf{A}}_{k}^{\top }\mathbf{A}%
		_{k}\widehat{\mathbf{H}}_{k}+o_{P}(1)=\widehat{\mathbf{H}}_{k}^{\top }%
		\widehat{\mathbf{H}}_{k}+o_{P}(1),
	\end{equation*}%
	which concludes the proof.
\end{proof}

\begin{proof}[Proof of Theorem \protect\ref{th:7}]
	Based on equation (\ref{dec-m-hat}), for all $1\leq i\leq p_{k}$, it holds
	that
	\begin{equation*}
		\widehat{\mathbf{A}}_{k,i\cdot }^{\top }-\widehat{\mathbf{H}}_{k}^{\top }%
		\mathbf{A}_{k,i\cdot }^{\top }=\frac{1}{Tp}\sum_{t=1}^{T}\widehat{%
			\boldsymbol{\Lambda }}_{k}^{-1}\left( \widehat{\mathbf{A}}_{k}^{\top }%
		\mathbf{E}_{k,t}\mathbf{B}_{k}\mathbf{F}_{k,t}^{\top }\mathbf{A}_{k,i\cdot
		}^{\top }+\widehat{\mathbf{A}}_{k}^{\top }\mathbf{A}_{k}\mathbf{F}_{k,t}%
		\mathbf{B}_{k}^{\top }\boldsymbol{e}_{k,t,i\cdot }^{\top }+\widehat{\mathbf{A%
		}}_{k}^{\top }\mathbf{E}_{k,t}\boldsymbol{e}_{k,t,i\cdot }^{\top }\right) .
	\end{equation*}%
	Using now: equation (\ref{eq:LambdaHat}) and Assumptions \ref{as-2}\textit{%
		(i)}and \ref{as-4}\textit{(i)}, it follows that%
	\begin{eqnarray}
		&&\left\Vert \frac{1}{Tp}\sum_{t=1}^{T}\widehat{\boldsymbol{\Lambda }}%
		_{k}^{-1}\widehat{\mathbf{A}}_{k}^{\top }\mathbf{E}_{k,t}\mathbf{B}_{k}%
		\mathbf{F}_{k,t}^{\top }\mathbf{A}_{k,i\cdot }^{\top }\right\Vert _{F}
		\lesssim \frac{1}{\sqrt{Tp}}\left\Vert \frac{1}{\sqrt{T}}%
		\sum_{t=1}^{T}\left( \frac{\widehat{\mathbf{A}}_{k}}{p_{k}^{1/2}}\right)
		^{\top }\mathbf{E}_{k,t}\left( \frac{\mathbf{B}_{k}}{p_{-k}^{1/2}}\right)
		\mathbf{F}_{k,t}^{\top }\right\Vert _{F}=O_{P}\left( \frac{1}{\sqrt{Tp}}%
		\right) .  \label{th7-1} 
	\end{eqnarray}%
	By the same token, and using Assumptions \ref{as-2}\textit{(i)} and \ref%
	{as-4}\textit{(ii)}, it holds that
	\begin{eqnarray}
		&&\left\Vert \frac{1}{Tp}\sum_{t=1}^{T}\widehat{\boldsymbol{\Lambda }}%
		_{k}^{-1}\widehat{\mathbf{A}}_{k}^{\top }\mathbf{A}_{k}\mathbf{F}_{k,t}%
		\mathbf{B}_{k}^{\top }\boldsymbol{e}_{k,t,i\cdot }^{\top }\right\Vert _{F}
		\label{th7-2} \\
		&\leq &\left\Vert \widehat{\boldsymbol{\Lambda }}_{k}^{-1}\right\Vert
		_{F}\left\Vert \frac{1}{p_{k}}\widehat{\mathbf{A}}_{k}^{\top }\mathbf{A}%
		_{k}\right\Vert _{F}\left\vert \frac{1}{Tp_{-k}}\sum_{t=1}^{T}\mathbf{F}%
		_{k,t}\mathbf{B}_{k}^{\top }\boldsymbol{e}_{k,t,i\cdot }^{\top }\right\vert
		=O_{P}\left( \frac{1}{\sqrt{Tp_{-k}}}\right) .  \notag
	\end{eqnarray}%
	Moreover, by the same arguments used above, and using Lemma \ref{lemma1}%
	\textit{(iii)}, we have
	\begin{eqnarray}
		&&\left\Vert \frac{1}{Tp}\sum_{t=1}^{T}\widehat{\boldsymbol{\Lambda }}%
		_{k}^{-1}\widehat{\mathbf{A}}_{k}^{\top }\mathbf{E}_{k,t}\boldsymbol{e}%
		_{k,t,i\cdot }^{\top }\right\Vert _{F} \lesssim \left\Vert \frac{1}{Tp}\sum_{t=1}^{T}\left( \widehat{\mathbf{A}}%
		_{k}-\mathbf{A}_{k}\widehat{\mathbf{H}}_{k}\right) ^{\top }\mathbf{E}_{k,t}%
		\boldsymbol{e}_{k,t,i\cdot }^{\top }\right\Vert _{F}+\left\Vert \frac{1}{Tp}%
		\sum_{t=1}^{T}\widehat{\mathbf{H}}_{k}^{\top }\mathbf{A}_{k}^{\top }\mathbf{E%
		}_{k,t}\boldsymbol{e}_{k,t,i\cdot }^{\top }\right\Vert _{F}   \label{th7-3} \\
		&\lesssim &\left\Vert \widehat{\mathbf{A}}_{k}-\mathbf{A}_{k}\widehat{%
			\mathbf{H}}_{k}\right\Vert _{F}\left\Vert \frac{1}{Tp}\sum_{t=1}^{T}\mathbf{E%
		}_{k,t}\boldsymbol{e}_{k,t,i\cdot }^{\top }\right\Vert _{F}+\left\Vert \frac{%
			1}{Tp}\sum_{t=1}^{T}\mathbf{A}_{k}^{\top }\mathbf{E}_{k,t}\boldsymbol{e}%
		_{k,t,i\cdot }^{\top }\right\Vert _{F}  =O_{P}\left( \frac{1}{p_{k}}\right) +O_{P}\left( \frac{1}{\sqrt{Tp}}%
		\right) .  \notag
	\end{eqnarray}%
	Therefore, as $\min \left\{ T,p_{1},\ldots ,p_{K}\right\} \rightarrow \infty
	$ under the restriction
%	\begin{equation*}
		$Tp_{-k}=o\left( p_{k}^{2}\right)$ ,
%	\end{equation*}%
	it follows that
	\begin{equation*}
		\sqrt{Tp_{-k}}\left( \widehat{\mathbf{A}}_{k,i\cdot }^{\top }-\widehat{%
			\mathbf{H}}_{k}^{\top }\mathbf{A}_{k,i\cdot }^{\top }\right) =\widehat{%
			\boldsymbol{\Lambda }}_{k}^{-1}\frac{\widehat{\mathbf{A}}_{k}^{\top }\mathbf{%
				A}_{k}}{p_{k}}\frac{1}{\sqrt{Tp_{-k}}}\sum_{t=1}^{T}\mathbf{F}_{k,t}\mathbf{B%
		}_{k}^{\top }\boldsymbol{e}_{k,t,i\cdot }^{\top }+{o}_{{P}}({1}).
	\end{equation*}%
	Recall that, by Lemma \ref{lemma2}, it holds that as $\min \left\{
	T,p_{1},\ldots ,p_{K}\right\} \rightarrow \infty $
%	\begin{equation*}
		$\widehat{\boldsymbol{\Lambda }}_{k}\overset{P}{\rightarrow }\boldsymbol{%
			\Lambda }_{k}$;
%	\end{equation*}%
	further, from the proof of Theorem \ref{th:1}, it holds that
%	\begin{equation*}
		$\widehat{\mathbf{H}}_{k}=p_{k}^{-1}\mathbf{A}_{k}^{\top }\widehat{%
			\mathbf{A}}_{k}+{o}_{{P}}({1})$;
%	\end{equation*}%
	using Assumption \ref{as-1}\textit{(ii)}, the definition of $\widehat{%
		\mathbf{H}}_{k}$ in (\ref{eq:defHk}), and recalling that $\mathbf{\Sigma }%
	_{k}$ has spectral decomposition $\mathbf{\Sigma }_{k}=\boldsymbol{\Gamma }%
	_{k}\boldsymbol{\Lambda }_{k}\boldsymbol{\Gamma }_{k}^{\top }$, we have
%	\begin{equation*}
		$\widehat{\mathbf{H}}_{k}=\boldsymbol{\Gamma }_{k}\boldsymbol{\Lambda }_{k}%
		\boldsymbol{\Gamma }_{k}^{\top }\widehat{\mathbf{H}}_{k}\boldsymbol{\Lambda }%
		_{k}^{-1}+{o}_{{P}}({1})$.
%	\end{equation*}%
	By rearranging the last equation and recalling that $\boldsymbol{\Gamma }%
	_{k}^{\top }\boldsymbol{\Gamma }_{k}=\mathbf{I}_{r_{k}}$,
%	\begin{equation*}
		$\boldsymbol{\Gamma }_{k}^{\top }\widehat{\mathbf{H}}_{k}\boldsymbol{\Lambda }%
		_{k}=\boldsymbol{\Lambda }_{k}\boldsymbol{\Gamma }_{k}^{\top }\widehat{%
			\mathbf{H}}_{k}+{o}_{{P}}({1})$.
%	\end{equation*}%
	We note that $\boldsymbol{\Lambda }_{k}$ is diagonal with distinct entries;
	hence, $\boldsymbol{\Gamma }_{k}^{\top }\widehat{\mathbf{H}}_{k}$ must be
	asymptotically diagonal. Furthermore, since
%	\begin{equation*}
		$\widehat{\mathbf{H}}_{k}^{\top }\widehat{\mathbf{H}}_{k}=\left( \boldsymbol{%
			\Gamma }_{k}^{\top }\widehat{\mathbf{H}}_{k}\right) ^{\top }\boldsymbol{%
			\Gamma }_{k}^{\top }\widehat{\mathbf{H}}_{k}=\mathbf{I}_{r_{k}}+{o}_{{P}}({1}%
		)$,
%	\end{equation*}%
	the diagonal entries of $\boldsymbol{\Gamma }_{k}^{\top }\widehat{\mathbf{H}}%
	_{k}$ must be asymptotically $1$ or $-1$, and we can always choose the
	column signs of $\widehat{\mathbf{A}}_{k}$in such a way that
%	\begin{equation*}
		$\boldsymbol{\Gamma }_{k}^{\top }\widehat{\mathbf{H}}_{k}=\mathbf{I}_{r_{k}}+{%
			o}_{{P}}({1})$.
%	\end{equation*}%
	Therefore, it finally follows that
%	\begin{equation*}
		${p_{k}}^{-1}{\widehat{\mathbf{A}}_{k}^{\top }\mathbf{A}_{k}}=\widehat{%
			\mathbf{H}}_{k}^{\top }+{o}_{{P}}({1})=\boldsymbol{\Gamma }_{k}^{\top }+{o}_{%
			{P}}({1})$.
%	\end{equation*}%
	Consequently, by Slutsky's theorem and Assumption \ref{as-5}\textit{(i)}, as
	$\min \left\{ T,p_{1},\ldots ,p_{K}\right\} \rightarrow \infty $ under the
	restriction $Tp_{-k}=o\left( p_{k}^{2}\right) $, it holds that
	\begin{equation*}
		\sqrt{Tp_{-k}}\left( \widehat{\mathbf{A}}_{k,i\cdot }^{\top }-\widehat{%
			\mathbf{H}}_{k}^{\top }\mathbf{A}_{k,i\cdot }^{\top }\right) \overset{D}{%
			\rightarrow }\mathcal{N}\left( \mathbf{0},\boldsymbol{\Lambda }_{k}^{-1}%
		\boldsymbol{\Gamma }_{k}^{\top }\mathbf{V}_{ki}\boldsymbol{\Gamma }_{k}%
		\boldsymbol{\Lambda }_{k}^{-1}\right) ,
	\end{equation*}%
	which concludes the proof of part \textit{(i)} of the theorem. Whenever $%
	Tp_{-k}\gtrsim p_{k}^{2}$, it immediately follows from (\ref{th7-1})-(\ref%
	{th7-3}) that
%	\begin{equation*}
		$\Vert \widehat{\mathbf{A}}_{k,i\cdot }^{\top }-\widehat{\mathbf{H}}%
		_{k}^{\top }\mathbf{A}_{k,i\cdot }^{\top }\Vert =O_{P}\left(p_k^{-1}%
		\right) ,$
%	\end{equation*}%
	proving part \textit{(ii)} of the theorem.
\end{proof}

\begin{proof}[Proof of Theorem \protect\ref{th:2}]
	The proof follows a similar logic to the proof of Theorem \ref{th:1}. Recall
	(\ref{dec-m-tilde}), and define the diagonal matrix containing the first $%
	r_{k}$ eigenvalues of $\widetilde{\mathbf{M}}_{k}\ $as%
%	\begin{equation*}
		$\widetilde{\boldsymbol{\Lambda }}_{k}=diag\left\{ \lambda _{1}\left(
		\widetilde{\mathbf{M}}_{k}\right) ,...,\lambda _{r_{k}}\left( \widetilde{%
			\mathbf{M}}_{k}\right) \right\} $.
%	\end{equation*}%
	By definition, \linebreak $\widetilde{\mathbf{A}}_{k}:=\sqrt{p_{k}}\,\widetilde{\mathbf{%
			U}}_{k}$, where the columns of $\widetilde{\mathbf{U}}_{k}$ are orthonormal
	basis of the eigenvectors corresponding to the $r_{k}$ largest eigenvalues
	of $\widetilde{\mathbf{M}}_{k}$. Hence
%	\begin{equation*}
		$\widetilde{\mathbf{A}}_{k}\widetilde{\mathbf{\Lambda }}_{k}=\widetilde{%
			\mathbf{M}}_{k}\widetilde{\mathbf{A}}_{k}.$
%	\end{equation*}%
	Define now
	\begin{equation}
		\widetilde{\mathbf{H}}_{k}:=\frac{1}{Tp_{k}p_{-k}^{2}}\left( \sum_{t=1}^{T}%
		\mathbf{F}_{k,t}\mathbf{B}_{k}^{\top }\widehat{\mathbf{B}}_{k}\widehat{%
			\mathbf{B}}_{k}^{\top }\mathbf{B}_{k}\mathbf{F}_{k,t}^{\top }\right) \mathbf{%
			A}_{k}^{\top }\widetilde{\mathbf{A}}_{k}\widetilde{\boldsymbol{\Lambda }}%
		_{k}^{-1}.  \label{eq:Htildedefinition}
	\end{equation}%
	It is easy to see, in the light of (\ref{dec-m-tilde}), that
	\begin{equation}
		\widetilde{\mathbf{A}}_{k}-\mathbf{A}_{k}\widetilde{\mathbf{H}}_{k}=(%
		\mathcal{VI}+\mathcal{VII}+\mathcal{VIII})\widetilde{\mathbf{A}}_{k}%
		\widetilde{\boldsymbol{\Lambda }}_{k}^{-1}.  \label{eq:Atildeapp}
	\end{equation}%
	Note now that, using Assumption \ref{as-2}\textit{(ii)} and the definitions
	of $\widetilde{\mathbf{A}}_{k}$ in (\ref{eq:Atildekdef}) and $\widehat{%
		\mathbf{B}}_{k}$ following from (\ref{eq:Ahatkdef}) respectively, ithat $%
	\left\Vert \mathbf{A}_{k}\right\Vert _{F}^{2}\asymp \left\Vert \widetilde{%
		\mathbf{A}}_{k}\right\Vert _{F}^{2}\asymp p_{k}$ and $\left\Vert \mathbf{B}%
	_{k}\right\Vert _{F}^{2}\asymp \left\Vert \widehat{\mathbf{B}}%
	_{k}\right\Vert _{F}^{2}\asymp p_{-k}$ as $p_{-k}$, $p_{k}\rightarrow \infty
	$. Assumption \ref{as-1}\textit{(ii)} states that
%	\begin{equation*}
		$T^{-1}\sum_{t=1}^{T}\mathbf{F}_{k,t}\mathbf{F}_{k,t}^{\top }\overset{P}{%
			\rightarrow }\boldsymbol{\Sigma }_{k}$,
%	\end{equation*}
	as $T\rightarrow \infty $; further, by Lemma \ref{lemma2}, the diagonal
	entries of $\widetilde{\mathbf{\Lambda }}_{k}$ converge to some positive
	constants as $\min \{T,p_{1},\ldots ,p_{k}\}\rightarrow \infty $. Hence
	\begin{equation}
		\left\Vert \widetilde{\mathbf{\Lambda }}_{k}\right\Vert _{F}=O_{P}(1),\quad
		\left\Vert \widetilde{\mathbf{\Lambda }}_{k}^{-1}\right\Vert _{F}=O_{P}(1).
		\label{eq:LambdaTilde}
	\end{equation}%
	Therefore, we have $\Vert \widetilde{\mathbf{H}}_{k}\Vert _{F}=O_{P}(1)$.
	Combining (\ref{eq:Atildeapp}) and (\ref{eq:LambdaTilde}) and using Lemma %
	\ref{lemma2}, it holds that
	\begin{align}
		& \frac{1}{p_{k}}\left\Vert \widetilde{\mathbf{A}}_{k}-\mathbf{A}_{k}%
		\widetilde{\mathbf{H}}_{k}\right\Vert _{F}^{2}  =\frac{1}{p_{k}}\left\Vert (\mathcal{VI}+\mathcal{VII}+\mathcal{VIII})%
		\widetilde{\mathbf{A}}_{k}\widetilde{\boldsymbol{\Lambda }}%
		_{k}^{-1}\right\Vert _{F}^{2}  \notag \\
		& \leq \frac{1}{p_{k}}\left( \left\Vert \mathcal{VI}\widetilde{\mathbf{A}}%
		_{k}\right\Vert _{F}^{2}+\left\Vert \mathcal{VII}\widetilde{\mathbf{A}}%
		_{k}\right\Vert _{F}^{2}+\left\Vert \mathcal{VIII}\widetilde{\mathbf{A}}%
		_{k}\right\Vert _{F}^{2}\right)  \notag \\
		& =O_{P}\left( \frac{1}{Tp_{-k}}+\frac{1}{p^{2}}+w_{-k}^{2}\left( \frac{1}{%
			p_{k}^{2}}+\frac{1}{Tp_{k}}\right) +m_{-k}\right) +o_{P}\left( \frac{1}{p_{k}%
		}\left\Vert \widetilde{\mathbf{A}}_{k}-\mathbf{A}_{k}\widetilde{\mathbf{H}}%
		_{k}\right\Vert _{F}^{2}\right) .  \label{eq:statement33}
	\end{align}%
	By the same logic as in the proof of Theorem \ref{th:1}, the last term in (%
	\ref{eq:statement33}) can be neglected. In order to complete the proof, it
	remains to show that $\widetilde{\mathbf{H}}_{k}^{\top }\widetilde{\mathbf{H}%
	}_{k}\overset{P}{\rightarrow }\mathbf{I}_{r_{k}}$. Using, as in the proof of
	Theorem \ref{th:1}, the Cauchy-Schwartz inequality and (\ref{eq:statement33}%
	), it holds that
	\begin{eqnarray*}
		\left\Vert \frac{1}{p_{k}}\mathbf{A}_{k}^{\top }\left( \widetilde{\mathbf{A}}%
		_{k}-\mathbf{A}_{k}\widetilde{\mathbf{H}}_{k}\right) \right\Vert _{F}^{2}
		&\leq &\frac{\left\Vert \mathbf{A}_{k}\right\Vert _{F}^{2}}{p_{k}}\frac{%
			\left\Vert \widetilde{\mathbf{A}}_{k}-\mathbf{A}_{k}\widetilde{\mathbf{H}}%
			_{k}\right\Vert _{F}^{2}}{p_{k}}=o_{P}\left( 1\right) , \\
		\left\Vert \frac{1}{p_{k}}\widetilde{\mathbf{A}}_{k}^{\top }\left(
		\widetilde{\mathbf{A}}_{k}-\mathbf{A}_{k}\widetilde{\mathbf{H}}_{k}\right)
		\right\Vert _{F}^{2} &\leq &\frac{\left\Vert \widetilde{\mathbf{A}}%
			_{k}\right\Vert _{F}^{2}}{p_{k}}\frac{\left\Vert \widetilde{\mathbf{A}}_{k}-%
			\mathbf{A}_{k}\widetilde{\mathbf{H}}_{k}\right\Vert _{F}^{2}}{p_{k}}%
		=o_{P}\left( 1\right) ,
	\end{eqnarray*}%
	since, as noted above, $\left\Vert \mathbf{A}_{k}\right\Vert _{F}^{2}\asymp
	\left\Vert \widetilde{\mathbf{A}}_{k}\right\Vert _{F}^{2}\asymp p_{k}$.
	Recalling that $p_{k}^{-1}\widetilde{\mathbf{A}}_{k}^{\top }\widetilde{%
		\mathbf{A}}_{k}=\mathbf{I}_{r_{k}}$ and $p_{k}^{-1}\mathbf{A}_{k}^{\top }%
	\mathbf{A}_{k}\rightarrow \mathbf{I}_{r_{k}}$, then from (\ref%
	{eq:statement33})
%	\begin{equation*}
		$\mathbf{I}_{r_{k}}={p_{k}}^{-1}\widetilde{\mathbf{A}}_{k}^{\top }\mathbf{A%
		}_{k}\widetilde{\mathbf{H}}_{k}+{o}_{{P}}({1})=\widetilde{\mathbf{H}}%
		_{k}^{\top }\widetilde{\mathbf{H}}_{k}+{o}_{{P}}({1}),$
%	\end{equation*}%
	which concludes the proof.
\end{proof}

\begin{proof}[Proof of Lemma \protect\ref{lem:1}]
	Consider part \textit{(i)} of the lemma, and note that%
	\begin{eqnarray*}
		&&\widehat{\mathbf{B}}_{k}-\mathbf{B}_{k}\widehat{\mathbf{H}}_{-k} =\widehat{\mathbf{A}}_{K}\otimes \cdots \otimes \widehat{\mathbf{A}}%
		_{k+1}\otimes \widehat{\mathbf{A}}_{k-1}\otimes \cdots \otimes \widehat{%
			\mathbf{A}}_{1}-\left( \otimes _{j=1,j\neq k}^{K}\mathbf{A}_{j}\right)
		\left( \otimes _{j=1,j\neq k}^{K}\widehat{\mathbf{H}}_{j}\right) \\
		&=&\left( \widehat{\mathbf{A}}_{K}-\mathbf{A}_{K}\widehat{\mathbf{H}}_{K}+%
		\mathbf{A}_{K}\widehat{\mathbf{H}}_{K}\right) \otimes \cdots \otimes \left(
		\widehat{\mathbf{A}}_{k+1}-\mathbf{A}_{k+1}\widehat{\mathbf{H}}_{k+1}+%
		\mathbf{A}_{k+1}\widehat{\mathbf{H}}_{k+1}\right) \\
		&&\otimes \left( \widehat{\mathbf{A}}_{k-1}-\mathbf{A}_{k-1}\widehat{\mathbf{%
				H}}_{k-1}+\mathbf{A}_{k-1}\widehat{\mathbf{H}}_{k-1}\right) \otimes \cdots
		\otimes \left( \widehat{\mathbf{A}}_{1}-\mathbf{A}_{1}\widehat{\mathbf{H}}%
		_{1}+\mathbf{A}_{1}\widehat{\mathbf{H}}_{1}\right) \\
		&&-\left( \otimes _{j=1,j\neq k}^{K}\left( \mathbf{A}_{j}\widehat{\mathbf{H}}%
		_{j}\right) \right) \\
		&=&\left( \widehat{\mathbf{A}}_{K}-\mathbf{A}_{K}\widehat{\mathbf{H}}%
		_{K}\right) \otimes \left( \otimes _{j\neq k,K}\widehat{\mathbf{A}}%
		_{j}\right) +\mathbf{A}_{K}\widehat{\mathbf{H}}_{K}\otimes \left( \widehat{%
			\mathbf{A}}_{K-1}-\mathbf{A}_{K-1}\widehat{\mathbf{H}}_{K-1}\right) \otimes
		\left( \otimes _{j\neq k,K-1,K}\widehat{\mathbf{A}}_{j}\right) \\
		&&+\cdots +\otimes _{j\neq 1,k}\left( \mathbf{A}_{j}\widehat{\mathbf{H}}%
		_{j}\right) \otimes \left( \widehat{\mathbf{A}}_{1}-\mathbf{A}_{1}\widehat{%
			\mathbf{H}}_{1}\right) . \label{diocane}
	\end{eqnarray*}%
	Hence, using Theorem \ref{th:1}, Assumption \ref{as-2}\textit{(ii)}, and the
	definition of $\widehat{\mathbf{A}}_{\ell }$ in (\ref{eq:Ahatkdef}), it holds
	that
	\begin{eqnarray*}
		&&\frac{1}{p_{-k}}\left\Vert \widehat{\mathbf{B}}_{k}-\mathbf{B}_{k}\widehat{%
			\mathbf{H}}_{-k}\right\Vert _{F}^{2} \lesssim \frac{1}{p_{-k}}\sum_{j=1,j\neq k}^{K}\left\Vert \left( \otimes
		_{\ell =k+1,\ell \neq k}^{K}\mathbf{A}_{\ell }\widehat{\mathbf{H}}_{\ell
		}\right) \otimes \left( \widehat{\mathbf{A}}_{j}-\mathbf{A}_{j}\widehat{%
			\mathbf{H}}_{j}\right) \otimes \left( \otimes _{\ell =1,\ell \neq k}^{j-1}%
		\widehat{\mathbf{A}}_{\ell }\right) \right\Vert _{F}^{2} \\
		&\lesssim &\frac{1}{p_{-k}}\sum_{j=1,j\neq k}^{K}\left\Vert \left( \otimes
		_{\ell =k+1,\ell \neq k}^{K}\mathbf{A}_{\ell }\right) \otimes \left(
		\widehat{\mathbf{A}}_{j}-\mathbf{A}_{j}\widehat{\mathbf{H}}_{j}\right)
		\otimes \left( \otimes _{\ell =1,\ell \neq k}^{j-1}\mathbf{A}_{\ell }\right)
		\right\Vert _{F}^{2} \lesssim \frac{1}{p_{-k}}\sum_{j=1,j\neq k}^{K}\left( \prod_{\ell =1,\ell
			\neq j,k}^{K}\left\Vert \mathbf{A}_{\ell }\right\Vert _{F}^{2}\right)
		\left\Vert \widehat{\mathbf{A}}_{j}-\mathbf{A}_{j}\widehat{\mathbf{H}}%
		_{j}\right\Vert _{F}^{2} \\
		&\lesssim &\sum_{j=1,j\neq k}^{K}\frac{1}{p_{j}}\left\Vert \widehat{\mathbf{A%
		}}_{j}-\mathbf{A}_{j}\widehat{\mathbf{H}}_{j}\right\Vert
		_{F}^{2}=O_{P}\left( \sum_{j=1,j\neq k}^{K}w_{j}\right)
	\end{eqnarray*}%
	By letting
%	\begin{equation*}
		$w_{-k}=\sum_{j=1,j\neq k}^{K}w_{j}$,
%	\end{equation*}%
	it follows immediately that the sufficient condition in (\ref{equ:3.1})%
	\textit{(i)} is satisfied. We now turn to part \textit{(ii)} of the lemma.
	Recall that by equation (\ref{dec-m-hat}), it holds that
%	\begin{equation*}
		$\widehat{\mathbf{M}}_{k}=\mathcal{I}+\mathcal{II}+\mathcal{III}+\mathcal{IV}$;
%	\end{equation*}%
	further, from (\ref{eq:statement31}), we have
%	\begin{equation*}
		$\left( \widehat{\mathbf{A}}_{k}-\mathbf{A}_{k}\widehat{\mathbf{H}}%
		_{k}\right) =\left( \mathcal{II}+\mathcal{III}+\mathcal{IV}\right) \widehat{%
			\mathbf{A}}_{k}\widehat{\mathbf{\Lambda }}_{k}^{-1}.$
%	\end{equation*}%
	Thus, using (\ref{eq:LambdaHat}) and (\ref{diocane}), we obtain%
	\begin{eqnarray}
		&&\frac{1}{p_{k}}\left\Vert \frac{1}{Tp_{-k}}\sum_{s=1}^{T}\mathbf{E}%
		_{k,s}\left( \widehat{\mathbf{B}}_{k}-\mathbf{B}_{k}\widehat{\mathbf{H}}%
		_{-k}\right) \mathbf{F}_{k,s}^{\top }\right\Vert _{F}^{2}  \label{eq:131415}
		\\
		&\lesssim &\frac{1}{p_{k}}\sum_{j=1,j\neq k}^{K}\left\Vert \frac{1}{T}%
		\sum_{s=1}^{T}\mathbf{E}_{k,s}\left( \left( \otimes _{{\ell }=j+1,{\ell }%
			\neq k}^{K}\frac{1}{p_{\ell }}\mathbf{A}_{\ell }\widehat{\mathbf{H}}_{\ell
		}\right) \otimes \frac{1}{p_{j}}\left( \widehat{\mathbf{A}}_{j}-\mathbf{A}%
		_{j}\widehat{\mathbf{H}}_{j}\right) \otimes \left( \otimes _{{\ell }=1,{\ell
			}\neq k}^{j-1}\frac{1}{p_{\ell }}\widehat{\mathbf{A}}_{\ell }\right) \right)
		\mathbf{F}_{k,s}^{\top }\right\Vert _{F}^{2}  \notag \\
		&\lesssim &\frac{1}{p_{k}}\sum_{j=1,j\neq k}^{K}\left\Vert \frac{1}{T}%
		\sum_{s=1}^{T}\mathbf{E}_{k,s}\left( \left( \otimes _{{\ell }=j+1,{\ell }%
			\neq k}^{K}\frac{1}{p_{\ell }}\mathbf{A}_{\ell }\right) \otimes \frac{1}{%
			p_{j}}\left( \widehat{\mathbf{A}}_{j}-\mathbf{A}_{j}\widehat{\mathbf{H}}%
		_{j}\right) \otimes \left( \otimes _{{\ell }=1,{\ell }\neq k}^{j-1}\frac{1}{%
			p_{\ell }}\mathbf{A}_{\ell }\right) \right) \mathbf{F}_{k,s}^{\top
		}\right\Vert _{F}^{2}  \notag \\
		&\lesssim &\frac{1}{p_{k}}\sum_{j=1,j\neq k}^{K}\left\Vert \frac{1}{Tp_{-k}}%
		\sum_{s=1}^{T}\mathbf{E}_{k,s}\left( \left( \otimes _{{\ell }=j+1,{\ell }%
			\neq k}^{K}\mathbf{A}_{\ell }\right) \otimes \left( \left( \mathcal{II}+%
		\mathcal{III}+\mathcal{IV}\right) \widehat{\mathbf{A}}_{j}\widehat{\mathbf{%
				\Lambda }}_{j}^{-1}\right) \otimes \left( \otimes _{{\ell }=1,{\ell }\neq
			k}^{j-1}\mathbf{A}_{\ell }\right) \right) \mathbf{F}_{k,s}^{\top
		}\right\Vert _{F}^{2}  \notag \\
		&\lesssim &\frac{1}{p_{k}}\sum_{j=1,j\neq k}^{K}\left\Vert \frac{1}{Tp_{-k}}%
		\sum_{s=1}^{T}\mathbf{E}_{k,s}\left( \left( \otimes _{{\ell }=j+1,{\ell }%
			\neq k}^{K}\mathbf{A}_{\ell }\right) \otimes \left( \mathcal{II}\widehat{%
			\mathbf{A}}_{j}\right) \otimes \left( \otimes _{{\ell }=1,{\ell }\neq
			k}^{j-1}\mathbf{A}_{\ell }\right) \right) \mathbf{F}_{k,s}^{\top
		}\right\Vert _{F}^{2}  \notag \\
		&&+\frac{1}{p_{k}}\sum_{j=1,j\neq k}^{K}\left\Vert \frac{1}{Tp_{-k}}%
		\sum_{s=1}^{T}\mathbf{E}_{k,s}\left( \left( \otimes _{{\ell }=j+1,{\ell }%
			\neq k}^{K}\mathbf{A}_{\ell }\right) \otimes \left( \mathcal{III}\widehat{%
			\mathbf{A}}_{j}\right) \otimes \left( \otimes _{{\ell }=1,{\ell }\neq
			k}^{j-1}\mathbf{A}_{\ell }\right) \right) \mathbf{F}_{k,s}^{\top
		}\right\Vert _{F}^{2}  \notag \\
		&&\frac{1}{p_{k}}\sum_{j=1,j\neq k}^{K}\left\Vert \frac{1}{Tp_{-k}}%
		\sum_{s=1}^{T}\mathbf{E}_{k,s}\left( \left( \otimes _{{\ell }=j+1,{\ell }%
			\neq k}^{K}\mathbf{A}_{\ell }\right) \otimes \left( \mathcal{IV}\widehat{%
			\mathbf{A}}_{j}\right) \otimes \left( \otimes _{{\ell }=1,{\ell }\neq
			k}^{j-1}\mathbf{A}_{\ell }\right) \right) \mathbf{F}_{k,s}^{\top
		}\right\Vert _{F}^{2}  \notag \\
		&=&\mathcal{XII}+\mathcal{XIII}+\mathcal{XIV}\text{.}  \notag
	\end{eqnarray}%
	Consider $\mathcal{XII}$; by the definition of $\mathcal{II}$ in (\ref%
	{dec-m-hat}), and using Lemmas \ref{lemma1}\textit{(ii)} and \ref{lemma3},
	and Theorem \ref{th:1}, it follows that%
	\begin{eqnarray}
		&&\mathcal{XII}  \lesssim \frac{1}{T^{2}p_{k}p_{-k}^{2}}\sum_{j=1,j\neq k}^{K}\left\Vert
		\frac{1}{Tp}\sum_{t=1}^{T}\mathbf{F}_{j,t}\mathbf{B}_{j}^{\top }\mathbf{E}%
		_{j,t}^{\top }\widehat{\mathbf{A}}_{j}\right\Vert _{F}^{2}\left\Vert
		\sum_{s=1}^{T}\mathbf{E}_{k,s}\mathbf{B}_{k}\mathbf{F}_{k,s}^{\top
		}\right\Vert _{F}^{2}  \label{d12}  \\
		&\lesssim &\frac{1}{T^{4}p_{k}^{3}p_{-k}^{4}}\sum_{j=1,j\neq k}^{K}\bigg(
		\left\Vert \frac{1}{Tp}\sum_{t=1}^{T}\mathbf{F}_{j,t}\mathbf{B}_{j}^{\top }%
		\mathbf{E}_{j,t}^{\top }\mathbf{A}_{j}\right\Vert _{F}^{2}+\left\Vert \frac{1%
		}{Tp}\sum_{t=1}^{T}\mathbf{F}_{j,t}\mathbf{B}_{j}^{\top }\mathbf{E}%
		_{j,t}^{\top }\right\Vert _{F}^{2}\left\Vert \widehat{\mathbf{A}}_{j}-%
		\mathbf{A}_{j}\widehat{\mathbf{H}}_{j}\right\Vert _{F}^{2}\bigg)  \left\Vert
		\sum_{s=1}^{T}\mathbf{E}_{k,s}\mathbf{B}_{k}\mathbf{F}_{k,s}^{\top
		}\right\Vert _{F}^{2}  \notag \\
		&=&O_{P}\left( \frac{1}{T^{2}p_{k}p_{-k}^{2}}\sum_{j=1,j\neq k}^{K}\left( 1+%
		\frac{p_{j}}{Tp_{-j}}+\frac{1}{p_{j}}\right) \right) .  \notag
	\end{eqnarray}%
	Turning to $\mathcal{XIII}$, recall that we have assumed $r_{k}=1$, and
	therefore $\mathbf{F}_{k,t}$ is a scalar, $\mathbf{A}_{k}$ is a $p_{k}$%
	-dimensional vector with entries $A_{ki}$, $\mathbf{B}_{k}$ is a $p_{-k}$%
	-dimensional vector with entries $B_{ki}$, and
%	\begin{equation*}
		$\mathbf{\zeta }_{i_{1},...,i_{K}}=Vec\left( T^{-1/2}\sum_{t=1}^{T}\mathcal{F}%
		_{t}e_{t,i_{1},...,i_{K}}\right) ,$
%	\end{equation*}%
	is also a scalar. Let now%
%	\begin{equation*}
		$h\left( m_{1},...,m_{N}\right) =m_{1}+\left( m_{2}-1\right) I_{1}+\left(
		m_{3}-1\right) I_{1}I_{2}+...+\left( m_{N}-1\right) I_{1}I_{2}...I_{N-1},$
%	\end{equation*}%
	where
%	\begin{equation*}
		$\left( m_{1},...,m_{N}\right) \in \mathbb{Z}^{I_{1}\times ...\times I_{N}},$
%	\end{equation*}%
	and define further%
%	\begin{equation*}
		$n_{k,i_{1},...,i_{k-1},i_{k+1},...,i_{K}}=h\left(
		i_{1},...,i_{k-1},i_{k+1},...,i_{K}\right) ,$
%	\end{equation*}%
	with $1\leq i_{j}\leq p_{j}$ - then it is easy to see that $%
	n_{k,i_{1},...,i_{k-1},i_{k+1},...,i_{K}}$ runs from $1$, when $i_{j}=1$ for
	all $j$, to $p_{-k}$ when $i_{j}=p_{j}$ for all $j$. Hence, using Assumption %
	\ref{as-4}\textit{(ii)}, it follows that%
	\begin{eqnarray*}
		&&\mathbb{E}\left( \left\Vert \sum_{s=1}^{T}\mathbf{E}_{k,s}\left( \left(
		\otimes _{{\ell }=j+1,{\ell }\neq k}^{K}\mathbf{A}_{\ell }\right) \otimes
		\left( \frac{1}{Tp}\sum_{t=1}^{T}\mathbf{E}_{j,t}\mathbf{B}_{j}\mathbf{F}%
		_{j,t}^{\top }\right) \otimes \left( \otimes _{{\ell }=1,{\ell }\neq k}^{j-1}%
		\mathbf{A}_{\ell }\right) \right) \mathbf{F}_{k,s}^{\top }\right\Vert
		_{F}^{2}\right) \\
		&=&\frac{1}{p^{2}}\sum_{i_{k}=1}^{p_{k}}\mathbb{E}\left( \frac{1}{T}%
		\sum_{t,s=1}^{T}\sum_{i_{1}=1}^{p_{1}}...\sum_{i_{k-1}=1}^{p_{k-1}}%
		\sum_{i_{k+1}=1}^{p_{k+1}}...\sum_{i_{K}=1}^{p_{K}}e_{s,i_{1},...,i_{K}}%
		\left( \prod_{\ell =1,\ell \neq j,k}^{K}A_{\ell ,i_{\ell }}\right) \left(
		\mathbf{e}_{j,t,i_{j}\cdot }\mathbf{B}_{j}\right) \mathbf{F}_{k,s}\mathbf{F}%
		_{k,t}\right) ^{2} \\
		&=&\frac{1}{p^{2}}\sum_{i_{k}=1}^{p_{k}}\mathbb{E}\left( \frac{1}{T}%
		\sum_{t,s=1}^{T}\sum_{i_{1}=1}^{p_{1}}...\sum_{i_{k-1}=1}^{p_{k-1}}%
		\sum_{i_{k+1}=1}^{p_{k+1}}...\sum_{i_{K}=1}^{p_{K}}\sum_{h_{1}=1}^{p_{1}}...%
		\sum_{h_{j-1}=1}^{p_{j-1}}\sum_{h_{j+1}=1}^{p_{j+1}}...%
		\sum_{h_{K}=1}^{p_{K}}e_{s,i_{1},...,i_{K}}\right. \\
		&&\left. \left( \prod_{\ell =1,\ell \neq j,k}^{K}A_{\ell ,i_{\ell }}\right)
		e_{s,h_{1},...,i_{j},...,h_{K}}B_{j,n_{j},h_{1},...,h_{j-1},h_{j+1},...,h_{K}}%
		\mathbf{F}_{k,s}\mathbf{F}_{k,t}\right) ^{2} \\
		&\lesssim &\frac{1}{p^{2}}\sum_{i_{k}=1}^{p_{k}}\mathbb{E}\left(
		\sum_{i_{1}=1}^{p_{1}}...\sum_{i_{k-1}=1}^{p_{k-1}}%
		\sum_{i_{k+1}=1}^{p_{k+1}}...\sum_{i_{K}=1}^{p_{K}}\sum_{h_{1}=1}^{p_{1}}...%
		\sum_{h_{j-1}=1}^{p_{j-1}}\sum_{h_{j+1}=1}^{p_{j+1}}...\sum_{h_{K}=1}^{p_{K}}%
		\mathbf{\zeta }_{i_{1},...,i_{K}}\mathbf{\zeta }_{h_{1},...,i_{j},...,h_{K}}%
		\right. \\
		&&\left. -\mathbb{E}\left( \mathbf{\zeta }_{i_{1},...,i_{K}}\mathbf{\zeta }%
		_{h_{1},...,i_{j},...,h_{K}}\right) \right) ^{2} \\
		&&+\frac{1}{p^{2}}\sum_{i_{k}=1}^{p_{k}}\mathbb{E}\left(
		\sum_{i_{1}=1}^{p_{1}}...\sum_{i_{k-1}=1}^{p_{k-1}}%
		\sum_{i_{k+1}=1}^{p_{k+1}}...\sum_{i_{K}=1}^{p_{K}}\sum_{h_{1}=1}^{p_{1}}...%
		\sum_{h_{j-1}=1}^{p_{j-1}}\sum_{h_{j+1}=1}^{p_{j+1}}...\sum_{h_{K}=1}^{p_{K}}%
		\mathbb{E}\left( \mathbf{\zeta }_{i_{1},...,i_{K}}\mathbf{\zeta }%
		_{h_{1},...,i_{j},...,h_{K}}\right) \right) ^{2} \\
		&=&\frac{1}{p^{2}}\sum_{i_{k}=1}^{p_{k}}\sum_{i_{1},i_{1}^{\prime
			}=1}^{p_{1}}...\sum_{i_{k-1},i_{k-1}^{\prime
			}=1}^{p_{k-1}}\sum_{i_{k+1},i_{k+1}^{\prime
			}=1}^{p_{k+1}}...\sum_{i_{K},i_{K}^{\prime
			}=1}^{p_{K}}\sum_{h_{1},h_{1}^{\prime
			}=1}^{p_{1}}...\sum_{h_{j-1},h_{j-1}^{\prime
			}=1}^{p_{j-1}}\sum_{h_{j+1},h_{j+1}^{\prime
			}=1}^{p_{j+1}}...\sum_{h_{K},h_{K}^{\prime }=1}^{p_{K}} \\
		&&Cov\left( \mathbf{\zeta }_{i_{1},...,i_{K}}\mathbf{\zeta }%
		_{h_{1},...,i_{j},...,h_{K}},\mathbf{\zeta }_{i_{1}^{\prime
			},...,i_{K}^{\prime }}\mathbf{\zeta }_{h_{1}^{\prime
			},...,i_{j},...,h_{K}^{\prime }}\right) \\
		&&+\frac{1}{p^{2}}\sum_{i_{k}=1}^{p_{k}}\mathbb{E}\left(
		\sum_{i_{1}=1}^{p_{1}}...\sum_{i_{k-1}=1}^{p_{k-1}}%
		\sum_{i_{k+1}=1}^{p_{k+1}}...\sum_{i_{K}=1}^{p_{K}}\sum_{h_{1}=1}^{p_{1}}...%
		\sum_{h_{j-1}=1}^{p_{j-1}}\sum_{h_{j+1}=1}^{p_{j+1}}...\sum_{h_{K}=1}^{p_{K}}%
		\mathbb{E}\left( \mathbf{\zeta }_{i_{1},...,i_{K}}\mathbf{\zeta }%
		_{h_{1},...,i_{j},...,h_{K}}\right) \right) ^{2} \\
		&\lesssim &\frac{1}{p_{j}}+\frac{1}{p_{k}}.
	\end{eqnarray*}%
	Hence, from the definition of $\mathcal{III}$ in (\ref{dec-m-hat}), by
	Theorem \ref{th:1} it holds that
	\begin{align}
		\mathcal{XIII}& \lesssim \frac{1}{p_{k}}\sum_{j=1,j\neq k}^{K}\left\Vert
		\frac{1}{Tp_{-k}}\sum_{s=1}^{T}\mathbf{E}_{k,s}\left( \left( \otimes _{{\ell
			}=j+1,{\ell }\neq k}^{K}\mathbf{A}_{\ell }\right) \otimes \left( \frac{1}{Tp}%
		\sum_{t=1}^{T}\mathbf{E}_{j,t}\mathbf{B}_{j}\mathbf{F}_{j,t}^{\top }\right)
		\otimes \left( \otimes _{{\ell }=1,{\ell }\neq k}^{j-1}\mathbf{A}_{\ell
		}\right) \right) \mathbf{F}_{k,s}^{\top }\right\Vert _{F}^{2}\times\left\Vert
		\mathbf{A}_{j}^{\top }\widehat{\mathbf{A}}_{j}\right\Vert _{F}^{2}
		\label{d13} \\
		& =O_{P}\left( \frac{1}{T^{2}p_{k}p_{-k}^{2}}\sum_{j=1,j\neq k}^{K}\left(
		p_{j}+\frac{p_{j}^{2}}{p_{k}}\right) \right) .  \notag
	\end{align}%
	We finally consider $\mathcal{XIV}$; recalling Lemma \ref{lemma1}\textit{%
		(iii)}, it holds that%
	\begin{align}
		& \mathcal{XIV}    \lesssim \frac{1}{p_{k}}\sum_{\substack{ j=1  \\ j\neq k}}^{K}\left\Vert
		\frac{1}{Tp_{-k}}\sum_{s=1}^{T}\mathbf{E}_{k,s}\left( \left( \otimes _{{\ell
			}=j+1,{\ell }\neq k}^{K}\mathbf{A}_{\ell }\right) \otimes \left( \frac{1}{Tp}%
		\sum_{t=1}^{T}\mathbf{E}_{j,t}\mathbf{E}_{j,t}^{\top }\widehat{\mathbf{A}}%
		_{j}\right) \otimes \left( \otimes _{{\ell }=1,{\ell }\neq k}^{j-1}\mathbf{A}%
		_{\ell }\right) \right) \mathbf{F}_{k,s}^{\top }\right\Vert _{F}^{2}  \label{d14}
		\\
		& \lesssim \frac{1}{p_{k}}\left\Vert \frac{1}{Tp_{-k}}\sum_{s=1}^{T}\mathbf{E%
		}_{k,s}\mathbf{F}_{k,s}\right\Vert _{F}^{2}\left( \sum_{j=1,j\neq k}^{K}%
		\left[ \left( \prod_{\ell =1,\ell \neq j,k}^{K}\left\Vert \mathbf{A}_{\ell
		}\right\Vert _{F}^{2}\right) \left( \left\Vert \frac{1}{Tp}\sum_{t=1}^{T}%
		\mathbf{E}_{j,t}\mathbf{E}_{j,t}^{\top }\mathbf{A}_{j}\right\Vert
		_{F}^{2}+\left\Vert \frac{1}{Tp}\sum_{t=1}^{T}\mathbf{E}_{j,t}\mathbf{E}%
		_{j,t}^{\top }\right\Vert _{F}^{2}\right.\right.\right.\notag\\
		&\qquad\qquad\qquad\qquad\qquad\left.\left.\left.\phantom{\sum_{k=1}^K}\left\Vert \widehat{\mathbf{A}}_{j}-%
		\mathbf{A}_{j}\widehat{\mathbf{H}}_{j}\right\Vert _{F}^{2}\right) \right]
		\right)  \notag \\
		& =O_{P}\left( \frac{1}{T^{2}p^{2}}\sum_{j=1,j\neq k}^{K}\left(
		p_{j}^{2}+Tp_{-j}^{2}\right) \right) .  \notag
	\end{align}%
	Indeed, to offer more detail on (\ref{d14}), we note that, using Assumption %
	\ref{as-4} and the same logic as above%
	\begin{eqnarray*}
		&&\mathbb{E}\left\Vert \frac{1}{Tp_{-k}}\sum_{s=1}^{T}\mathbf{E}_{k,s}%
		\mathbf{F}_{k,s}\right\Vert _{F}^{2} =\left( \frac{1}{Tp_{-k}}\right) ^{2}\sum_{i=1}^{p_{k}}\sum_{j=1}^{p_{-k}}%
		\mathbb{E}\left( \sum_{s=1}^{T}e_{k,t,ij}\mathbf{F}_{k,s}\right) ^{2}\leq
		c_{0}\frac{p_{k}}{Tp_{-k}},
	\end{eqnarray*}%
	whence%
	\begin{equation*}
		\left\Vert \frac{1}{Tp_{-k}}\sum_{s=1}^{T}\mathbf{E}_{k,s}\mathbf{F}%
		_{k,s}\right\Vert _{F}^{2}=O_{P}\left( \frac{p_{k}}{Tp_{-k}}\right) .
	\end{equation*}%
	Moreover, we know from Lemma \ref{lemma1}\textit{(iii)} that
	\begin{equation*}
		\left\Vert \frac{1}{Tp}\sum_{t=1}^{T}\mathbf{E}_{j,t}\mathbf{E}_{j,t}^{\top }%
		\mathbf{A}_{j}\right\Vert _{F}^{2}=O_{P}\left( \frac{1}{p_{k}}\right)
		+O_{P}\left( \frac{1}{Tp_{-k}}\right) ,
	\end{equation*}%
	and, by Assumption \ref{as-2}\textit{(i)}, that $\left\Vert \mathbf{A}_{\ell
	}\right\Vert _{F}^{2}=c_{0}p_{\ell }$. Noting that
	\begin{equation*}
		\prod_{\ell =1,\ell \neq j,k}^{K}p_{\ell }=\frac{p}{p_{j}p_{k}},
	\end{equation*}
	we have%
	\begin{eqnarray}
		&&\frac{1}{p_{k}}\left\Vert \frac{1}{Tp_{-k}}\sum_{s=1}^{T}\mathbf{E}_{k,s}%
		\mathbf{F}_{k,s}\right\Vert _{F}^{2}\left( \sum_{j=1,j\neq k}^{K}\left[
		\left( \prod_{\ell =1,\ell \neq j,k}^{K}\left\Vert \mathbf{A}_{\ell
		}\right\Vert _{F}^{2}\right) \left\Vert \frac{1}{Tp}\sum_{t=1}^{T}\mathbf{E}%
		_{j,t}\mathbf{E}_{j,t}^{\top }\mathbf{A}_{j}\right\Vert _{F}^{2}\right]
		\right)  \notag \\
		&=&\frac{1}{p_{k}}\frac{p_{k}}{Tp_{-k}}\left( \sum_{j=1,j\neq k}^{K}\left(
		\prod_{\ell =1,\ell \neq j,k}^{K}p_{\ell }\right) \left( \frac{1}{p_{j}}+%
		\frac{1}{Tp_{-j}}\right) \right) O_{P}\left( 1\right)  \notag \\
		&=&\frac{1}{Tp_{-k}}\left( \sum_{j=1,j\neq k}^{K}\left( \frac{p}{%
			p_{k}p_{j}^{2}}+\frac{1}{Tp_{k}}\right) \right) O_{P}\left( 1\right)
		=O_{P}\left( \frac{1}{pT^{2}}\right) +O_{P}\left( \sum_{j=1,j\neq k}^{K}%
		\frac{1}{p_{j}^{2}T}\right) .  \label{d14a}
	\end{eqnarray}%
	Also, using the facts that%
	\begin{eqnarray*}
		\left\Vert \frac{1}{Tp}\sum_{t=1}^{T}\mathbf{E}_{j,t}\mathbf{E}_{j,t}^{\top
		}\right\Vert _{F}^{2} =O_{P}\left( \frac{1}{Tp_{-j}}\right) +O_{P}\left(
		\frac{1}{p_{j}^{2}}\right) , \qquad\qquad 
		\left\Vert \widehat{\mathbf{A}}_{j}-\mathbf{A}_{j}\widehat{\mathbf{H}}%
		_{j}\right\Vert _{F}^{2} =O_{P}\left( p_{j}\widehat{w}_{j}\right) ,
	\end{eqnarray*}%
	we have%
	\begin{eqnarray}
		&&\frac{1}{p_{k}}\left\Vert \frac{1}{Tp_{-k}}\sum_{s=1}^{T}\mathbf{E}_{k,s}%
		\mathbf{F}_{k,s}\right\Vert _{F}^{2}\left( \sum_{j=1,j\neq k}^{K}\left[
		\left( \prod_{\ell =1,\ell \neq j,k}^{K}\left\Vert \mathbf{A}_{\ell
		}\right\Vert _{F}^{2}\right) \left\Vert \frac{1}{Tp}\sum_{t=1}^{T}\mathbf{E}%
		_{j,t}\mathbf{E}_{j,t}^{\top }\right\Vert _{F}^{2}\left\Vert \widehat{%
			\mathbf{A}}_{j}-\mathbf{A}_{j}\widehat{\mathbf{H}}_{j}\right\Vert _{F}^{2}%
		\right] \right) \notag  \\
		&=&\frac{1}{p_{k}}\frac{p_{k}}{Tp_{-k}}\left( \sum_{j=1,j\neq k}^{K}\left(
		\prod_{\ell =1,\ell \neq j,k}^{K}p_{\ell }\right) \left( \frac{1}{Tp_{-j}}+%
		\frac{1}{p_{j}^{2}}\right) p_{j}\widehat{w}_{j}\right) O_{P}\left( 1\right)
		=\frac{1}{T}\left( \sum_{j=1,j\neq k}^{K}\left( \frac{1}{Tp_{-j}}+\frac{1}{%
			p_{j}^{2}}\right) \widehat{w}_{j}\right) ,  \label{d14b}
	\end{eqnarray}%
	which is dominated. Then, (\ref{d14}) follows from combining (\ref{d14a})
	and (\ref{d14b}). The desired result obtains by combining (\ref{d12}), (\ref%
	{d13}) and (\ref{d14}), this showing that the sufficient condition (\ref%
	{equ:3.1})\textit{(ii)} is satisfied with%
	\begin{equation*}
		m_{-k}=\sum_{j=1,j\neq k}^{K}\left( \frac{1}{T^{2}p_{-j}^{2}}+\frac{1}{%
			Tp_{j}^{2}}\right) .
	\end{equation*}
\end{proof}

\begin{proof}[Proof of Corollary \protect\ref{cor:1}]
	The proof follows directly from Theorem \ref{th:2} and Lemma \ref{lem:1}.
\end{proof}

\begin{proof}[Proof of Theorem \protect\ref{th:5}]
	Based on equation (\ref{eq:Atildeapp}), for any $1\leq i\leq p_{k}$ it holds that%
	\begin{eqnarray*}
		\widetilde{\mathbf{A}}_{k,i}^{\top }-\widetilde{\mathbf{H}}_{k}\mathbf{A}%
		_{k,i}^{\top } &=&\frac{1}{Tp_{k}p_{-k}^{2}}\widetilde{\mathbf{\Lambda }}_{k}^{-1}%
		\widetilde{\mathbf{A}}_{k}^{\top }\sum_{t=1}^{T}\left( \mathbf{A}_{k}\mathbf{%
			F}_{k,t}\mathbf{B}_{k}^{\top }\widehat{\mathbf{B}}_{k}\widehat{\mathbf{B}}%
		_{k}^{\top }\mathbf{e}_{t,k,i\cdot }^{\top }+\mathbf{E}_{k,t}\widehat{%
			\mathbf{B}}_{k}\widehat{\mathbf{B}}_{k}^{\top }\mathbf{B}_{k}\mathbf{F}%
		_{k,t}^{\top }\mathbf{A}_{k,i\cdot }^{\top }+\mathbf{E}_{k,t}\widehat{%
			\mathbf{B}}_{k}\widehat{\mathbf{B}}_{k}^{\top }\mathbf{e}_{t,k,i\cdot
		}^{\top }\right) \\
		&=&\mathcal{XV}+\mathcal{XVI}+\mathcal{XVII}\text{.}
	\end{eqnarray*}%
	Using Theorem \ref{th:1}, it holds that%
	\begin{eqnarray*}
		\mathcal{XV} &=&\widetilde{\mathbf{\Lambda }}_{k}^{-1}\frac{\widetilde{\mathbf{A}}%
			_{k}^{\top }\mathbf{A}_{k}}{p_{k}}\frac{1}{Tp_{-k}}\sum_{t=1}^{T}\mathbf{F}%
		_{k,t}\frac{\mathbf{B}_{k}^{\top }\widehat{\mathbf{B}}_{k}}{p_{-k}}\left(
		\widehat{\mathbf{B}}_{k}-\mathbf{B}_{k}\widehat{\mathbf{H}}_{-k}+\mathbf{B}%
		_{k}\widehat{\mathbf{H}}_{-k}\right) ^{\top }\mathbf{e}_{t,k,i\cdot }^{\top }
		\\
		&=&\widetilde{\mathbf{\Lambda }}_{k}^{-1}\frac{\widetilde{\mathbf{A}}%
			_{k}^{\top }\mathbf{A}_{k}}{p_{k}}\frac{1}{Tp_{-k}}\sum_{t=1}^{T}\mathbf{F}%
		_{k,t}\left( \mathbf{I}_{r_{-k}}+O_{P}\left( w_{k}^{1/2}\right) \right)
		\mathbf{B}_{k}^{\top }\mathbf{e}_{t,k,i\cdot }^{\top } +\widetilde{\mathbf{\Lambda }}_{k}^{-1}\frac{\widetilde{\mathbf{A}}%
			_{k}^{\top }\mathbf{A}_{k}}{p_{k}}\frac{\mathbf{B}_{k}^{\top }\widehat{%
				\mathbf{B}}_{k}}{p_{-k}}\frac{1}{Tp_{-k}}\sum_{t=1}^{T}\mathbf{F}%
		_{k,t}\left( \widehat{\mathbf{B}}_{k}-\mathbf{B}_{k}\widehat{\mathbf{H}}%
		_{-k}\right) ^{\top }\mathbf{e}_{t,k,i\cdot }^{\top }
	\end{eqnarray*}%
	By Lemma \ref{lem:1}, it follows that
	\begin{equation*}
		\left\Vert \frac{1}{Tp_{-k}}\sum_{t=1}^{T}\mathbf{F}_{k,t}\left( \widehat{%
			\mathbf{B}}_{k}-\mathbf{B}_{k}\widehat{\mathbf{H}}_{-k}\right) ^{\top }%
		\boldsymbol{e}_{k,t,i\cdot }^{\top }\right\Vert _{F}={O}_{{P}}\left(
		\sum_{j=1,j\neq k}^{K}\left( \frac{1}{p_{j}\sqrt{T}}+\frac{1}{Tp_{-j}}%
		\right) \right) .
	\end{equation*}%
	Hence, using Theorems \ref{th:1} and \ref{th:2}, equation (\ref%
	{eq:LambdaTilde}) and noting that, by Lemma \ref{lemma1}\textit{(ii)}%
	\begin{equation*}
		\widetilde{\mathbf{\Lambda }}_{k}^{-1}\frac{\widetilde{\mathbf{A}}_{k}^{\top
			}\mathbf{A}_{k}}{p_{k}}\frac{1}{Tp_{-k}}\sum_{t=1}^{T}\mathbf{F}_{k,t}%
		\mathbf{B}_{k}^{\top }\boldsymbol{e}_{k,t,i\cdot }^{\top }={O}_{{P}}\left(
		\frac{1}{\sqrt{Tp_{-k}}}\right) ,
	\end{equation*}
	it follows that%
	\begin{equation}
		\mathcal{XV}=\widetilde{\mathbf{\Lambda }}_{k}^{-1}\frac{\widetilde{\mathbf{A%
			}}_{k}^{\top }\mathbf{A}_{k}}{p_{k}}\frac{1}{Tp_{-k}}\sum_{t=1}^{T}\mathbf{F}%
		_{k,t}\mathbf{B}_{k}^{\top }\boldsymbol{e}_{k,t,i\cdot }^{\top }+{o}_{{P}%
		}\left( \frac{1}{\sqrt{Tp_{-k}}}\right) +{O}_{{P}}\left( \sum_{j=1,j\neq
			k}^{K}\left( \frac{1}{p_{j}\sqrt{T}}+\frac{1}{Tp_{-j}}\right) \right) .
		\label{d15}
	\end{equation}%
	Using now by Theorem \ref{th:1}, Assumption \ref{as-2}\textit{(i)} and (\ref%
	{eq:LambdaTilde}), we now turn to showing that%
	\begin{eqnarray}
		\left\Vert \mathcal{XVI}\right\Vert _{F}  &\lesssim &\left\Vert \frac{1}{Tp}\sum_{t=1}^{T}\widetilde{\mathbf{A}}%
		_{k}^{\top }\mathbf{E}_{k,t}\widehat{\mathbf{B}}_{k}\mathbf{F}_{k,t}^{\top
		}\right\Vert _{F}   \label{d16} \\
		&\lesssim &\left\Vert \widetilde{\mathbf{A}}_{k}-\mathbf{A}_{k}\widetilde{%
			\mathbf{H}}_{k}\right\Vert _{F}\left\Vert \frac{1}{Tp}\sum_{t=1}^{T}\mathbf{E%
		}_{k,t}\widehat{\mathbf{B}}_{k}\mathbf{F}_{k,t}^{\top }\right\Vert
		_{F}+\left\Vert \frac{1}{Tp}\sum_{t=1}^{T}\mathbf{A}_{k}^{\top }\mathbf{E}%
		_{k,t}\widehat{\mathbf{B}}_{k}\mathbf{F}_{k,t}^{\top }\right\Vert _{F}
		\notag \\
		&=&O_{P}\left( \frac{1}{\sqrt{Tp}}+\frac{1}{Tp_{-k}}+\sum_{j=1,j\neq
			k}^{K}\left( \frac{1}{Tp_{-j}}+\frac{1}{\sqrt{Tp_{k}}p_{j}}+\frac{1}{Tp_{j}%
			\sqrt{p_{-k}}}+\frac{1}{Tp_{j}^{2}}\right) \right) .  \notag
	\end{eqnarray}%
	In order to prove (\ref{d16}), we begin by noting that by Lemmas \ref{lemma1}%
	\textit{(ii)} and \ref{lem:1}, it follows that%
	\begin{eqnarray*}
		\left\Vert \sum_{t=1}^{T}\mathbf{E}_{k,t}\widehat{\mathbf{B}}_{k}\mathbf{F}%
		_{k,t}^{\top }\right\Vert _{F} &\lesssim& \left\Vert \sum_{t=1}^{T}\mathbf{E}_{k,t}\mathbf{B}_{k}\mathbf{F}%
		_{k,t}^{\top }\right\Vert _{F}+\left\Vert \sum_{t=1}^{T}\mathbf{E}%
		_{k,t}\left( \widehat{\mathbf{B}}_{k}-\mathbf{B}_{k}\widehat{\mathbf{H}}%
		_{-k}\right) \mathbf{F}_{k,t}^{\top }\right\Vert _{F} \\
		&=&O_{P}\left( Tp+\sum_{j=1,j\neq k}^{K}\left( p_{j}^{2}+Tp_{-j}^{2}\right)
		\right) .
	\end{eqnarray*}%
	We also note that%
	\begin{eqnarray*}
		&&\left\Vert \sum_{t=1}^{T}\mathbf{A}_{k}^{\top }\mathbf{E}_{k,t}\widehat{%
			\mathbf{B}}_{k}\mathbf{F}_{k,t}^{\top }\right\Vert _{F}^{2} \lesssim \left\Vert \sum_{t=1}^{T}\mathbf{A}_{k}^{\top }\mathbf{E}_{k,t}%
		\mathbf{B}_{k}\widehat{\mathbf{H}}_{-k}\mathbf{F}_{k,t}^{\top }\right\Vert
		_{F}^{2}+\left\Vert \sum_{t=1}^{T}\mathbf{A}_{k}^{\top }\mathbf{E}%
		_{k,t}\left( \widehat{\mathbf{B}}_{k}-\mathbf{B}_{k}\widehat{\mathbf{H}}%
		_{-k}\right) \mathbf{F}_{k,t}^{\top }\right\Vert _{F}^{2}.
	\end{eqnarray*}%
	Recalling that $\left\Vert \widehat{\mathbf{H}}_{-k}\right\Vert
	_{F}^{2}=O_{P}\left( 1\right) $, Assumption \ref{as-4}\textit{(i)} readily
	entails that%
	\begin{equation*}
		\left\Vert \sum_{t=1}^{T}\mathbf{A}_{k}^{\top }\mathbf{E}_{k,t}\mathbf{B}_{k}%
		\widehat{\mathbf{H}}_{-k}\mathbf{F}_{k,t}^{\top }\right\Vert
		_{F}^{2}=O_{P}\left( Tp\right) .
	\end{equation*}%
	Further (recall the notation in (\ref{dec-m-hat}))%
	\begin{eqnarray*}
		&&\left\Vert \sum_{t=1}^{T}\mathbf{A}_{k}^{\top }\mathbf{E}_{k,t}\left(
		\widehat{\mathbf{B}}_{k}-\mathbf{B}_{k}\widehat{\mathbf{H}}_{-k}\right)
		\mathbf{F}_{k,t}^{\top }\right\Vert _{F}^{2} \\
		&\lesssim &\sum_{j=1,j\neq k}^{K}\left\Vert \sum_{t=1}^{T}\mathbf{A}%
		_{k}^{\top }\mathbf{E}_{k,t}\left( \left( \otimes _{{\ell }=j+1,{\ell }\neq
			k}^{K}\mathbf{A}_{\ell }\right) \otimes \left( \mathcal{II}\widehat{\mathbf{A%
		}}_{j}\right) \otimes \left( \otimes _{{\ell }=1,{\ell }\neq k}^{j-1}\mathbf{%
			A}_{\ell }\right) \right) \mathbf{F}_{k,t}^{\top }\right\Vert _{F}^{2} \\
		&&+\sum_{j=1,j\neq k}^{K}\left\Vert \sum_{t=1}^{T}\mathbf{A}_{k}^{\top }%
		\mathbf{E}_{k,t}\left( \left( \otimes _{{\ell }=j+1,{\ell }\neq k}^{K}%
		\mathbf{A}_{\ell }\right) \otimes \left( \mathcal{III}\widehat{\mathbf{A}}%
		_{j}\right) \otimes \left( \otimes _{{\ell }=1,{\ell }\neq k}^{j-1}\mathbf{A}%
		_{\ell }\right) \right) \mathbf{F}_{k,t}^{\top }\right\Vert _{F}^{2} \\
		&&+\sum_{j=1,j\neq k}^{K}\left\Vert \sum_{t=1}^{T}\mathbf{A}_{k}^{\top }%
		\mathbf{E}_{k,t}\left( \left( \otimes _{{\ell }=j+1,{\ell }\neq k}^{K}%
		\mathbf{A}_{\ell }\right) \otimes \left( \mathcal{IV}\widehat{\mathbf{A}}%
		_{j}\right) \otimes \left( \otimes _{{\ell }=1,{\ell }\neq k}^{j-1}\mathbf{A}%
		_{\ell }\right) \right) \mathbf{F}_{k,t}^{\top }\right\Vert _{F}^{2} \\
		&=&\mathcal{XVIII}+\mathcal{XIX}+\mathcal{XX}.
	\end{eqnarray*}%
	Upon inspecting the proof of Lemma \ref{lem:1}, it can be shown that%
	\begin{eqnarray*}
		\mathcal{XVIII} &\lesssim &\sum_{j=1,j\neq k}^{K}\left\Vert \sum_{t=1}^{T}%
		\mathbf{A}_{k}^{\top }\mathbf{E}_{k,t}\mathbf{B}_{k}\mathbf{F}_{k,t}^{\top
		}\right\Vert _{F}^{2}\left\Vert \frac{1}{Tp}\sum_{t=1}^{T}\mathbf{F}_{j,t}%
		\mathbf{B}_{j}^{\top }\mathbf{E}_{j,t}^{\top }\widehat{\mathbf{A}}_{j}%
		\mathbf{F}_{k,t}^{\top }\right\Vert _{F}^{2} =O_{P}\left( \sum_{j=1,j\neq k}^{K}\left( 1+\frac{p_{j}}{Tp_{-j}}\right)
		\right) ,
	\end{eqnarray*}%
	\begin{eqnarray*}
		\mathcal{XIX} &\lesssim &T^{2}p_{k}p_{-k}^{2}\left\Vert \mathbf{A}%
		_{k}\right\Vert _{F}^{2}\mathcal{XIII} =O_{P}\left( \sum_{j=1,j\neq k}^{K}\left( p_{j}^{2}+p_{j}p_{k}\right)
		\right) ,
	\end{eqnarray*}%
	and
	\begin{eqnarray*}
		\mathcal{XX} &\lesssim &\sum_{j=1,j\neq k}^{K}\left\Vert \sum_{t=1}^{T}\mathbf{A}%
		_{k}^{\top }\mathbf{E}_{k,t}\mathbf{F}_{k,t}\right\Vert _{F}^{2}\left(
		\prod_{\ell =1,\ell \neq j,k}^{K}\left\Vert \mathbf{A}_{\ell }\right\Vert
		_{F}^{2}\right) \bigg( \left\Vert \frac{1}{Tp}\sum_{t=1}^{T}\mathbf{E}_{j,t}%
		\mathbf{E}_{j,t}^{\top }\mathbf{A}_{j}\right\Vert _{F}^{2}+\left\Vert \frac{1%
		}{Tp}\sum_{t=1}^{T}\mathbf{E}_{j,t}\mathbf{E}_{j,t}^{\top }\right\Vert
		_{F}^{2}\left\Vert \widehat{\mathbf{A}}_{j}-\mathbf{A}_{j}\widehat{\mathbf{H}%
		}_{j}\right\Vert _{F}^{2}\bigg) \\
		&=&O_{P}\left( \sum_{j=1,j\neq k}^{K}\frac{1}{p_{k}}\left( Tp_{-j}^{2}+\frac{%
			p_{j}^{2}}{T}\right) \right) .
	\end{eqnarray*}%
	Putting all together, we have%
	\begin{equation*}
		\left\Vert \sum_{t=1}^{T}\mathbf{A}_{k}^{\top }\mathbf{E}_{k,t}\widehat{%
			\mathbf{B}}_{k}\mathbf{F}_{k,t}^{\top }\right\Vert _{F}^{2}=O_{P}\left(
		Tp+\sum_{j=1,j\neq k}^{K}\left( p_{j}^{2}+p_{j}p_{k}+\frac{Tp_{-j}^{2}}{p_{k}%
		}\right) \right) ,
	\end{equation*}%
	whence we finally obtain (\ref{d16}).
	
	Finally we consider $\mathcal{XVII}$. Following the same passages as in the
	proof of Lemma \ref{lemma5}, we obtain%
	\begin{eqnarray}
		\left\Vert \mathcal{XVI}\right\Vert _{F}  &\lesssim &\frac{1}{Tp}+\frac{1}{p^{2}}+\left( \frac{1}{p_{-k}^{2}}+\frac{%
			w_{-k}}{p_{-k}}\right) \widetilde{w}_{k}+\left( \frac{1}{Tp}+\frac{1}{p^{2}}%
		\right) w_{-k}+\left( \frac{1}{Tp_{k}}+\frac{1}{p_{k}^{2}}\right)
		w_{-k}^{2}+\left( \frac{1}{p_{-k}}+w_{-k}\right) w_{-k}\widetilde{w}_{k}\notag  \\
		&=&O_{P}\left( \frac{1}{p}\right) +O_{P}\left( \sum_{j=1,j\neq k}^{K}\frac{1%
		}{p_{k}p_{j}^{2}}\right) +o_{P}\left( \frac{1}{\sqrt{Tp_{-k}}}\right)
		+o_{P}\left( \frac{1}{Tp_{-j}}\right) +o_{P}\left( \frac{1}{\sqrt{Tp_{j}}}%
		\right) .  \label{d17}
	\end{eqnarray}%
	
	Now, consider the restriction
	\begin{equation}\label{eq:restriction_contradict}
		Tp_{-k}=o\left( \min \left(
		p^{2},T^{2}p_{-j}^{2},Tp_{j}^{2},p_{k}^{2}p_{j}^{4}\right) \right) ,
	\end{equation}
	for all $j\ne k$ and as $\min\{T, p_1,\ldots, p_K\}\to\infty$, and notice that if $K\ge 3$ this cannot be satisfied. Indeed,  if $K=3$, we would need $Tp_2p_3 = o(Tp_2^2)$ but also and $Tp_2p_3 = o(Tp_3^2)$, which would create a contradiction.
	
	Hence let us consider first the case $K\le 2$, and by combining (\ref{d15})-(\ref{d17}), as $\min \left\{ T,p_{1},\ldots
	,p_{K}\right\} \rightarrow \infty $ under the restriction \eqref{eq:restriction_contradict}
	%\begin{equation*}
	%Tp_{-k}=o\left( \min \left(
	%p^{2},T^{2}p_{-j}^{2},Tp_{j}^{2},p_{k}^{2}p_{j}^{4}\right) \right) ,
	%\end{equation*}
	we have
	\begin{equation*}
		\sqrt{Tp_{-k}}\left( \widetilde{\mathbf{A}}_{k,i\cdot }^{\top }-\widetilde{%
			\mathbf{H}}_{k}^{\top }\mathbf{A}_{k,i\cdot }^{\top }\right) =\widetilde{%
			\boldsymbol{\Lambda }}_{k}^{-1}\frac{\widetilde{\mathbf{A}}_{k}^{\top }%
			\mathbf{A}_{k}}{p_{k}}\frac{1}{\sqrt{Tp_{-k}}}\sum_{t=1}^{T}\mathbf{F}_{k,t}%
		\mathbf{B}_{k}^{\top }\boldsymbol{e}_{k,t,i\cdot }^{\top }+{o}_{{P}}({1}).
	\end{equation*}%
	We know from Lemma \ref{lemma4} that $\widetilde{\boldsymbol{\Lambda }}_{k}%
	\overset{p}{\rightarrow }\boldsymbol{\Lambda }_{k}$; further, recall that,
	by Assumption \ref{as-1}\textit{(ii)}, $\mathbf{\Sigma }_{k}$ has spectral
	decomposition $\mathbf{\Sigma }_{k}=\boldsymbol{\Gamma }_{k}\boldsymbol{%
		\Lambda }_{k}\boldsymbol{\Gamma }_{k}^{\top }$. Using the definition of $%
	\widetilde{\mathbf{H}}_{k}$ in (\ref{eq:Htildedefinition}), Theorems \ref%
	{th:1} and \ref{th:2} entail
%	\begin{equation*}
$		\widetilde{\mathbf{H}}_{k}=\boldsymbol{c}_{k}\boldsymbol{\Lambda }_{k}%
		\boldsymbol{\Gamma }_{k}^{\top }\widetilde{\mathbf{H}}_{k}\boldsymbol{%
			\Lambda }_{k}^{-1}+{o}_{{P}}({1}).$
%	\end{equation*}%
	By rearranging the last equation, and since $\boldsymbol{\Gamma }_{k}^{\top }%
	\boldsymbol{\Gamma }_{k}=\mathbf{I}_{r_{k}}$,
%	\begin{equation*}
$		\boldsymbol{\Gamma }_{k}^{\top }\tilde{\mathbf{H}}_{k}\boldsymbol{\Lambda }%
		_{k}=\boldsymbol{\Lambda }_{k}\boldsymbol{\Gamma }_{k}^{\top }\widetilde{%
			\mathbf{H}}_{k}+{o}_{{P}}({1}).$
%	\end{equation*}%
	Hence, $\boldsymbol{\Gamma }_{k}^{\top }\widetilde{\mathbf{H}}_{k}=\mathbf{I}%
	_{r_{k}}+{o}_{{P}}({1})$. Therefore, from Theorem \ref{th:2}, $p_{k}^{-1}%
	\widetilde{\mathbf{A}}_{k}^{\top }\mathbf{A}_{k}=\boldsymbol{\Gamma }%
	_{k}^{\top }+{o}_{{P}}({1})$. Using now Slutsky's theorem, and Assumption %
	\ref{as-5}\textit{(i)}, it readily follows that
	\begin{equation*}
		\sqrt{Tp_{-k}}\left( \widetilde{\mathbf{A}}_{k,i\cdot }^{\top }-\widetilde{%
			\mathbf{H}}_{k}^{\top }\mathbf{A}_{k,i\cdot }^{\top }\right) \overset{d}{%
			\rightarrow }\mathcal{N}\left( \mathbf{0},\boldsymbol{\Lambda }_{k}^{-1}%
		\boldsymbol{\Gamma }_{k}^{\top }\mathbf{V}_{ki}\boldsymbol{\Gamma }_{k}%
		\boldsymbol{\Lambda }_{k}^{-1}\right) ,
	\end{equation*}%
	which concludes the proof of part (i).
	
	Finally, consider the case in which $K\le 2$ and we impose the restriction
	\begin{equation}\label{eq:restriction_K3}
		Tp_{-k}\gtrsim \min \left\{
		p^{2},T^{2}p_{-j}^{2},Tp_{j}^{2},p_{k}^{2}p_{j}^{2}\right\} ,
	\end{equation}
	for all $j\ne k$ and as $\min\{T, p_1,\ldots, p_K\}\to\infty$, and notice also that if $K\ge 3$, then \eqref{eq:restriction_K3} always hold. Then, the proof of part (ii) follows from \eqref{d15}-\eqref{d17} under \eqref{eq:restriction_K3}. This completes the proof.
\end{proof}

\begin{proof}[Proof of Theorem \protect\ref{th:6}]
	By definition, it holds that%
	\begin{eqnarray}
		\widetilde{\mathcal{F}}_{t} &=&\frac{1}{p}\mathcal{X}_{t}\times _{1}%
		\widetilde{\mathbf{A}}_{1}^{\top }\times _{2}...\times _{K}\widetilde{%
			\mathbf{A}}_{K}^{\top }=\frac{1}{p}\left( \mathcal{F}_{t}\times _{1}\mathbf{A}_{1}\times
		_{2}...\times _{K}\mathbf{A}_{K}+\mathcal{E}_{t}\right) \times _{1}%
		\widetilde{\mathbf{A}}_{1}^{\top }\times _{2}...\times _{K}\widetilde{%
			\mathbf{A}}_{K}^{\top }\notag \\
		&=&\frac{1}{p}\mathcal{F}_{t}\times _{1}\widetilde{\mathbf{A}}_{1}^{\top
		}\left( \mathbf{A}_{1}-\widetilde{\mathbf{A}}_{1}\widetilde{\mathbf{H}}%
		_{1}^{-1}+\widetilde{\mathbf{A}}_{1}\widetilde{\mathbf{H}}_{1}^{-1}\right)
		\times _{2}...\times _{K}\widetilde{\mathbf{A}}_{K}^{\top }\left( \mathbf{A}%
		_{K}-\widetilde{\mathbf{A}}_{K}\widetilde{\mathbf{H}}_{K}^{-1}+\widetilde{%
			\mathbf{A}}_{K}\widetilde{\mathbf{H}}_{K}^{-1}\right) \notag \\
		&&+\frac{1}{p}\mathcal{E}_{t}\times _{1}\left( \widetilde{\mathbf{A}}_{1}-%
		\mathbf{A}_{1}\widetilde{\mathbf{H}}_{1}+\mathbf{A}_{1}\widetilde{\mathbf{H}}%
		_{1}\right) ^{\top }\times _{2}...\times _{K}\left( \widetilde{\mathbf{A}}%
		_{K}-\mathbf{A}_{K}\widetilde{\mathbf{H}}_{K}+\mathbf{A}_{K}\widetilde{%
			\mathbf{H}}_{K}\right) ^{\top }.\label{eq:AEFT}
	\end{eqnarray}%
	Then we have%
	\begin{eqnarray*}
		&&\left\Vert \widetilde{\mathcal{F}}_{t}-\mathcal{F}_{t}\times _{1}%
		\widetilde{\mathbf{H}}_{1}^{-1}\times _{2}...\times _{K}\widetilde{\mathbf{H}%
		}_{K}^{-1}\right\Vert _{F} \lesssim \sum_{k=1}^{K}\left\Vert \frac{1}{p_{k}}\widetilde{\mathbf{A}}%
		_{k}^{\top }\left( \mathbf{A}_{k}-\widetilde{\mathbf{A}}_{k}\widetilde{%
			\mathbf{H}}_{k}^{-1}\right) \mathbf{F}_{k,t}\right\Vert
		_{F}+\sum_{k=1}^{K}\left\Vert \frac{1}{p}\left( \widetilde{\mathbf{A}}_{k}-%
		\mathbf{A}_{k}\widetilde{\mathbf{H}}_{k}\right) ^{\top }\mathbf{E}_{k,t}%
		\mathbf{B}_{k}\right\Vert _{F}+\sum_{k=1}^{K}\left\Vert \frac{1}{p}\mathbf{A}%
		_{k}^{\top }\mathbf{E}_{k,t}\mathbf{B}_{k}\right\Vert _{F}.
	\end{eqnarray*}%
	The first term follows immediately by using Lemma \ref{lemma6}; similarly,
	it also follows that
	\begin{eqnarray*}
		\left\Vert \frac{1}{p}\left( \widetilde{\mathbf{A}}_{k}-\mathbf{A}_{k}%
		\widetilde{\mathbf{H}}_{k}\right) ^{\top }\mathbf{E}_{k,t}\mathbf{B}%
		_{k}\right\Vert _{F}^{2} &\leq &\frac{1}{p^{2}}\left\Vert \widetilde{\mathbf{A}}_{k}-\mathbf{A}_{k}%
		\widetilde{\mathbf{H}}_{k}\right\Vert _{F}^{2}\left\Vert \mathbf{E}_{k,t}%
		\mathbf{B}_{k}\right\Vert _{F}^{2} \\
		&\lesssim &\frac{1}{pp_{-k}}\widetilde{w}_{k}\left(
		\sum_{i=1}^{p_{k}}\sum_{h=1}^{r_{-k}}\left(
		\sum_{j=1}^{p_{-k}}e_{k,t,ij}B_{k,jh}\right) ^{2}\right)  \lesssim \frac{1}{pp_{-k}}\widetilde{w}_{k}\sum_{i=1}^{p_{k}}%
		\sum_{j_{1},j_{2}=1}^{p_{-k}}\left\vert
		e_{k,t,ij_{1}}e_{k,t,ij_{2}}\right\vert =O_{P}\left( 1\right) \frac{1}{p_{-k}%
		}\widetilde{w}_{k},
	\end{eqnarray*}%
	having used (\ref{weak-dep-4}) in the last passage. Also, it holds that%
	\begin{eqnarray*}
		\left\Vert \frac{1}{p}\mathbf{A}_{k}^{\top }\mathbf{E}_{k,t}\mathbf{B}%
		_{k}\right\Vert _{F}^{2} &=&\frac{1}{p^{2}}\sum_{h=1}^{r_{k}}\sum_{g=1}^{r_{-k}}\left(
		\sum_{i=1}^{p_{k}}\sum_{j=1}^{p_{-k}}A_{k,ih}B_{k,jg}e_{k,t,ij}\right) ^{2}
		\lesssim \frac{1}{p^{2}}\sum_{i_{1},i_{2}=1}^{p_{k}}%
		\sum_{j_{1},j_{2}=1}^{p_{-k}}\left\vert
		e_{k,t,ij_{1}}e_{k,t,ij_{2}}\right\vert =O_{P}\left( 1\right) \frac{1}{p},
	\end{eqnarray*}%
	again by (\ref{weak-dep-4}). Hence, we have%
	\begin{eqnarray}
		\left\Vert \widetilde{\mathcal{F}}_{t}-\mathcal{F}_{t}\times _{1}%
		\widetilde{\mathbf{H}}_{1}^{-1}\times _{2}...\times _{K}\widetilde{\mathbf{H}%
		}_{K}^{-1}\right\Vert _{F}&\lesssim& O_{P}\left( 1\right) \left( \frac{1}{\sqrt{p}}+\sum_{k=1}^{K}%
		\left( z_{k}+\sqrt{\frac{\widetilde{w}_{k}}{p_{-k}}}\right) \right) \notag \\
		&=&O_{P}\left( \frac{1}{\sqrt{p}}+\sum_{k=1}^{K}\left( \frac{1}{\sqrt{T}p_{k}%
		}+\frac{1}{\sqrt{Tp_{-k}}}\sum_{j=1,j\neq k}^K\frac{1}{p_{k}p_{j}^{2}}\right)
		\right) ,\label{eq:consFapp}
	\end{eqnarray}%
	where we have used the short-hand notation (based on Lemma \ref{lemma6})%
	\begin{eqnarray*}
		z_{k} &=&O_{P}\left( \frac{1}{\sqrt{Tp}}\right) +O_{P}\left( \frac{1}{p}%
		\right) +O_{P}\left( \sum_{j=1}^{K}\frac{1}{Tp_{-j}}\right)  +O_{P}\left( \sum_{j=1,j\neq k}^{K}\left( \frac{1}{Tp_{j}\sqrt{p_{-k}}}+%
		\frac{1}{p_{j}\sqrt{Tp_{k}}}+\frac{1}{p_{k}p_{j}^{2}}+\frac{1}{Tp_{j}^{2}}%
		\right) \right) .
	\end{eqnarray*}%
	
	Now, if for all $1\le k\le K$
	\begin{equation}
		p = o\left(\min\left\{ Tp_{-k}^2,Tp_k p_j^2, Tp_{-k} p_j^2,p_k^2p_j^4\right\}\right)\label{eq:restrictionF}
	\end{equation}
	for  all $j\ne k$ and as $\min\{T,p_1,\ldots, p_K\}\to\infty$, then from \eqref{eq:AEFT} we have
	\begin{align}
		\sqrt p \left(\widetilde{\mathcal F}_t-\mathcal F_t\times_{k=1}^K \widetilde{\mathbf H}_k^{-1}\right)=\frac 1{\sqrt p}\mathcal E_t \times_{k=1}^K \left(\mathbf A_k\widetilde{\mathbf H}_k\right)^\top+o_P(1),\nonumber
	\end{align}
	from which it follows that for all $1\le k\le K$
	\[
	\sqrt p\left(\widetilde{\mathbf F}_{t,k}-\widetilde {\mathbf H}_k^{-1} {\mathbf F}_{t,k} \left(\widetilde {\mathbf H}_{-k}^{-1}\right)^\top\right)=
	\frac 1{\sqrt p} \widetilde{\mathbf H}_k^\top{\mathbf A}_k^\top{\mathbf E}_{k,t}{\mathbf B}_k \widetilde{\mathbf H}_{-k}
	+o_P(1),
	\]
	where $\widetilde {\mathbf H}_{-k}=\otimes_{j\in[K]\slash \{k\}}\widetilde {\mathbf H}_j$ and so $\widetilde {\mathbf H}_{-k}^{-1}=\otimes_{j\in[K]\slash \{k\}}\widetilde {\mathbf H}_j^{-1}$. And therefore,
	\begin{equation}
		\sqrt p\left(\text{Vec}\left(\widetilde{\mathbf F}_{t,k}\right)-\left(\widetilde {\mathbf H}_{-k} \otimes \widetilde {\mathbf H}_{k}\right)^{-1} \text{Vec}\left({\mathbf F}_{t,k}\right)\right)
		= \frac 1{\sqrt p}\left(\widetilde {\mathbf H}_{-k} \otimes \widetilde {\mathbf H}_{k}\right)^\top \left(\mathbf B_k\otimes \mathbf A_k\right)^\top\text{Vec}\left(\mathbf E_{k,t} \right)+o_P(1).\nonumber
	\end{equation}
	The proof of part (i) follows from  Slutsky's theorem, and Assumption %
	\ref{as-5}\textit{(ii)}. The proof of part (ii) follows from \eqref{eq:consFapp} when condition \eqref{eq:restrictionF} does not hold.
\end{proof}

\begin{proof}[Proof of Theorem \protect\ref{atilde}]
	We begin by noting that%
	\begin{equation}
		\widetilde{\mathbf{H}}_{k}^{\ast }=\widetilde{\mathbf{\Lambda }}%
		_{k}^{-1}\left( \frac{\widetilde{\mathbf{A}}_{k}^{\top }\mathbf{A}_{k}}{p_{k}%
		}\right) \frac{1}{T}\sum_{t=1}^{T}\mathbf{F}_{k,t}\left( \frac{\mathbf{B}%
			_{k}^{\top }\widehat{\mathbf{B}}_{k}^{\ast }}{p_{-k}}\right) \left( \frac{%
			\mathbf{B}_{k}^{\top }\widehat{\mathbf{B}}_{k}^{\ast }}{p_{-k}}\right)
		^{\top }\mathbf{F}_{k,t}^{\top }.  \label{h-tilde-star}
	\end{equation}%
	The proof of the first part follows from showing that%
	\begin{equation*}
		\frac{1}{p_{-k}}\left\Vert \widehat{\mathbf{B}}_{k}^{\ast }-\mathbf{B}_{k}%
		\widehat{\mathbf{H}}_{-k}^{\ast }\right\Vert _{F}^{2}=O_{P}\left( \frac{1}{%
			Tp_{k}}\right) +O_{P}\left( \frac{1}{p_{-k}^{2}}\right) ,
	\end{equation*}%
	with $\widehat{\mathbf{H}}_{-k}^{\ast }=\otimes _{j=1,j\neq k}^{K}\widehat{%
		\mathbf{H}}_{k}^{\ast }$, using the same logic as in the remainder of the
	proofs, and subsequently from a direct application of Corollary 3.1 in \cite%
	{Yu2021Projected}. The second part also follows immediately, by
	the same passages as in the proof of Theorem \ref{th:5}, \textit{mutatis
		mutandis} using $\widehat{\mathbf{B}}_{k}^{\ast }$ instead of $\widehat{%
		\mathbf{B}}_{k}$. 
\end{proof}

\begin{proof}[Proof of Theorem \protect\ref{th:6bis}]
	Recall the notation
%	\begin{eqnarray*}
		$\widetilde{S}_{t,i_{1},...,i_{K}} =\widetilde{\mathcal{F}}_{t}\times _{1}%
		\widetilde{\mathbf{A}}_{1,i_{1}\cdot }\times _{2}...\times _{K}\widetilde{%
			\mathbf{A}}_{K,i_{K}\cdot },$ and $S_{t,i_{1},...,i_{K}} =\mathcal{F}_{t}\times _{1}\mathbf{A}_{1,i_{1}\cdot
		}\times _{2}...\times _{K}\mathbf{A}_{K,i_{K}\cdot }$
%	\end{eqnarray*}%
	We have from Theorems \ref{th:5} and \ref{th:6}
	\begin{eqnarray*}
		\left\vert \widetilde{S}_{t,i_{1},...,i_{K}}-S_{t,i_{1},...,i_{K}}\right\vert  &\lesssim &\sum_{k=1}^{K}\left\Vert \left( \widetilde{\mathbf{A}}%
		_{k,i_{k}\cdot }-\mathbf{A}_{k,i_{k}\cdot }\widetilde{\mathbf{H}}_{k}\right)
		^{\top }\widetilde{\mathbf{F}}_{k,t}\mathbf{B}_{k,j_{k}\cdot }\widetilde{%
			\mathbf{H}}_{-k}\right\Vert _{F} \\
		&&+\left\Vert \left( \widetilde{\mathcal{F}}_{t}-\mathcal{F}_{t}\times _{1}%
		\widetilde{\mathbf{H}}_{1}^{-1}\times _{2}...\times _{K}\widetilde{\mathbf{H}%
		}_{K}^{-1}\right) \times _{1}\left( \mathbf{A}_{1,i_{1}\cdot }\widetilde{%
			\mathbf{H}}_{1}\right) ^{\top }\times _{2}...\times _{K}\left( \mathbf{A}%
		_{K,i_{K}\cdot }\widetilde{\mathbf{H}}_{K}\right) ^{\top }\right\Vert _{F} \\
		&\lesssim &O_{P}\left( 1\right) \left( \sum_{k=1}^{K}\sqrt{\widetilde{w}_{k}}%
		+\frac{1}{\sqrt{p}}+\sum_{k=1}^{K}\left( z_{k}+\sqrt{\frac{\widetilde{w}_{k}%
			}{p_{-k}}}\right) \right)  =O_{P}\left( \frac{1}{\sqrt{p}}+\sum_{k=1}^{K}\left( \frac{1}{\sqrt{T}p_{k}%
		}+\frac{1}{\sqrt{Tp_{-k}}}\sum_{j=1,j\neq k}\frac{1}{p_{k}p_{j}^{2}}\right)
		\right) ,
	\end{eqnarray*}%
	whence the result obtains.
\end{proof}

\begin{proof}[Proof of Theorem \protect\ref{th:4}]
	Recall that, by Lemma \ref{lemma2}, it holds that, for all $j\leq r_{k}$, it
	holds that%
	\begin{equation}
		\lambda _{j}\left( \widehat{\mathbf{M}}_{k}\right) =\lambda _{j}\left(
		\boldsymbol{\Sigma }_{k}\right) +o_{P}(1),  \label{number-1}
	\end{equation}%
	with $\lambda _{j}\left( \boldsymbol{\Sigma }_{k}\right) >0$. Moreover,
	Lemma \ref{lemma3} and standard arguments based on Weyl's inequality entail
	that, for all $j>r_{k}$%
	\begin{equation}
		\lambda _{j}\left( \widehat{\mathbf{M}}_{k}\right) =O_{P}\left( \dfrac{1}{%
			\sqrt{Tp_{-k}}}+\dfrac{1}{p_{k}}\right) .  \label{number-2}
	\end{equation}%
	Hence, using elementary arguments, (\ref{number-1}) entails%
	\begin{equation*}
		\max_{1\leq j\leq r_{k}-1}\frac{\lambda _{j}\left( \widehat{\mathbf{M}}%
			_{k}\right) }{\lambda _{j+1}\left( \widehat{\mathbf{M}}_{k}\right) +\widehat{%
				c}\delta _{p_{1},...,p_{K},T}}\leq \max_{1\leq j\leq r_{k}-1}\frac{\lambda
			_{j}\left( \widehat{\mathbf{M}}_{k}\right) }{\lambda _{j+1}\left( \widehat{%
				\mathbf{M}}_{k}\right) }=O_{P}\left( 1\right) ,
	\end{equation*}%
	and, by (\ref{number-2}) and the definition of $\delta _{p_{1},...,p_{K},T}$%
	, it also follows that
	\begin{equation*}
		\max_{r_{k}+1\leq j\leq r_{\max }}\frac{\lambda _{j}\left( \widehat{\mathbf{M%
			}}_{k}\right) }{\lambda _{j+1}\left( \widehat{\mathbf{M}}_{k}\right) +%
			\widehat{c}\delta _{p_{1},...,p_{K},T}}\leq \max_{1\leq j\leq r_{k}-1}\frac{%
			\lambda _{j}\left( \widehat{\mathbf{M}}_{k}\right) }{\widehat{c}\delta
			_{p_{1},...,p_{K},T}}=O_{P}\left( 1\right) .
	\end{equation*}%
	Finally, combining (\ref{number-1}) and (\ref{number-2}), there exists a $%
	c_{0}>0$ such that%
	\begin{equation*}
		\mathbb{P}\left( \frac{\lambda _{r_{k}}\left( \widehat{\mathbf{M}}%
			_{k}\right) }{\lambda _{r_{k}+1}\left( \widehat{\mathbf{M}}_{k}\right) +%
			\widehat{c}\delta _{p_{1},...,p_{K},T}}\geq c_{0}\delta
		_{p_{1},...,p_{K},T}^{-1}\lambda _{r_{k}}\left( \widehat{\mathbf{M}}%
		_{k}\right) \right) =1.
	\end{equation*}%
	The final result obtains readily upon noting that
%	\begin{equation*}
		$\delta _{p_{1},...,p_{K},T}^{-1}\lambda _{r_{k}}\left( \widehat{\mathbf{M}}%
		_{k}\right) \overset{P}{\rightarrow }\infty ,$
%	\end{equation*}%
	as $\min \left\{ p_{1},...,p_{K},T\right\} \rightarrow \infty $. We conclude
	by showing that, when choosing $\widehat{c}$ and $\widetilde{c}$ according
	to (\ref{bt2}), it holds that%
	\begin{eqnarray}
		\widehat{c}^{L}+o_{P}\left( 1\right)  &\leq &\widehat{c}\leq \widehat{c}%
		^{U}+o_{P}\left( 1\right) ,  \label{ul-1} \\
		\widetilde{c}^{L}+o_{P}\left( 1\right)  &\leq &\widetilde{c}\leq \widetilde{c%
		}^{U}+o_{P}\left( 1\right) ,  \label{ul-2}
	\end{eqnarray}%
	which, in turn, automatically entails that the theorem holds when using $%
	\widehat{c}$ and $\widetilde{c}$. We focus only on (\ref{ul-1}) to save
	space, since the proof of (\ref{ul-2}) is basically the same. We begin by
	noting that%
%	\begin{equation*}
		$\widehat{c}\geq \lambda _{1}\left( \widehat{\mathbf{M}}_{k}\right)
		=c_{0}+o_{P}\left( 1\right) ,$
%	\end{equation*}%
	where $c_{0}>0$, as an immediate consequence of Lemma \ref{lemma2}. Also,
	setting $r_{k}=1$ for simplicity%
	\begin{eqnarray*}
		\sum_{j=1}^{p_{k}}\lambda _{j}\left( \widehat{\mathbf{M}}_{k}\right)  &=&%
		\frac{1}{Tp}\sum_{i=1}^{p_{k}}\sum_{j=1}^{p_{-k}}x_{k,t,ij}^{2} \leq 2\left( \frac{1}{Tp}\sum_{i=1}^{p_{k}}\sum_{j=1}^{p_{-k}}%
		\sum_{t=1}^{T}A_{ki}^{2}B_{kj}^{2}F_{k,t}^{2}+\frac{1}{Tp}%
		\sum_{i=1}^{p_{k}}\sum_{j=1}^{p_{-k}}\sum_{t=1}^{T}e_{k,t,ij}^{2}\right) .
	\end{eqnarray*}%
	It holds that
	\begin{eqnarray*}
		&&\mathbb{E}\left( \frac{1}{Tp}\sum_{i=1}^{p_{k}}\sum_{j=1}^{p_{-k}}%
		\sum_{t=1}^{T}A_{ki}^{2}B_{kj}^{2}F_{k,t}^{2}\right)  \leq \frac{1}{Tp}\sum_{i=1}^{p_{k}}\sum_{j=1}^{p_{-k}}%
		\sum_{t=1}^{T}A_{ki}^{2}B_{kj}^{2}\mathbb{E}\left( F_{k,t}^{2}\right) \leq
		c_{0},
	\end{eqnarray*}%
	having used Assumptions \ref{as-1} and \ref{as-2}; further, by Assumption %
	\ref{as-3}\textit{(i)}, it follows immediately that $\sum_{i=1}^{p_{k}}%
	\sum_{j=1}^{p_{-k}}\sum_{t=1}^{T}e_{k,t,ij}^{2}=O_{P}\left( Tp\right) $,
	whence the upper bound for $\sum_{j=1}^{p_{k}}\lambda _{j}\left( \widehat{%
		\mathbf{M}}_{k}\right) $.
\end{proof}

\begin{proof}[Proof of Theorem \protect\ref{th:4-bis}]
	We begin by noting that, if $\hat{r}_{j}^{(s-1)}>r_{j}$ for all $j\in
	\lbrack K]\backslash k$, we can and we will assume $\hat{r}%
	_{j}^{(s-1)}=r_{j}+1$, without loss of generality because $r_{\max }$ is a
	constant. Hence, by definition, $\widehat{\mathbf{A}}_{j}^{(s)}=\left(
	\widehat{\mathbf{A}}_{j},\widehat{\gamma }_{j}\right) $, where $\widehat{%
		\mathbf{A}}_{j}$ is the PCA estimator computed using the eigenvectors of $%
	\widehat{\mathbf{M}}_{k}^{\left( s\right) }$ corresponding to its $r_{j}$
	largest eigenvalues (using the true number of common factors $r_{j}$), and $%
	p_{j}^{-1/2}\widehat{\gamma }_{j}$ is the $\left( r_{j}+1\right) $-th
	eigenvector of $\widehat{\mathbf{M}}_{k}^{\left( s\right) }$; by
	construction, $\widehat{\mathbf{A}}_{j}^{\top }\widehat{\gamma }_{j}=\mathbf{%
		0}$. Therefore, it follows that%
	\begin{eqnarray*}
		\widetilde{\mathbf{M}}_{k}^{\left( s\right) } 
		&=&\frac{1}{Tp}\sum_{t=1}^{T}\mathbf{X}_{k,t}\widehat{\mathbf{B}}%
		_{k}^{\left( s\right) }\left( \widehat{\mathbf{B}}_{k}^{\left( s\right)
		}\right) ^{\top }\mathbf{X}_{k,t}^{\top } =\frac{1}{Tp}\sum_{t=1}^{T}\mathbf{X}_{k,t}\left( \otimes _{j\neq k}\left(
		\widehat{\mathbf{A}}_{j},\widehat{\gamma }_{j}\right) \right) \left( \otimes
		_{j\neq k}\left( \widehat{\mathbf{A}}_{j},\widehat{\gamma }_{j}\right)
		\right) ^{\top }\mathbf{X}_{k,t}^{\top } \\
		&=&\frac{1}{Tp}\sum_{t=1}^{T}\mathbf{X}_{k,t}\left( \otimes _{j\neq k}\left(
		\widehat{\mathbf{A}}_{j}\widehat{\mathbf{A}}_{j}^{\top }+\widehat{\gamma }%
		_{j}\widehat{\gamma }_{j}^{\top }\right) \right) \mathbf{X}_{k,t}^{\top } =\widetilde{\mathbf{M}}_{k}+\frac{1}{Tp}\sum_{t=1}^{T}\mathbf{X}%
		_{k,t}\left( \otimes _{j\neq k}\left( \widehat{\gamma }_{j}\widehat{\gamma }%
		_{j}^{\top }\right) \right) \mathbf{X}_{k,t}^{\top },
	\end{eqnarray*}%
	where the second matrix is non-negative definite. Using Weyl's inequality,
	it is easy to see that%
	\begin{equation*}
		\lambda _{j}\left( \widetilde{\mathbf{M}}_{k}^{\left( s\right) }\right) \geq
		\lambda _{j}\left( \widetilde{\mathbf{M}}_{k}\right) +\lambda _{\min }\left(
		\frac{1}{Tp}\sum_{t=1}^{T}\mathbf{X}_{k,t}\left( \otimes _{j\neq k}\left(
		\widehat{\gamma }_{j}\widehat{\gamma }_{j}^{\top }\right) \right) \mathbf{X}%
		_{k,t}^{\top }\right) \geq \lambda _{j}\left( \widetilde{\mathbf{M}}%
		_{k}\right) ,
	\end{equation*}%
	for all $j\leq r_{k}$. Thus, by Lemma \ref{lemma2}%
%	\begin{equation*}
		$\lambda _{j}\left( \widetilde{\mathbf{M}}_{k}^{\left( s\right) }\right) \geq
		\lambda _{j}\left( \mathbf{\Sigma }_{k}\right) +o_{P}\left( 1\right) ,$
%	\end{equation*}%
	for all $j\leq r_{k}$. Also note that, on account of Assumption \ref{as-2}%
	\textit{(ii)}, the matrix $\left( \widehat{\mathbf{A}}_{j}\widehat{\mathbf{A}%
	}_{j}^{\top }+\widehat{\gamma }_{j}\widehat{\gamma }_{j}^{\top }\right) $ is
	the idempotent, and therefore so is $\mathbf{I}_{r_{j}}-\left( \widehat{%
		\mathbf{A}}_{j}\widehat{\mathbf{A}}_{j}^{\top }+\widehat{\gamma }_{j}%
	\widehat{\gamma }_{j}^{\top }\right) $. Hence%
	\begin{equation*}
		\widehat{\mathbf{M}}_{k}-\widetilde{\mathbf{M}}_{k}^{\left( s\right) }=\frac{%
			1}{Tp}\sum_{t=1}^{T}\mathbf{X}_{k,t}\left( \otimes _{j\neq k}\left[ \mathbf{I%
		}_{r_{j}}-\left( \widehat{\mathbf{A}}_{j}\widehat{\mathbf{A}}_{j}^{\top }+%
		\widehat{\gamma }_{j}\widehat{\gamma }_{j}^{\top }\right) \right] \right)
		\mathbf{X}_{k,t}^{\top },
	\end{equation*}%
	is a non-negative definite matrix. Hence it is immediate to see that, for
	all $j$%
%	\begin{equation*}
		$\lambda _{j}\left( \widehat{\mathbf{M}}_{k}\right) -\lambda _{j}\left(
		\widetilde{\mathbf{M}}_{k}^{\left( s\right) }\right) \geq 0.$
%	\end{equation*}%
	Thus, by Lemma \ref{lemma3}, it follows that
	\begin{equation*}
		\lambda _{j}\left( \widetilde{\mathbf{M}}_{k}^{\left( s\right) }\right)
		=O_{P}\left( \dfrac{1}{\sqrt{Tp_{-k}}}+\dfrac{1}{p_{k}}\right) ,
	\end{equation*}%
	for all $j>r_{k}$. From hereon, the proof is the same as that of Theorem \ref%
	{th:4}.
\end{proof}

\section*{Acknowledgements}

He's work is supported by NSF China (12171282), Qilu and Xiaomi Young Scholars
Program of Shandong University.

%\bibliography{}
\bibliographystyle{myjmva}
%\begin{thebibliography}
\bibliography{Ref}
%\end{document}

\appendix

\section{Technical lemmas\label{lemmas}}

Throughout this appendix, $C_1,C_2,\ldots$ denote some positive constants
that do not depend on $p_1,p_2,T$.

We begin with a series of results which are a direct consequence of
Assumptions \ref{as-1}-\ref{as-4}, and which will be used throughout the
whole paper.

\begin{lemma}\label{lemma1}
	We assume that Assumptions \ref{as-1}-\ref{as-4} hold. Then it holds that
	
	(i)%
	\begin{eqnarray*}
		\sum_{t=1}^{T}\mathbb{E}\left\Vert \mathbf{E}_{k,t}\mathbf{A}_{k}\right\Vert
		_{F}^{2} &=&O\left( Tp\right) , \\
		\sum_{t=1}^{T}\mathbb{E}\left\Vert \mathbf{E}_{k,t}\Bb_{k}\right\Vert
		_{F}^{2} &=&O\left( Tp\right) ;
	\end{eqnarray*}
	
	(ii)
	\begin{eqnarray*}
		\mathbb{E}\left\Vert \sum_{t=1}^{T}\mathbf{F}_{k,t}\mathbf{B}_{k}^{\top }%
		\mathbf{E}_{k,t}^{\top }\right\Vert _{F}^{2} &=&O\left( Tp\right) , \\
		\mathbb{E}\left\Vert \sum_{t=1}^{T}\mathbf{F}_{k,t}^{\top }\mathbf{A}%
		_{k}^{\top }\mathbf{E}_{k,t}\right\Vert _{F}^{2} &=&O\left( Tp\right) , \\
		\mathbb{E}\left\Vert \sum_{t=1}^{T}\mathbf{F}_{k,t}\mathbf{B}_{k}^{\top }%
		\mathbf{E}_{k,t}^{\top }\mathbf{A}_{k}\right\Vert _{F}^{2} &=&O\left(
		Tp\right) , \\
		\mathbb{E}\left\Vert \sum_{t=1}^{T}\mathbf{F}_{k,t}^{\top }\mathbf{A}%
		_{k}^{\top }\mathbf{E}_{k,t}\mathbf{B}_{k}\right\Vert _{F}^{2} &=&O\left(
		Tp\right) ;
	\end{eqnarray*}
	
	(iii) for all $1\leq i_{k}\leq p_{k}$%
	\begin{eqnarray*}
		\mathbb{E}\left\Vert \sum_{t=1}^{T}\mathbf{E}_{k,t}\mathbf{e}_{t,k,i\cdot
		}^{\top }\right\Vert ^{2} &=&O\left( Tp\right) +O\left( \left(
		Tp_{-k}\right) ^{2}\right) , \\
		\mathbb{E}\left\Vert \sum_{t=1}^{T}\mathbf{A}_{k}^{\top }\mathbf{E}_{k,t}%
		\mathbf{e}_{t,k,i\cdot }^{\top }\right\Vert ^{2} &=&O\left( Tp\right)
		+O\left( \left( Tp_{-k}\right) ^{2}\right) ,
	\end{eqnarray*}%
	and%
	\begin{eqnarray*}
		\mathbb{E} \left\Vert \sum_{t=1}^{T}\mathbf{E}_{k,t}\mathbf{E}%
		_{k,t}^{\top }\right\Vert _{F}^{2}&=&O\left(
		Tp_{k}^{2}p_{-k}+T^{2}p_{k}p_{-k}^{2}\right) , \\
		\mathbb{E} \left\Vert \sum_{t=1}^{T}\mathbf{E}_{k,t}\mathbf{E}%
		_{k,t}^{\top }\mathbf{A}_{k}\right\Vert _{F}^{2} &=&O\left(
		Tp_{k}^{2}p_{-k}+T^{2}p_{k}p_{-k}^{2}\right) .
	\end{eqnarray*}
	
	\begin{proof}
		Given that $\left\{ r_{k},1\leq k\leq K\right\} $ are fixed, we show the
		results $r_{k}=1$ for all $1\le k\le K$, for simplicity and without loss
		of generality; thus, henceforth $\Ab_{k}$ is a $p_{k}$-dimensional vector
		with entries $A_{ki}$. We begin with part (i); using Assumptions \ref{as-2}%
		\textit{(i)} and \ref{as-3}\textit{(ii)}, and equations (\ref{weak-dep-1})-(%
		\ref{weak-dep-4}), it follows that
		\begin{equation*}
			\begin{aligned} \sum_{t=1}^{T} \mathbb{E}\left[\left\|\mathbf{E}_{k,t}^{\top}
				\Ab_k\right\|_{F}^{2}\right] &= \sum_{t=1}^T\EE\left[\sum_{j=1}^{p_{-k}}\left(\bm
				e_{k,t,\cdot j}^\top\Ab_k\right)^2\right]=\sum_{t=1}^T\sum_{j=1}^{p_{-k}}\EE\left[\left(%
				\sum_{i=1}^{p_k}e_{k,t,i j} A_{k,i}\right)^2\right]\\ &\leq
				\bar{a}_k^2\sum_{t=1}^T
				\sum_{j=1}^{p_{-k}}\sum_{i=1}^{p_k}\sum_{i_1=1}^{p_k}\EE\left[e_{k,t,i
					j}e_{k,t,i_1 j}\right] =O\left(Tp\right). \end{aligned}
		\end{equation*}%
		The second equation can be proven by similar passages. As far as part (ii)
		is concerned, it follows directly from Assumption \ref{as-4}\textit{(i)}.
		Finally, we consider part (iii). Using Assumptions \ref{as-3}\textit{%
			(ii)-(iii)} and (\ref{weak-dep-4}), we have
		\begin{eqnarray*}
			&&\mathbb{E}\left\Vert \sum_{t=1}^{T}\mathbf{E}_{k,t}\mathbf{e}_{t,k,i\cdot
			}^{\top }\right\Vert ^{2} \\
			&=&\mathbb{E}\left( \sum_{i_{1}=1}^{p_{k}}\left(
			\sum_{t=1}^{T}\sum_{j=1}^{p_{-k}}e_{t,k,ij}e_{t,k,i_{1}j}\right) ^{2}\right)
			\\
			&\leq &2\sum_{i_{1}=1}^{p_{k}}\mathbb{E}\left(
			\sum_{t=1}^{T}\sum_{j=1}^{p_{-k}}\left( e_{t,k,ij}e_{t,k,i_{1}j}-\mathbb{E}%
			\left( e_{t,k,ij}e_{t,k,i_{1}j}\right) \right) \right)
			^{2}+2\sum_{i_{1}=1}^{p_{k}}\left( \sum_{t=1}^{T}\sum_{j=1}^{p_{-k}}\mathbb{E%
			}\left( e_{t,k,ij}e_{t,k,i_{1}j}\right) \right) ^{2} \\
			&=&2\sum_{i_{1}=1}^{p_{k}}\sum_{t,s=1}^{T}\sum_{j_{1},j_{2}=1}^{p_{-k}}Cov%
			\left(
			e_{t,k,ij_{1}}e_{t,k,i_{1}j_{1}},e_{s,k,ij_{2}}e_{s,k,i_{1}j_{2}}\right)
			\\
			&&+2\sum_{i_{1}=1}^{p_{k}}\sum_{t,s=1}^{T}\sum_{j_{1},j_{2}=1}^{p_{-k}}\left\vert \mathbb{E}\left( e_{t,k,ij_{1}}e_{t,k,i_{1}j_{1}}\right) \right\vert
			\left\vert \mathbb{E}\left( e_{s,k,ij_{2}}e_{s,k,i_{1}j_{2}}\right)
			\right\vert \\
			&\leq
			&c_{0}Tp+c_{1}Tp_{-k}\sum_{i_{1}=1}^{p_{k}}\sum_{t=1}^{T}%
			\sum_{j_{1}=1}^{p_{-k}}\left\vert \mathbb{E}\left(
			e_{t,k,ij_{1}}e_{t,k,i_{1}j_{1}}\right) \right\vert \\
			&\leq &c_{0}Tp+c_{1}\left( Tp_{-k}\right) ^{2}.
		\end{eqnarray*}%
		Similarly, using Assumptions \ref{as-2}\textit{(i)} and \ref{as-3}\textit{%
			(ii)-(iii)}, and (\ref{weak-dep-4}), it follows that%
		\begin{eqnarray*}
			&&\mathbb{E}\left\Vert \sum_{t=1}^{T}\mathbf{A}_{k}^{\top }\mathbf{E}_{k,t}%
			\mathbf{e}_{t,k,i\cdot }^{\top }\right\Vert ^{2} \\
			&=&\mathbb{E}\left(
			\sum_{i_{1}=1}^{p_{k}}\sum_{t=1}^{T}%
			\sum_{j=1}^{p_{-k}}A_{k,i_{1}}e_{t,k,ij}e_{t,k,i_{1}j}\right) ^{2} \\
			&\leq &2\mathbb{E}\left(
			\sum_{i_{1}=1}^{p_{k}}\sum_{t=1}^{T}\sum_{j=1}^{p_{-k}}A_{k,i_{1}}\left(
			e_{t,k,ij}e_{t,k,i_{1}j}-\mathbb{E}\left( e_{t,k,ij}e_{t,k,i_{1}j}\right)
			\right) \right) ^{2}+2\left(
			\sum_{i_{1}=1}^{p_{k}}\sum_{t=1}^{T}\sum_{j=1}^{p_{-k}}A_{k,i_{1}}\mathbb{E}%
			\left( e_{t,k,ij}e_{t,k,i_{1}j}\right) \right) ^{2} \\
			&\leq
			&2\sum_{i_{1},i_{2}=1}^{p_{k}}\sum_{t,s=1}^{T}\sum_{j_{1},j_{2}=1}^{p_{-k}}%
			\left\vert A_{k,i_{1}}\right\vert \left\vert A_{k,i_{2}}\right\vert
			Cov\left(
			e_{t,k,ij_{1}}e_{t,k,i_{1}j_{1}},e_{s,k,ij_{2}}e_{s,k,i_{2}j_{2}}\right) \\
			&&+2\sum_{i_{1},i_{2}=1}^{p_{k}}\sum_{t,s=1}^{T}\sum_{j_{1},j_{2}=1}^{p_{-k}}\left\vert A_{k,i_{1}}\right\vert \left\vert A_{k,i_{2}}\right\vert \left\vert
			\mathbb{E}\left( e_{t,k,ij_{1}}e_{t,k,i_{1}j_{1}}\right) \right\vert
			\left\vert \mathbb{E}\left( e_{s,k,ij_{2}}e_{s,k,i_{2}j_{2}}\right)
			\right\vert \\
			&\leq &c_{0}Tp+c_{1}p_{k}\left( Tp_{-k}\right) ^{2}.
		\end{eqnarray*}%
		The proof of the remaining two result is essentially the same, and we omit
		it to save space.
	\end{proof}
\end{lemma}

We now consider a set of results pertaining to $\widehat{\mathbf{M}}_{k}$.
We will extensively use the following decomposition, which is a direct
consequence of (\ref{eq:Mhatk})%
\begin{eqnarray}
	\widehat{\mathbf{M}}_{k} &=&\frac{1}{Tp}\sum_{t=1}^{T}\mathbf{X}_{k,t}%
	\mathbf{X}_{k,t}^{\top }  \label{dec-m-hat} \\
	&=&\frac{1}{Tp}\sum_{t=1}^{T}\left( \mathbf{A}_{k}\mathbf{F}_{k,t}\mathbf{B}%
	_{k}^{\top }+\mathbf{E}_{k,t}\right) \left( \mathbf{A}_{k}\mathbf{F}_{k,t}%
	\mathbf{B}_{k}^{\top }+\mathbf{E}_{k,t}\right) ^{\top }  \notag \\
	&=&\frac{1}{Tp}\sum_{t=1}^{T}\mathbf{A}_{k}\mathbf{F}_{k,t}\mathbf{B}%
	_{k}^{\top }\mathbf{B}_{k}\mathbf{F}_{k,t}^{\top }\mathbf{A}_{k}^{\top }+%
	\frac{1}{Tp}\sum_{t=1}^{T}\mathbf{A}_{k}\mathbf{F}_{k,t}\mathbf{B}_{k}^{\top
	}\mathbf{E}_{k,t}^{\top }  \notag \\
	&&+\frac{1}{Tp}\sum_{t=1}^{T}\mathbf{E}_{k,t}\mathbf{B}_{k}\mathbf{F}%
	_{k,t}^{\top }\mathbf{A}_{k}^{\top }+\frac{1}{Tp}\sum_{t=1}^{T}\mathbf{E}%
	_{k,t}\mathbf{E}_{k,t}^{\top }  \notag \\
	&=&\mathcal{I}+\mathcal{II}+\mathcal{III}+\mathcal{IV}.  \notag
\end{eqnarray}

\begin{lemma}
	\label{lemma2}We assume that Assumptions \ref{as-1}-\ref{as-4} hold. Then it
	holds that, as $\min \left\{ T,p_{1},...,p_{K}\right\} \rightarrow \infty $%
	\begin{equation*}
		\lambda _{j}\left( \widehat{\mathbf{M}}_{k}\right) =\left\{
		\begin{array}{ll}
			\lambda _{j}\left( \mathbf{\Sigma }_{k}\right) +o_{P}\left( 1\right) & j\leq
			r_{k} \\
			O_{P}\left( \left( Tp_{-k}\right) ^{-1/2}\right) +O_{P}\left(
			p_{-k}^{-1}\right) & j>r_{k}%
		\end{array}%
		\right. .
	\end{equation*}
	
	\begin{proof}
		As shown in equation (\ref{dec-m-hat}), $\widehat{\mathbf{M}}_{k}=\mathcal{I}%
		+\mathcal{II}+\mathcal{III}+\mathcal{IV}$. We will study the spectral norms
		of these four terms and show that $\mathcal{I}$ is the dominant one.
		Firstly, by Assumptions \ref{as-1}\textit{(ii)} and \ref{as-2}\textit{(ii)},
		we have
		\begin{equation*}
			\begin{gathered} \frac{1}{T p_{-k}} \sum_{t=1}^T \mathbf{F}_{k,t}
				\Bb_k^{\top} \Bb_k\Fb_{kt}^{\top} \stackrel{P}{\rightarrow}
				\boldsymbol{\Sigma}_{k},\\ \cI\stackrel{P}{\rightarrow}\frac
				{\Ab_k\bSigma_k\Ab_k^{\top}}{p_k}, \end{gathered}
		\end{equation*}%
		while the leading $r_{k}$ eigenvalues of $p_{k}^{-1}\Ab_{k}\boldsymbol{%
			\Sigma }_{k}\Ab_{k}^{\top }$ are asymptotically equal to those of $%
		\boldsymbol{\Sigma }_{k}$. Hence, $\lambda _{j}(\mathcal{I})=\lambda
		_{j}\left( \mathbf{\Sigma }_{k}\right) +o_{P}(1)$ for $j\leq r_{k}$ while $%
		\lambda _{j}(\mathcal{I})=0$ for $j>r_{k}$ because $\mathop{\mathrm{rank}}(%
		\mathcal{I})\leq r_{k}$. Secondly, using the Cauchy-Schwartz inequality,
		Assumption \ref{as-2}\textit{(ii)} and Lemma \ref{lemma1}\textit{(ii)},
		\begin{eqnarray*}
			\left\Vert \mathcal{II}\right\Vert _{F} &\leq &\frac{1}{Tp}\left\Vert
			\mathbf{A}_{k}\right\Vert _{F}\left\Vert \sum_{t=1}^{T}\mathbf{F}_{k,t}%
			\mathbf{B}_{k}^{\top }\mathbf{E}_{k,t}^{\top }\right\Vert _{F} \\
			&\lesssim &\frac{1}{\sqrt{Tp_{-k}}}\left\Vert \frac{1}{\sqrt{Tp}}%
			\sum_{t=1}^{T}\mathbf{F}_{k,t}\mathbf{B}_{k}^{\top }\mathbf{E}_{k,t}^{\top
			}\right\Vert _{F} \\
			&=&O_{P}\left( \frac{1}{\sqrt{Tp_{-k}}}\right) .
		\end{eqnarray*}%
		Exactly the same logic also yields $\left\Vert \mathcal{III}\right\Vert
		_{F}=O_{P}\left( \left( Tp_{-k}\right) ^{-1/2}\right) $. Finally, consider $%
		\mathcal{IV}$, and let $\mathbf{U}_{\mathbf{E}}=\left( Tp\right)
		^{-1}\sum_{t=1}^{T}\mathbb{E}\left( \mathbf{E}_{k,t}\mathbf{E}_{k,t}^{\top
		}\right) $; then, using Assumption \ref{as-3}\textit{(iii)}, it follows that%
		\begin{eqnarray*}
			&&\mathbb{E}\left\Vert \mathcal{IV-}\mathbf{U}_{\mathbf{E}}\right\Vert
			_{F}^{2} \\
			&=&\frac{1}{\left( Tp\right) ^{2}}\sum_{i_{1},i_{2}=1}^{p_{k}}\mathbb{E}%
			\left( \sum_{t=1}^{T}\sum_{j=1}^{p_{-k}}\left( e_{t,k,i_{1}j}e_{t,k,i_{2}j}-%
			\mathbb{E}\left( e_{t,k,i_{1}j}e_{t,k,i_{2}j}\right) \right) \right) ^{2} \\
			&=&\frac{1}{\left( Tp\right) ^{2}}\sum_{i_{1},i_{2}=1}^{p_{k}}%
			\sum_{t,s=1}^{T}\sum_{j_{1},j_{2}=1}^{p_{-k}}Cov\left(
			e_{t,k,i_{1}j_{1}}e_{t,k,i_{2}j_{1}},e_{s,k,i_{1}j_{2}}e_{s,k,i_{2}j_{2}}%
			\right) =O\left( \frac{1}{Tp_{-k}}\right) .
		\end{eqnarray*}%
		Further, by Assumption \ref{as-3}\textit{(ii)} and (\ref{weak-dep-4})%
		\begin{eqnarray*}
			\left\Vert \mathbf{U}_{\mathbf{E}}\right\Vert _{1} &=&\left\Vert \mathbf{U}_{%
				\mathbf{E}}\right\Vert _{\infty } \\
			&=&\frac{1}{Tp}\max_{1\leq i\leq p_{k}}\sum_{i_{1}=1}^{p_{k}}\left\vert
			\sum_{t=1}^{T}\sum_{j=1}^{p_{-k}}\mathbb{E}\left(
			e_{t,k,i_{1}j}e_{t,k,i_{2}j}\right) \right\vert  \\
			&\leq &\max_{1\leq i\leq
				p_{k}}\sum_{i_{1}=1}^{p_{k}}\sum_{t=1}^{T}\sum_{j=1}^{p_{-k}}\left\vert
			\mathbb{E}\left( e_{t,k,i_{1}j}e_{t,k,i_{2}j}\right) \right\vert =O\left(
			p_{k}^{-1}\right) ,
		\end{eqnarray*}%
		which also entails that $\left\Vert \mathbf{U}_{\mathbf{E}}\right\Vert
		_{F}=O\left( p_{k}^{-1}\right) $. Hence it follows that $\left\Vert \mathcal{%
			IV}\right\Vert _{F}=O_{P}\left( \left( Tp_{-k}\right) ^{-1/2}\right)
		+O\left( p_{k}^{-1}\right) $. The proof now follows from (repeated
		applications of) Weyl's inequality.
	\end{proof}
\end{lemma}

\begin{lemma}
	\label{lemma3}We assume that Assumptions \ref{as-1}-\ref{as-4} hold. Then it
	holds that, as $\min \left\{ T,p_{1},...,p_{K}\right\} \rightarrow \infty $%
	\begin{eqnarray*}
		\frac{1}{p_{k}}\left\Vert \mathcal{II}\widehat{\mathbf{A}}_{k}\right\Vert
		_{F}^{2} &=&O_{P}\left( \frac{1}{Tp_{-k}}\right) , \\
		\frac{1}{p_{k}}\left\Vert \mathcal{III}\widehat{\mathbf{A}}_{k}\right\Vert
		_{F}^{2} &=&O_{P}\left( \frac{1}{Tp_{-k}}\right) , \\
		\frac{1}{p_{k}}\left\Vert \mathcal{IV}\widehat{\mathbf{A}}_{k}\right\Vert
		_{F}^{2} &=&O_{P}\left( \frac{1}{Tp}\right) +O_{P}\left( \frac{1}{p_{k}^{2}}%
		\right) +O_{P}\left( \frac{1}{Tp_{-k}}+\frac{1}{p_{k}}\right) \times \frac{1%
		}{p_{k}}\left\Vert \widehat{\mathbf{A}}_{k}-\mathbf{A}_{k}\widehat{\mathbf{H}%
		}_{k}\right\Vert _{F}^{2},
	\end{eqnarray*}%
	where recall that, by (\ref{dec-m-hat}), $\widehat{\mathbf{M}}_{k}=\mathcal{I%
	}+\mathcal{II}+\mathcal{III}+\mathcal{IV}$.
	
	\begin{proof}
		We begin by noting that, from equation (\ref{dec-m-hat}) and Lemma \ref%
		{lemma1}\textit{(ii)}, it follows that%
		\begin{eqnarray*}
			\frac{1}{p_{k}}\left\Vert \mathcal{II}\widehat{\mathbf{A}}_{k}\right\Vert
			_{F}^{2} &=&\frac{1}{p_{k}}\left\Vert \frac{1}{Tp}\sum_{t=1}^{T}\mathbf{A}%
			_{k}\mathbf{F}_{k,t}\mathbf{B}_{k}^{\top }\mathbf{E}_{k,t}^{\top }\widehat{%
				\mathbf{A}}_{k}\right\Vert _{F}^{2} \\
			&\leq &\frac{1}{p_{k}}\left\Vert \mathbf{A}_{k}\right\Vert
			_{F}^{2}\left\Vert \widehat{\mathbf{A}}_{k}\right\Vert _{F}^{2}\left\Vert
			\frac{1}{Tp}\sum_{t=1}^{T}\mathbf{F}_{k,t}\mathbf{B}_{k}^{\top }\mathbf{E}%
			_{k,t}^{\top }\right\Vert _{F}^{2} \\
			&=&O_{P}\left( \frac{1}{Tp_{-k}}\right) ,
		\end{eqnarray*}%
		having used Assumption \ref{as-2}\textit{(ii)} and the fact that, by
		construction, $\left\Vert \widehat{\mathbf{A}}_{k}\right\Vert
		_{F}^{2}=c_{0}p_{k}$. By the same token, it is easy to see that%
		\begin{eqnarray*}
			\frac{1}{p_{k}}\left\Vert \mathcal{III}\widehat{\mathbf{A}}_{k}\right\Vert
			_{F}^{2} &=&\frac{1}{p_{k}}\left\Vert \frac{1}{Tp}\sum_{t=1}^{T}\mathbf{E}%
			_{k,t}\mathbf{B}_{k}\mathbf{F}_{k,t}^{\top }\mathbf{A}_{k}^{\top }\widehat{%
				\mathbf{A}}_{k}\right\Vert _{F}^{2} \\
			&\leq &\frac{1}{p_{k}}\left\Vert \mathbf{A}_{k}\right\Vert
			_{F}^{2}\left\Vert \widehat{\mathbf{A}}_{k}\right\Vert _{F}^{2}\left\Vert
			\frac{1}{Tp}\sum_{t=1}^{T}\mathbf{E}_{k,t}\mathbf{B}_{k}\mathbf{F}%
			_{k,t}^{\top }\right\Vert _{F}^{2} \\
			&=&O_{P}\left( \frac{1}{Tp_{-k}}\right) .
		\end{eqnarray*}%
		Finally, using Lemma \ref{lemma1}\textit{(iii)} and recalling that $%
		\left\Vert \widehat{\mathbf{H}}_{k}\right\Vert _{F}^{2}=O_{P}\left( 1\right)
		$, it follows that%
		\begin{eqnarray*}
			&&\frac{1}{p_{k}}\left\Vert \mathcal{IV}\widehat{\mathbf{A}}_{k}\right\Vert
			_{F}^{2} \\
			&=&\frac{1}{p_{k}}\left\Vert \frac{1}{Tp}\sum_{t=1}^{T}\left( \widehat{%
				\mathbf{A}}_{k}-\mathbf{A}_{k}\widehat{\mathbf{H}}_{k}+\mathbf{A}_{k}%
			\widehat{\mathbf{H}}_{k}\right) ^{\top }\mathbf{E}_{k,t}\mathbf{E}%
			_{k,t}^{\top }\right\Vert _{F}^{2} \\
			&\leq &\frac{1}{p_{k}}\sum_{i=1}^{p_{k}}\left\Vert \frac{1}{Tp}\sum_{t=1}^{T}%
			\widehat{\mathbf{H}}_{k}^{\top }\mathbf{A}_{k}^{\top }\mathbf{E}%
			_{k,t}\mathbf{e}_{t,k,i\cdot }+\frac{1}{Tp}\sum_{t=1}^{T}\left( \widehat{%
				\mathbf{A}}_{k}-\mathbf{A}_{k}\widehat{\mathbf{H}}_{k}\right) ^{\top }%
			\mathbf{E}_{k,t}\mathbf{e}_{t,k,i\cdot }\right\Vert _{F}^{2} \\
			&\leq &O_{P}\left( 1\right) \frac{2}{p_{k}}\sum_{i=1}^{p_{k}}\left\Vert
			\frac{1}{Tp}\sum_{t=1}^{T}\mathbf{A}_{k}^{\top }\mathbf{E}_{k,t}%
			\mathbf{e}_{t,k,i\cdot }\right\Vert _{F}^{2}+\frac{2}{p_{k}}%
			\sum_{i=1}^{p_{k}}\left\Vert \frac{1}{Tp}\sum_{t=1}^{T}\left( \widehat{%
				\mathbf{A}}_{k}-\mathbf{A}_{k}\widehat{\mathbf{H}}_{k}\right) ^{\top }%
			\mathbf{E}_{k,t}\mathbf{e}_{t,k,i\cdot }\right\Vert _{F}^{2} \\
			&=&O_{P}\left( \frac{1}{Tp}+\frac{1}{p_{k}^{2}}\right) +O_{P}\left( \frac{1}{%
				Tp_{-k}}+\frac{1}{p_{k}}\right) \times \frac{1}{p_{k}}\left\Vert \widehat{%
				\mathbf{A}}_{k}-\mathbf{A}_{k}\widehat{\mathbf{H}}_{k}\right\Vert _{F}^{2},
		\end{eqnarray*}%
		having used Assumption \ref{as-3}.
	\end{proof}
\end{lemma}

We now report a set of results concerning $\widetilde{\mathbf{M}}_{k}$. By
definition it holds that%
\begin{eqnarray}
	\widetilde{\mathbf{M}}_{k} &=&\frac{1}{Tp_{k}p_{-k}^{2}}\sum_{t=1}^{T}%
	\mathbf{X}_{k,t}\widehat{\mathbf{B}}_{k}\widehat{\mathbf{B}}_{k}^{\top }%
	\mathbf{X}_{k,t}^{\top }  \label{dec-m-tilde} \\
	&=&\frac{1}{Tp_{k}p_{-k}^{2}}\sum_{t=1}^{T}\left( \mathbf{A}_{k}\mathbf{F}%
	_{k,t}\mathbf{B}_{k}^{\top }+\mathbf{E}_{k,t}\right) \widehat{\mathbf{B}}_{k}%
	\widehat{\mathbf{B}}_{k}^{\top }\left( \mathbf{A}_{k}\mathbf{F}_{k,t}\mathbf{%
		B}_{k}^{\top }+\mathbf{E}_{k,t}\right) ^{\top }  \notag \\
	&=&\frac{1}{Tpp_{-k}}\sum_{t=1}^{T}\mathbf{A}_{k}\mathbf{F}_{k,t}\mathbf{B}%
	_{k}^{\top }\widehat{\mathbf{B}}_{k}\widehat{\mathbf{B}}_{k}^{\top }\mathbf{B%
	}_{k}\mathbf{F}_{k,t}^{\top }\mathbf{A}_{k}^{\top }+\frac{1}{Tpp_{-k}}%
	\sum_{t=1}^{T}\mathbf{E}_{k,t}\widehat{\mathbf{B}}_{k}\widehat{\mathbf{B}}%
	_{k}^{\top }\mathbf{B}_{k}\mathbf{F}_{k,t}^{\top }\mathbf{A}_{k}^{\top } 
	\notag \\
	&&+\frac{1}{Tpp_{-k}}\sum_{t=1}^{T}\mathbf{A}_{k}\mathbf{F}_{k,t}\mathbf{B}%
	_{k}^{\top }\widehat{\mathbf{B}}_{k}\widehat{\mathbf{B}}_{k}^{\top }\mathbf{E%
	}_{k,t}^{\top }+\frac{1}{Tpp_{-k}}\sum_{t=1}^{T}\mathbf{E}_{k,t}\widehat{\mathbf{B}%
	}_{k}\widehat{\mathbf{B}}_{k}^{\top }\mathbf{E}_{k,t}^{\top }  \notag \\
	&=&\mathcal{V}+\mathcal{VI}+\mathcal{VII}+\mathcal{VIII}.  \notag
\end{eqnarray}

\begin{lemma}
	\label{lemma4}We assume that Assumptions \ref{as-1}-\ref{as-4} and (\ref%
	{equ:3.1}) hold. Then, as $\min \left\{ T,p_{1},...,p_{K}\right\}
	\rightarrow \infty $, it holds that%
	\begin{equation*}
		\lambda _{j}\left( \widetilde{\mathbf{M}}_{k}\right) =\lambda _{j}\left(
		\mathbf{\Sigma }_{k}\right) +o_{P}\left( 1\right) ,
	\end{equation*}%
	for all $j\leq r_{k}$.
	
	\begin{proof}
		Recall that, from equation (\ref{dec-m-tilde}), it holds that
		\begin{equation*}
			\widetilde{\mathbf{M}}_{k}=\mathcal{V}+\mathcal{VI}+\mathcal{VII}+\mathcal{%
				VIII}.
		\end{equation*}%
		We have%
		\begin{eqnarray*}
			\mathcal{V} &\mathcal{=}&\frac{1}{Tp_{k}p_{-k}^{2}}\mathbf{A}_{k}\left(
			\sum_{t=1}^{T}\mathbf{F}_{k,t}\mathbf{B}_{k}^{\top }\widehat{\mathbf{B}}_{k}%
			\widehat{\mathbf{B}}_{k}^{\top }\mathbf{B}_{k}\mathbf{F}_{k,t}^{\top
			}\right) \mathbf{A}_{k}^{\top } \\
			&=&\frac{1}{Tp_{k}p_{-k}^{2}}\mathbf{A}_{k}\left( \sum_{t=1}^{T}\mathbf{F}%
			_{k,t}\mathbf{B}_{k}^{\top }\left( \widehat{\mathbf{B}}_{k}-\mathbf{B}_{k}%
			\widehat{\mathbf{H}}_{-k}+\mathbf{B}_{k}\widehat{\mathbf{H}}_{-k}\right)
			\left( \widehat{\mathbf{B}}_{k}-\mathbf{B}_{k}\widehat{\mathbf{H}}_{-k}+%
			\mathbf{B}_{k}\widehat{\mathbf{H}}_{-k}\right) ^{\top }\mathbf{B}_{k}\mathbf{%
				F}_{k,t}^{\top }\right) \mathbf{A}_{k}^{\top } \\
			&=&\frac{1}{Tp_{k}p_{-k}^{2}}\mathbf{A}_{k}\left( \sum_{t=1}^{T}\mathbf{F}%
			_{k,t}\mathbf{B}_{k}^{\top }\mathbf{B}_{k}\widehat{\mathbf{H}}_{-k}\widehat{%
				\mathbf{H}}_{-k}^{\top }\mathbf{B}_{k}^{\top }\mathbf{B}_{k}\mathbf{F}%
			_{k,t}^{\top }\right) \mathbf{A}_{k}^{\top } \\
			&&+\frac{1}{Tp_{k}p_{-k}^{2}}\mathbf{A}_{k}\left( \sum_{t=1}^{T}\mathbf{F}%
			_{k,t}\mathbf{B}_{k}^{\top }\mathbf{B}_{k}\widehat{\mathbf{H}}_{-k}\left(
			\widehat{\mathbf{B}}_{k}-\mathbf{B}_{k}\widehat{\mathbf{H}}_{-k}\right)
			^{\top }\mathbf{B}_{k}\mathbf{F}_{k,t}^{\top }\right) \mathbf{A}_{k}^{\top }
			\\
			&&+\frac{1}{Tp_{k}p_{-k}^{2}}\mathbf{A}_{k}\left( \sum_{t=1}^{T}\mathbf{F}%
			_{k,t}\mathbf{B}_{k}^{\top }\left( \widehat{\mathbf{B}}_{k}-\mathbf{B}_{k}%
			\widehat{\mathbf{H}}_{-k}\right) \widehat{\mathbf{H}}_{-k}^{\top }\mathbf{B}%
			_{k}^{\top }\mathbf{B}_{k}\mathbf{F}_{k,t}^{\top }\right) \mathbf{A}%
			_{k}^{\top } \\
			&&+\frac{1}{Tp_{k}p_{-k}^{2}}\mathbf{A}_{k}\left( \sum_{t=1}^{T}\mathbf{F}%
			_{k,t}\mathbf{B}_{k}^{\top }\left( \widehat{\mathbf{B}}_{k}-\mathbf{B}_{k}%
			\widehat{\mathbf{H}}_{-k}\right) \left( \widehat{\mathbf{B}}_{k}-\mathbf{B}%
			_{k}\widehat{\mathbf{H}}_{-k}\right) ^{\top }\mathbf{B}_{k}\mathbf{F}%
			_{k,t}^{\top }\right) \mathbf{A}_{k}^{\top }.
		\end{eqnarray*}%
		We begin by noting that
		\begin{equation*}
			\widehat{\mathbf{H}}_{-k}\widehat{\mathbf{H}}_{-k}^{\top }\overset{P}{%
				\rightarrow }\mathbf{I}_{r_{-k}};
		\end{equation*}%
		then, standard passages based on (\ref{equ:3.1}), Assumption \ref{as-1}%
		\textit{(ii)} and Weyl's inequality, we have
		\begin{equation*}
			\lambda _{j}(\mathcal{V})=\lambda _{j}\left( p_{k}^{-1}\Ab_{k}\boldsymbol{%
				\Sigma }_{k}\Ab_{k}^{\top }\right) +o_{p}(1),
		\end{equation*}%
		for $j\leq r_{k}$. Further, since $\mathop{\mathrm{rank}}(\mathcal{V})\leq
		r_{k}$, it follows immediately that $\lambda _{j}(\mathcal{V})=0$ for all $%
		j>r_{k}$. Also, using again (\ref{equ:3.1}) and Lemma \ref{lemma1}\textit{%
			(ii)}%
		\begin{eqnarray*}
			\left\Vert \mathcal{VI}\right\Vert _{F} &=&\frac{1}{Tp_{k}p_{-k}^{2}}%
			\left\Vert \sum_{t=1}^{T}\mathbf{E}_{k,t}\widehat{\mathbf{B}}_{k}\widehat{%
				\mathbf{B}}_{k}^{\top }\mathbf{B}_{k}\mathbf{F}_{k,t}^{\top }\mathbf{A}%
			_{k}^{\top }\right\Vert _{F} \\
			&\lesssim &\frac{1}{Tp_{k}p_{-k}}\left\Vert \sum_{t=1}^{T}\mathbf{E}_{k,t}%
			\widehat{\mathbf{B}}_{k}\mathbf{F}_{k,t}^{\top }\mathbf{A}_{k}^{\top
			}\right\Vert _{F} \\
			&\lesssim &\frac{1}{Tp_{k}p_{-k}}\left\Vert \sum_{t=1}^{T}\mathbf{E}%
			_{k,t}\left( \widehat{\mathbf{B}}_{k}-\mathbf{B}_{k}\widehat{\mathbf{H}}%
			_{-k}\right) \mathbf{F}_{k,t}^{\top }\right\Vert _{F}\left\Vert \mathbf{A}%
			_{k}\right\Vert _{F}+\frac{1}{Tp_{k}p_{-k}}\left\Vert \sum_{t=1}^{T}\mathbf{E%
			}_{k,t}\mathbf{B}_{k}\mathbf{F}_{k,t}^{\top }\mathbf{A}_{k}^{\top
			}\right\Vert _{F} \\
			&=&O_{P}\left( m_{-k}^{1/2}\right) +O_{P}\left( \left( Tp_{-k}\right)
			^{-1/2}\right) =o_{P}\left( 1\right) .
		\end{eqnarray*}%
		The same logic also yields%
		\begin{equation*}
			\left\Vert \mathcal{VII}\right\Vert _{F}=O_{P}\left( w_{-k}^{1/2}\right)
			+O_{P}\left( \left( Tp_{-k}\right) ^{-1/2}\right) =o_{P}\left( 1\right) .
		\end{equation*}%
		Finally, using (\ref{equ:3.1}) and Lemma \ref{lemma1}\textit{(i)}%
		\begin{eqnarray*}
			\left\Vert \mathcal{VIII}\right\Vert _{F} &=&\frac{1}{Tp_{k}p_{-k}^{2}}%
			\left\Vert \sum_{t=1}^{T}\mathbf{E}_{k,t}\widehat{\mathbf{B}}_{k}^{\top }%
			\widehat{\mathbf{B}}_{k}\mathbf{E}_{k,t}^{\top }\right\Vert _{F} \\
			&\lesssim &\frac{1}{Tp_{k}p_{-k}^{2}}\sum_{t=1}^{T}\left\Vert \mathbf{E}%
			_{k,t}\widehat{\mathbf{B}}_{k}\right\Vert _{F}^{2} \\
			&\lesssim &\frac{1}{Tp_{k}p_{-k}^{2}}\sum_{t=1}^{T}\left\Vert \mathbf{E}%
			_{k,t}\left( \widehat{\mathbf{B}}_{k}-\mathbf{B}_{k}\widehat{\mathbf{H}}%
			_{-k}\right) \right\Vert _{F}^{2}+\frac{1}{Tp_{k}p_{-k}^{2}}%
			\sum_{t=1}^{T}\left\Vert \mathbf{E}_{k,t}\mathbf{B}_{k}\widehat{\mathbf{H}}%
			_{-k}\right\Vert _{F}^{2} \\
			&\lesssim &\frac{1}{Tp_{k}p_{-k}^{2}}\sum_{t=1}^{T}\left\Vert \mathbf{E}%
			_{k,t}\right\Vert _{F}^{2}\left\Vert \widehat{\mathbf{B}}_{k}-\mathbf{B}_{k}%
			\widehat{\mathbf{H}}_{-k}\right\Vert _{F}^{2}+\frac{1}{Tp_{k}p_{-k}^{2}}%
			\sum_{t=1}^{T}\left\Vert \mathbf{E}_{k,t}\mathbf{B}_{k}\widehat{\mathbf{H}}%
			_{-k}\right\Vert _{F}^{2} \\
			&=&O_{P}\left( w_{-k}\right) +O_{P}\left( p_{-k}^{-1}\right) =o_{P}\left(
			1\right) .
		\end{eqnarray*}%
		The proof now follows from Weyl's inequality.
	\end{proof}
	
\end{lemma}

\begin{lemma}
	\label{lemma5}We assume that Assumptions \ref{as-1}-\ref{as-4} and (\ref%
	{equ:3.1}) hold. Then, as $\min \left\{ T,p_{1},...,p_{K}\right\}
	\rightarrow \infty $, it holds that%
	\begin{eqnarray*}
		\frac{1}{p_{k}}\left\Vert \mathcal{VI}\widetilde{\mathbf{A}}_{k}\right\Vert
		_{F}^{2} &=&O_{P}\left( m_{-k}\right) +O_{P}\left( \frac{1}{Tp_{-k}}\right) ,
		\\
		\frac{1}{p_{k}}\left\Vert \mathcal{VII}\widetilde{\mathbf{A}}_{k}\right\Vert
		_{F}^{2} &=&O_{P}\left( m_{-k}\right) +O_{P}\left( \frac{1}{Tp_{-k}}\right) ,
		\\
		\frac{1}{p_{k}}\left\Vert \mathcal{VIII}\widetilde{\mathbf{A}}%
		_{k}\right\Vert _{F}^{2} &=&O_{P}\left( \frac{1}{Tp}+\frac{1}{p^{2}}\right)
		+O_{P}\left( \left( \frac{1}{Tp_{k}}+\frac{1}{p_{k}^{2}}\right)
		w_{-k}^{2}\right) \\&&+O_{P}\Big( \frac{1}{p_{-k}^{2}}
		+\frac{w_{-k}}{p_{-k}}%
		+w_{-k}^{2}\Big) \times \frac{1}{p_{k}}\left\Vert \widetilde{\mathbf{A}}%
		_{k}-\mathbf{A}_{k}\widetilde{\mathbf{H}}_{k}\right\Vert _{F}^{2},
	\end{eqnarray*}%
	where recall that, by (\ref{dec-m-tilde}), $\widetilde{\mathbf{M}}_{k}=%
	\mathcal{V}+\mathcal{VI}+\mathcal{VII}+\mathcal{VIII}$.
	
	\begin{proof}
		We begin by noting that, by Lemma \ref{lemma1}\textit{(ii)} and (\ref%
		{equ:3.1})%
		\begin{eqnarray}
			\frac{1}{p_{k}}\left\Vert \mathcal{VI}\widetilde{\mathbf{A}}_{k}\right\Vert
			_{F}^{2} &=&\frac{1}{p_{k}}\left\Vert \frac{1}{Tp_{k}p_{-k}^{2}}%
			\sum_{t=1}^{T}\mathbf{E}_{k,t}\widehat{\mathbf{B}}_{k}\widehat{\mathbf{B}}%
			_{k}^{\top }\mathbf{B}_{k}\mathbf{F}_{k,t}^{\top }\mathbf{A}_{k}^{\top }%
			\widetilde{\mathbf{A}}_{k}\right\Vert _{F}^{2}  \label{5-1} \\
			&\lesssim &\frac{1}{p_{k}}\left\Vert \frac{1}{Tp_{-k}}\sum_{t=1}^{T}\mathbf{E%
			}_{k,t}\widehat{\mathbf{B}}_{k}\mathbf{F}_{k,t}^{\top }\right\Vert _{F}^{2}
			\notag \\
			&\lesssim &\frac{1}{p_{k}}\left\Vert \frac{1}{Tp_{-k}}\sum_{t=1}^{T}\mathbf{E%
			}_{k,t}\left( \widehat{\mathbf{B}}_{k}-\mathbf{B}_{k}\widehat{\mathbf{H}}%
			_{-k}\right) \mathbf{F}_{k,t}^{\top }\right\Vert _{F}^{2}+\frac{1}{p_{k}}%
			\left\Vert \frac{1}{Tp_{-k}}\sum_{t=1}^{T}\mathbf{E}_{k,t}\mathbf{B}_{k}%
			\widehat{\mathbf{H}}_{-k}\mathbf{F}_{k,t}^{\top }\right\Vert _{F}^{2}  \notag
			\\
			&=&O_{P}\left( m_{-k}\right) +O_{P}\left( \frac{1}{Tp_{-k}}\right) .  \notag
		\end{eqnarray}%
		By the same logic, it can be readily shown that%
		\begin{equation*}
			\frac{1}{p_{k}}\left\Vert \mathcal{VII}\widetilde{\mathbf{A}}_{k}\right\Vert
			_{F}^{2}=O_{P}\left( m_{-k}+\frac{1}{Tp_{-k}}\right) .
		\end{equation*}%
		We now consider the last statement in the lemma; it holds that%
		\begin{eqnarray*}
			\mathcal{VIII}\widetilde{\mathbf{A}}_{k} &=&\frac{1}{Tp_{k}p_{-k}^{2}}%
			\sum_{t=1}^{T}\mathbf{E}_{k,t}\widehat{\mathbf{B}}_{k}\widehat{\mathbf{B}}%
			_{k}^{\top }\mathbf{E}_{k,t}^{\top }\widetilde{\mathbf{A}}_{k} \\
			&=&\frac{1}{Tp_{k}p_{-k}^{2}}\sum_{t=1}^{T}\mathbf{E}_{k,t}\left( \widehat{%
				\mathbf{B}}_{k}-\mathbf{B}_{k}\widehat{\mathbf{H}}_{-k}+\mathbf{B}_{k}%
			\widehat{\mathbf{H}}_{-k}\right) \left( \widehat{\mathbf{B}}_{k}-\mathbf{B}%
			_{k}\widehat{\mathbf{H}}_{-k}+\mathbf{B}_{k}\widehat{\mathbf{H}}_{-k}\right)
			^{\top }\mathbf{E}_{k,t}^{\top }\widetilde{\mathbf{A}}_{k} \\
			&=&\frac{1}{Tp_{k}p_{-k}^{2}}\sum_{t=1}^{T}\mathbf{E}_{k,t}\mathbf{B}_{k}%
			\widehat{\mathbf{H}}_{-k}\widehat{\mathbf{H}}_{-k}^{\top }\mathbf{B}%
			_{k}^{\top }\mathbf{E}_{k,t}^{\top }\widetilde{\mathbf{A}}_{k} \\
			&&+\frac{1}{Tp_{k}p_{-k}^{2}}\sum_{t=1}^{T}\mathbf{E}_{k,t}\left( \widehat{%
				\mathbf{B}}_{k}-\mathbf{B}_{k}\widehat{\mathbf{H}}_{-k}\right) \widehat{%
				\mathbf{H}}_{-k}^{\top }\mathbf{B}_{k}^{\top }\mathbf{E}_{k,t}^{\top }%
			\widetilde{\mathbf{A}}_{k} \\
			&&+\frac{1}{Tp_{k}p_{-k}^{2}}\sum_{t=1}^{T}\mathbf{E}_{k,t}\widehat{\mathbf{B%
			}}_{k}\left( \widehat{\mathbf{B}}_{k}-\mathbf{B}_{k}\widehat{\mathbf{H}}%
			_{-k}\right) ^{\top }\mathbf{E}_{k,t}^{\top }\widetilde{\mathbf{A}}_{k} \\
			&=&\mathcal{VIII}_{a}\mathcal{+VIII}_{b}\mathcal{+VIII}_{c}.
		\end{eqnarray*}%
		We will report our calculations for $r_{k}=1$, for simplicity and without
		loss of generality, for $1\leq k\leq K$. Consider $\mathcal{VIII}_{a}$;
		since $\left\Vert \widehat{\mathbf{H}}_{-k}\right\Vert _{F}=O_{P}\left(
		1\right) $,\ it holds that%
		\begin{equation*}
			\left\Vert \sum_{t=1}^{T}\mathbf{E}_{k,t}\mathbf{B}_{k}\widehat{\mathbf{H}}%
			_{-k}\widehat{\mathbf{H}}_{-k}^{\top }\mathbf{B}_{k}^{\top }\mathbf{E}%
			_{k,t}^{\top }\widetilde{\mathbf{A}}_{k}\right\Vert _{F}^{2}=O_{P}\left(
			1\right) \left\Vert \sum_{t=1}^{T}\mathbf{E}_{k,t}\mathbf{B}_{k}\mathbf{B}%
			_{k}^{\top }\mathbf{E}_{k,t}^{\top }\widetilde{\mathbf{A}}_{k}\right\Vert
			_{F}^{2}.
		\end{equation*}%
		Consider now%
		\begin{align}
			& \mathbb{E}\left( \left\Vert \sum_{t=1}^{T}\mathbf{E}_{k,t}\mathbf{B}_{k}%
			\mathbf{B}_{k}^{\top }\mathbf{E}_{k,t}^{\top }\mathbf{A}_{k}\right\Vert
			_{F}^{2}\right)  \notag \\
			& =\left\Vert \mathbf{B}_{k}\right\Vert _{F}^{2}\mathbb{E}\left( \left\Vert
			\sum_{t=1}^{T}\mathbf{E}_{k,t}\mathbf{B}_{k}\mathbf{A}_{k}^{\top }\mathbf{E}%
			_{k,t}\right\Vert _{F}^{2}\right)  \notag \\
			& =\left\Vert \mathbf{B}_{k}\right\Vert
			_{F}^{2}\sum_{i=1}^{p_{k}}\sum_{j=1}^{p_{-k}}\mathbb{E}\left( \left\Vert
			\sum_{t=1}^{T}\mathbf{B}_{k}^{\top }\boldsymbol{e}_{k,t,i\cdot }^{\top }%
			\boldsymbol{e}_{k,t,\cdot j}^{\top }\mathbf{A}_{k}\right\Vert _{F}^{2}\right)
			\notag \\
			& \lesssim p_{-k}^{2}p_{k}\mathbb{E}\left( \left\Vert \sum_{t=1}^{T}\mathbf{B%
			}_{k}^{\top }\boldsymbol{e}_{k,t,i\cdot }^{\top }\boldsymbol{e}_{k,t,\cdot
				j}^{\top }\mathbf{A}_{k}\right\Vert _{F}^{2}\right) .  \label{equ:b.0}
		\end{align}%
		Given that Assumption \ref{as-2}\textit{(ii)} entails that $\left\Vert
		\mathbf{B}_{k}\right\Vert _{F}^{2}=c_{0}p_{-k}$, it holds that%
		\begin{align}
			& \mathbb{E}\left( \left\Vert \sum_{t=1}^{T}\mathbf{B}_{k}^{\top }%
			\boldsymbol{e}_{k,t,i\cdot }^{\top }\boldsymbol{e}_{k,t,\cdot j}^{\top }%
			\mathbf{A}_{k}\right\Vert _{F}^{2}\right)  \notag \\
			& \leq \mathbb{E}\left( \left\Vert \mathbf{B}_{k}^{\top
			}\sum_{t=1}^{T}\left( \boldsymbol{e}_{k,t,i\cdot }^{\top }\boldsymbol{e}%
			_{k,t,\cdot j}^{\top }-\mathbb{E}\left( \boldsymbol{e}_{k,t,i\cdot }^{\top }%
			\boldsymbol{e}_{k,t,\cdot j}^{\top }\right) \right) \mathbf{A}%
			_{k}\right\Vert _{F}^{2}\right) +\left\Vert \sum_{t=1}^{T}\mathbb{E}\left(
			\mathbf{B}_{k}^{\top }\boldsymbol{e}_{k,t,i\cdot }^{\top }\boldsymbol{e}%
			_{k,t,\cdot j}^{\top }\mathbf{A}_{k}\right) \right\Vert _{F}^{2}  \notag \\
			\lesssim &
			\,\sum_{t,s=1}^{T}\sum_{i_{1},i_{2}=1}^{p_{k}}\sum_{j_{1},j_{2}=1}^{p_{-k}}%
			\left\vert Cov\left(
			e_{k,t,ij_{1}}e_{k,t,i_{1}j},e_{k,s,ij_{2}}e_{k,s,i_{2}j}\right) \right\vert
			+\left(
			\sum_{t=1}^{T}\sum_{i_{1}=1}^{p_{k}}\sum_{j_{1}=1}^{p_{-k}}\left\vert
			\mathbb{E}\left( e_{k,t,ij_{1}}e_{k,t,i_{1}j}\right) \right\vert \right) ^{2}
			\notag \\
			& =O\left( Tp+T^{2}\right) ,  \label{equ:b.2}
		\end{align}%
		for all $1\leq i\leq p_{k}$ and all $1\leq j\leq p_{-k}$, where we have used
		Assumptions \ref{as-2}\textit{(ii)}, \ref{as-3}\textit{(ii)}, and \ref{as-3}%
		\textit{(iii)}. Hence, from (\ref{equ:b.0}) and (\ref{equ:b.2}) and Lemma %
		\ref{lemma1}\textit{(i)}, and recalling that $\left\Vert \widetilde{\mathbf{H%
		}}_{k}\right\Vert _{F}^{2}=O_{P}\left( 1\right) $%
		\begin{align}
			& \frac{1}{p_{k}}\Vert \mathcal{VIII}_{a}\Vert _{F}^{2}  \notag \\
			& =\frac{1}{p_{k}}\left\Vert \frac{1}{Tp_{k}p_{-k}^{2}}\left( \sum_{t=1}^{T}%
			\mathbf{E}_{k,t}\mathbf{B}_{k}\widehat{\mathbf{H}}_{-k}\widehat{\mathbf{H}}%
			_{-k}^{\top }\mathbf{B}_{k}^{\top }\mathbf{E}_{k,t}^{\top }\mathbf{A}_{k}%
			\widetilde{\mathbf{H}}_{k}+\sum_{t=1}^{T}\mathbf{E}_{k,t}\mathbf{B}_{k}%
			\widehat{\mathbf{H}}_{-k}\widehat{\mathbf{H}}_{-k}^{\top }\mathbf{B}%
			_{k}^{\top }\mathbf{E}_{k,t}^{\top }\left( \widetilde{\mathbf{A}}_{k}-%
			\mathbf{A}_{k}\widetilde{\mathbf{H}}_{k}\right) \right) \right\Vert _{F}^{2}
			\notag \\
			& \lesssim \frac{1}{T^{2}p_{k}^{3}p_{-k}^{4}}\left( \left\Vert \sum_{t=1}^{T}%
			\mathbf{E}_{k,t}\mathbf{B}_{k}\mathbf{B}_{k}^{\top }\mathbf{E}_{k,t}^{\top }%
			\mathbf{A}_{k}\right\Vert _{F}^{2}+\left( \sum_{t=1}^{T}\left\Vert \mathbf{E}%
			_{k,t}\mathbf{B}_{k}\right\Vert _{F}^{2}\right) ^{2}\left\Vert \widetilde{%
				\mathbf{A}}_{k}-\mathbf{A}_{k}\widetilde{\mathbf{H}}_{k}\right\Vert
			_{F}^{2}\right)  \notag \\
			& =O_{P}\left( \frac{1}{Tp}+\frac{1}{p^{2}}\right) +O_{P}\left( \frac{1}{%
				p_{k}p_{-k}^{2}}\right) \times \left\Vert \widetilde{\mathbf{A}}_{k}-\mathbf{%
				A}_{k}\widetilde{\mathbf{H}}_{k}\right\Vert _{F}^{2}.  \label{eq:IXcal}
		\end{align}%
		We now turn to studying $\mathcal{VIII}_{b}$; it holds that%
		\begin{equation}
			\left\Vert \sum_{t=1}^{T}\mathbf{E}_{k,t}\left( \widehat{\mathbf{B}}_{k}-%
			\mathbf{B}_{k}\widehat{\mathbf{H}}_{-k}\right) \mathbf{B}_{k}^{\top }\mathbf{%
				E}_{k,t}^{\top }\mathbf{A}_{k}\right\Vert _{F}^{2}\leq
			\sum_{i=1}^{p_{k}}\sum_{j=1}^{p_{-k}}\left\Vert \sum_{t=1}^{T}e_{k,t,ij}%
			\mathbf{B}_{k}^{\top }\mathbf{E}_{k,t}^{\top }\mathbf{A}_{k}\right\Vert
			_{F}^{2}\left\Vert \widehat{\mathbf{B}}_{k}-\mathbf{B}_{k}\widehat{\mathbf{H}%
			}_{-k}\right\Vert _{F}^{2};  \label{eq:C8a}
		\end{equation}%
		using Assumptions \ref{as-2}\textit{(ii)}, \ref{as-3}\textit{(ii)}, and \ref%
		{as-3}\textit{(iii)}, and a similar logic as in the proof of (\ref{equ:b.2}%
		), it follows that for all $1\leq i\leq p_{k}$ and all $1\leq j\leq p_{-k}$,
		\begin{align}
			& \mathbb{E}\left( \left\Vert \sum_{t=1}^{T}e_{k,t,ij}\mathbf{B}_{k}^{\top }%
			\mathbf{E}_{k,t}^{\top }\mathbf{A}_{k}\right\Vert _{F}^{2}\right)  \notag \\
			& \lesssim
			\sum_{t,s=1}^{T}\sum_{i_{1},i_{2}=1}^{p_{k}}\sum_{j_{1},j_{2}=1}^{p_{-k}}%
			\left\vert Cov\left(
			e_{k,t,ij}e_{k,t,i_{1}j_{1}},e_{k,s,ij}e_{k,s,i_{2}j_{2}}\right) \right\vert
			\notag \\
			& +\left(
			\sum_{t=1}^{T}\sum_{i_{1}=1}^{p_{k}}\sum_{j_{1}=1}^{p_{-k}}\left\vert
			\mathbb{E}\left( e_{k,t,ij}e_{k,t,i_{1}j_{1}}\right) \right\vert \right) ^{2}
			\notag \\
			& =\,O\left( Tp+T^{2}\right) .  \label{eq:C8b}
		\end{align}%
		Hence, combining (\ref{eq:C8a}) and (\ref{eq:C8b}), and by the sufficient
		condition (\ref{equ:3.1}) and Lemma \ref{lemma1}\textit{(i)}, it follows
		that
		\begin{align}
			\frac{1}{p_{k}}\Vert \mathcal{VIII}_{b}\Vert _{F}^{2}& \leq +\frac{1}{p_{k}}%
			\left\Vert \frac{1}{Tp_{k}p_{-k}^{2}}\sum_{t=1}^{T}\mathbf{E}_{k,t}\left(
			\widehat{\mathbf{B}}_{k}-\mathbf{B}_{k}\widehat{\mathbf{H}}_{-k}\right)
			\widehat{\mathbf{H}}_{-k}^{\top }\mathbf{B}_{k}^{\top }\mathbf{E}%
			_{k,t}^{\top }\mathbf{A}_{k}\widetilde{\mathbf{H}}_{k}\right\Vert _{F}^{2}
			\notag \\
			& +\frac{1}{p_{k}}\left\Vert \frac{1}{Tp_{k}p_{-k}^{2}}\sum_{t=1}^{T}\mathbf{%
				E}_{k,t}\left( \widehat{\mathbf{B}}_{k}-\mathbf{B}_{k}\widehat{\mathbf{H}}%
			_{-k}\right) \widehat{\mathbf{H}}_{-k}^{\top }\mathbf{B}_{k}^{\top }\mathbf{E%
			}_{k,t}^{\top }\left( \widetilde{\mathbf{A}}_{k}-\mathbf{A}_{k}\widetilde{%
				\mathbf{H}}_{k}\right) \right\Vert _{F}^{2}  \notag \\
			\lesssim & \frac{1}{T^{2}p_{k}^{3}p_{-k}^{4}}\left\Vert \sum_{t=1}^{T}%
			\mathbf{E}_{k,t}\left( \widehat{\mathbf{B}}_{k}-\mathbf{B}_{k}\widehat{%
				\mathbf{H}}_{-k}\right) \mathbf{B}_{k}^{\top }\mathbf{E}_{k,t}^{\top }%
			\mathbf{A}_{k}\right\Vert _{F}^{2}  \notag \\
			& +\frac{1}{T^{2}p_{k}^{3}p_{-k}^{4}}\left( \sum_{t=1}^{T}\left\Vert \mathbf{%
				E}_{k,t}\right\Vert _{F}^{2}\right) \left\Vert \widehat{\mathbf{B}}_{k}-%
			\mathbf{B}_{k}\widehat{\mathbf{H}}_{-k}\right\Vert _{F}^{2}\left(
			\sum_{t=1}^{T}\left\Vert \mathbf{E}_{k,t}\mathbf{B}_{k}\right\Vert
			_{F}^{2}\right) \left\Vert \widetilde{\mathbf{A}}_{k}-\mathbf{A}_{k}%
			\widetilde{\mathbf{H}}_{k}\right\Vert _{F}^{2}  \notag \\
			=& O_{P}\left( \left( \frac{1}{Tp}+\frac{1}{p^{2}}\right) w_{-k}\right)
			+O_{P}\left( \frac{w_{-k}}{p}\right) \times \left\Vert \widetilde{\mathbf{A}}%
			_{k}-\mathbf{A}_{k}\widetilde{\mathbf{H}}_{k}\right\Vert _{F}^{2},
			\label{eq:Xcal}
		\end{align}%
		since $\sum_{t=1}^{T}\left\Vert \mathbf{E}_{k,t}\right\Vert
		_{F}^{2}=O_{P}(Tp)$ by the same arguments used in the proof of Lemma \ref%
		{lemma1}\textit{(i)}. We conclude by studying $\mathcal{VIII}_{c}$. By (\ref%
		{equ:3.1}) and parts \textit{(i)} and \textit{(iii)} of Lemma \ref{lemma1},
		it follows that
		\begin{align}
			& \frac{1}{p_{k}}\Vert \mathcal{VIII}_{c}\Vert _{F}^{2}  \notag \\
			& \leq \frac{1}{p_{k}}\left\Vert \frac{1}{Tp_{k}p_{-k}^{2}}\sum_{t=1}^{T}%
			\mathbf{E}_{k,t}\widehat{\mathbf{B}}_{k}\left( \widehat{\mathbf{B}}_{k}-%
			\mathbf{B}_{k}\widehat{\mathbf{H}}_{-k}\right) ^{\top }\mathbf{E}%
			_{k,t}^{\top }\left( \widetilde{\mathbf{A}}_{k}-\mathbf{A}_{k}\widetilde{%
				\mathbf{H}}_{k}\right) \right\Vert _{F}^{2}  \notag \\
			& +\frac{1}{p_{k}}\left\Vert \frac{1}{Tp_{k}p_{-k}^{2}}\sum_{t=1}^{T}\mathbf{%
				E}_{k,t}\widehat{\mathbf{B}}_{k}\left( \widehat{\mathbf{B}}_{k}-\mathbf{B}%
			_{k}\widehat{\mathbf{H}}_{-k}\right) ^{\top }\mathbf{E}_{k,t}^{\top }\mathbf{%
				A}_{k}\widetilde{\mathbf{H}}_{k}\right\Vert _{F}^{2}  \notag \\
			\lesssim & \frac{1}{T^{2}p_{k}^{3}p_{-k}^{4}}\sum_{t=1}^{T}\left\Vert
			\mathbf{E}_{k,t}\widehat{\mathbf{B}}_{k}\right\Vert
			_{F}^{2}\sum_{t=1}^{T}\left\Vert \mathbf{E}_{k,t}\left( \widehat{\mathbf{B}}%
			_{k}-\mathbf{B}_{k}\widehat{\mathbf{H}}_{-k}\right) \right\Vert
			_{F}^{2}\left\Vert \widetilde{\mathbf{A}}_{k}-\mathbf{A}_{k}\widetilde{%
				\mathbf{H}}_{k}\right\Vert _{F}^{2}  \notag \\
			& +\frac{1}{T^{2}p_{k}^{3}p_{-k}^{4}}\left\Vert \sum_{t=1}^{T}\mathbf{E}%
			_{k,t}\widehat{\mathbf{B}}_{k}\mathbf{A}_{k}^{\top }\mathbf{E}%
			_{k,t}\right\Vert _{F}^{2}\left\Vert \widehat{\mathbf{B}}_{k}-\mathbf{B}_{k}%
			\widehat{\mathbf{H}}_{-k}\right\Vert _{F}^{2}  \notag \\
			\leq & \frac{1}{T^{2}p_{k}^{3}p_{-k}^{4}}\left( \sum_{t=1}^{T}\left\Vert
			\mathbf{E}_{k,t}\mathbf{B}_{k}\right\Vert _{F}^{2}+\sum_{t=1}^{T}\left\Vert
			\mathbf{E}_{k,t}\left( \widehat{\mathbf{B}}_{k}-\mathbf{B}_{k}\widehat{%
				\mathbf{H}}_{-k}\right) \right\Vert _{F}^{2}\right) \sum_{t=1}^{T}\left\Vert
			\mathbf{E}_{k,t}\left( \widehat{\mathbf{B}}_{k}-\mathbf{B}_{k}\widehat{%
				\mathbf{H}}_{-k}\right) \right\Vert _{F}^{2}\left\Vert \widetilde{\mathbf{A}}%
			_{k}-\mathbf{A}_{k}\widetilde{\mathbf{H}}_{k}\right\Vert _{F}^{2}  \notag \\
			& +\frac{1}{T^{2}p_{k}^{3}p_{-k}^{4}}\left( \left\Vert \sum_{t=1}^{T}\mathbf{%
				E}_{k,t}\mathbf{B}_{k}\mathbf{A}_{k}^{\top }\mathbf{E}_{k,t}\right\Vert
			_{F}^{2}+\sum_{i=1}^{p_{k}}\left\Vert \sum_{t=1}^{T}\mathbf{A}_{k}^{\top }%
			\mathbf{E}_{k,t}\mathbf{e}_{k,t,i\cdot }^{\top }\right\Vert
			_{F}^{2}\left\Vert \widehat{\mathbf{B}}_{k}-\mathbf{B}_{k}\widehat{\mathbf{H}%
			}_{-k}\right\Vert _{F}^{2}\right) \left\Vert \widehat{\mathbf{B}}_{k}-%
			\mathbf{B}_{k}\widehat{\mathbf{H}}_{-k}\right\Vert _{F}^{2}  \notag \\
			=& O_{P}\left( \left( \frac{1}{Tp}+\frac{1}{p^{2}}\right) w_{-k}\right)
			+O_{P}\left( \left( \frac{1}{p_{k}^{2}}+\frac{1}{Tp_{k}}\right)
			w_{-k}^{2}\right) +O_{P}\left( \left( \frac{1}{p_{-k}}+w_{-k}\right) w_{-k}%
			\frac{1}{p_{k}}\right) \times \left\Vert \widetilde{\mathbf{A}}_{k}-\mathbf{A%
			}_{k}\widetilde{\mathbf{H}}_{k}\right\Vert _{F}^{2},  \label{eq:XIcal}
		\end{align}%
		since, by (\ref{equ:b.2}), $\left\Vert \sum_{t=1}^{T}\mathbf{E}_{k,t}\mathbf{%
			B}_{k}\mathbf{A}_{k}^{\top }\mathbf{E}_{k,t}\right\Vert
		_{F}^{2}=O_{P}(Tp^{2}+T^{2}p)$. Hence, combining (\ref{eq:IXcal}), (\ref%
		{eq:Xcal}) and (\ref{eq:XIcal}), it finally follows that
		\begin{align*}
			& \frac{1}{p_{k}}\left\Vert \mathcal{VIII}\widetilde{\mathbf{A}}%
			_{k}\right\Vert _{F}^{2} \\
			& =O_{P}\left( \frac{1}{Tp}+\frac{1}{p^{2}}\right) +O_{P}\left( \left( \frac{%
				1}{p_{k}^{2}}+\frac{1}{Tp_{k}}\right) w_{-k}^{2}\right) \\
			& +O_{P}\left( \frac{1}{p_{-k}^{2}}+\frac{w_{-k}}{p_{-k}}+w_{-k}^{2}\right)
			\times \frac{1}{p_{k}}\left\Vert \widetilde{\mathbf{A}}_{k}-\mathbf{A}_{k}%
			\widetilde{\mathbf{H}}_{k}\right\Vert _{F}^{2},
		\end{align*}%
		thus completing the proof.
	\end{proof}
\end{lemma}

\begin{lemma}
	\label{lemma6}We assume that Assumptions \ref{as-1}-\ref{as-4} hold and that
	$\widehat{\mathbf{A}}_{k}$ is used as projection matrix. Then, as $\min
	\left\{ T,p_{1},...,p_{K}\right\} \rightarrow \infty $, it holds that%
	\begin{eqnarray*}
		&&\left\Vert \frac{1}{p_{k}}\mathbf{A}_{k}^{\top }\left( \widetilde{\mathbf{A%
		}}_{k}-\mathbf{A}_{k}\widetilde{\mathbf{H}}_{k}\right) \right\Vert _{F} \\
		&=&O_{P}\left( \frac{1}{\sqrt{Tp}}\right) +O_{P}\left( \frac{1}{p}\right)
		+O_{P}\left( \sum_{j=1}^{K}\frac{1}{Tp_{-j}}\right) \\
		&&+O_{P}\left( \sum_{j=1,j\neq k}^{K}\left( \frac{1}{Tp_{j}\sqrt{p_{-k}}}+%
		\frac{1}{p_{j}\sqrt{Tp_{k}}}+\frac{1}{p_{k}p_{j}^{2}}+\frac{1}{Tp_{j}^{2}}%
		\right) \right) .
	\end{eqnarray*}
	
	\begin{proof}
		As in the above, we will assume $r_{k}=1$, $1\leq k\leq K$. By equation (\ref%
		{dec-m-tilde}), it holds that
		\begin{equation*}
			\frac{1}{p_{k}}\mathbf{A}_{k}^{\top }\left( \widetilde{\mathbf{A}}_{k}-%
			\mathbf{A}_{k}\widetilde{\mathbf{H}}_{k}\right) =\frac{1}{p_{k}}\mathbf{A}%
			_{k}^{\top }\left( \mathcal{VI}+\mathcal{VII}+\mathcal{VIII}\right)
			\widetilde{\mathbf{A}}_{k}\widetilde{\mathbf{\Lambda }}_{k}^{-1}.
		\end{equation*}%
		We know from (\ref{eq:LambdaHat}) that $\left\Vert \widetilde{\Lambda }%
		_{k}^{-1}\right\Vert _{F}=O_{P}\left( 1\right) $. It holds that%
		\begin{eqnarray*}
			&&\left\Vert \frac{1}{p_{k}}\mathbf{A}_{k}^{\top }\mathcal{VI}\widetilde{%
				\mathbf{A}}_{k}\right\Vert _{F} \\
			&=&\left\Vert \frac{1}{Tp^{2}}\mathbf{A}_{k}^{\top }\sum_{t=1}^{T}\mathbf{E}%
			_{k,t}\widehat{\mathbf{B}}_{k}\widehat{\mathbf{B}}_{k}^{\top }\mathbf{B}_{k}%
			\mathbf{F}_{k,t}^{\top }\mathbf{A}_{k}^{\top }\widetilde{\mathbf{A}}%
			_{k}\right\Vert _{F} \\
			&\lesssim &\frac{1}{\left( Tp\right) ^{1/2}}\left\Vert \frac{1}{T^{1/2}}%
			\sum_{t=1}^{T}\left( \frac{\mathbf{A}_{k}}{p_{k}^{1/2}}\right) ^{\top }%
			\mathbf{E}_{k,t}\left( \frac{\widehat{\mathbf{B}}_{k}}{p_{-k}^{1/2}}\right)
			\mathbf{F}_{k,t}^{\top }\right\Vert _{F}=O_{P}\left( \left( Tp\right)
			^{-1/2}\right) ,
		\end{eqnarray*}%
		by Assumption \ref{as-4}\textit{(i)}, recalling that $\left\Vert \mathbf{A}%
		_{k}\right\Vert _{F}=c_{0}p_{k}^{1/2}$\ and $\left\Vert \mathbf{B}%
		_{k}\right\Vert _{F}=c_{1}p_{-k}^{1/2}$. Similarly%
		\begin{eqnarray*}
			&&\left\Vert \frac{1}{p_{k}}\mathbf{A}_{k}^{\top }\mathcal{VII}\widetilde{%
				\mathbf{A}}_{k}\right\Vert _{F} \\
			&=&\left\Vert \frac{1}{Tp^{2}}\mathbf{A}_{k}^{\top }\mathbf{A}%
			_{k}\sum_{t=1}^{T}\mathbf{B}_{k}^{\top }\widehat{\mathbf{B}}_{k}\widehat{%
				\mathbf{B}}_{k}^{\top }\mathbf{E}_{k,t}^{\top }\mathbf{F}_{k,t}\widetilde{%
				\mathbf{A}}_{k}\right\Vert _{F} \\
			&\lesssim &\left\Vert \frac{1}{Tp}\sum_{t=1}^{T}\widehat{\mathbf{B}}%
			_{k}^{\top }\mathbf{E}_{k,t}^{\top }\mathbf{F}_{k,t}\mathbf{A}_{k}\widetilde{%
				\mathbf{H}}_{k}\right\Vert _{F}+\left\Vert \frac{1}{Tp}\sum_{t=1}^{T}%
			\widehat{\mathbf{B}}_{k}^{\top }\mathbf{E}_{k,t}^{\top }\mathbf{F}%
			_{k,t}\left( \widetilde{\mathbf{A}}_{k}-\mathbf{A}_{k}\widetilde{\mathbf{H}}%
			_{k}\right) \right\Vert _{F},
		\end{eqnarray*}%
		and
		\begin{eqnarray*}
			&&\left\Vert \frac{1}{Tp}\sum_{t=1}^{T}\widehat{\mathbf{B}}_{k}^{\top }%
			\mathbf{E}_{k,t}^{\top }\mathbf{F}_{k,t}\left( \widetilde{\mathbf{A}}_{k}-%
			\mathbf{A}_{k}\widetilde{\mathbf{H}}_{k}\right) \right\Vert _{F} \\
			&\lesssim &\left\Vert \frac{1}{Tp}\sum_{t=1}^{T}\widehat{\mathbf{H}}%
			_{-k}^{\top }\mathbf{B}_{k}^{\top }\mathbf{E}_{k,t}^{\top }\mathbf{F}%
			_{k,t}\right\Vert _{F}\left\Vert \widetilde{\mathbf{A}}_{k}-\mathbf{A}_{k}%
			\widetilde{\mathbf{H}}_{k}\right\Vert _{F}+\left\Vert \frac{1}{Tp}%
			\sum_{t=1}^{T}\left( \widehat{\mathbf{B}}_{k}-\mathbf{B}_{k}\widehat{\mathbf{%
					H}}_{-k}\right) ^{\top }\mathbf{E}_{k,t}^{\top }\mathbf{F}_{k,t}\right\Vert
			_{F}\left\Vert \widetilde{\mathbf{A}}_{k}-\mathbf{A}_{k}\widetilde{\mathbf{H}%
			}_{k}\right\Vert _{F} \\
			&=&O_{P}\left( \sqrt{p_{k}\widetilde{w}_{k}}\right) \left( \left\Vert \frac{1%
			}{Tp}\sum_{t=1}^{T}\mathbf{B}_{k}^{\top }\mathbf{E}_{k,t}^{\top }\mathbf{F}%
			_{k,t}\right\Vert _{F}+\left\Vert \frac{1}{Tp}\sum_{t=1}^{T}\left( \widehat{%
				\mathbf{B}}_{k}-\mathbf{B}_{k}\widehat{\mathbf{H}}_{-k}\right) ^{\top }%
			\mathbf{E}_{k,t}^{\top }\mathbf{F}_{k,t}\right\Vert _{F}\right) ,
		\end{eqnarray*}%
		by Theorem \ref{th:2}. Hence, using the same logic as in the passages above%
		\begin{eqnarray*}
			&&\left\Vert \frac{1}{Tp}\sum_{t=1}^{T}\widehat{\mathbf{B}}_{k}^{\top }%
			\mathbf{E}_{k,t}^{\top }\mathbf{F}_{k,t}\left( \widetilde{\mathbf{A}}_{k}-%
			\mathbf{A}_{k}\widetilde{\mathbf{H}}_{k}\right) \right\Vert _{F} \\
			&=&O_{P}\left( \sqrt{p_{k}\widetilde{w}_{k}}\right) O_{P}\left( \frac{1}{%
				\sqrt{Tp}}+\sqrt{\frac{m_{-k}}{p_{k}}}\right) .
		\end{eqnarray*}%
		Finally, the same logic as in the above yields%
		\begin{eqnarray*}
			&&\left\Vert \frac{1}{p_{k}}\mathbf{A}_{k}^{\top }\mathcal{VIII}\widetilde{%
				\mathbf{A}}_{k}\right\Vert _{F} \\
			&=&\left\Vert \frac{1}{Tp^{2}}\sum_{t=1}^{T}\mathbf{A}_{k}^{\top }\mathbf{E}%
			_{k,t}\widehat{\mathbf{B}}_{k}\widehat{\mathbf{B}}_{k}^{\top }\mathbf{E}%
			_{k,t}^{\top }\widetilde{\mathbf{A}}_{k}\right\Vert _{F} \\
			&\lesssim &\frac{1}{p}\sqrt{\frac{1}{Tp}\sum_{t=1}^{T}\left\Vert \frac{1}{%
					Tp^{2}}\sum_{t=1}^{T}\mathbf{A}_{k}^{\top }\mathbf{E}_{k,t}\widehat{\mathbf{B%
				}}_{k}\right\Vert _{F}\times \frac{1}{Tp}\sum_{t=1}^{T}\left\Vert \frac{1}{%
					Tp^{2}}\sum_{t=1}^{T}\widehat{\mathbf{B}}_{k}^{\top }\mathbf{E}_{k,t}^{\top }%
				\widetilde{\mathbf{A}}_{k}\right\Vert _{F}} \\
			&\lesssim &\frac{1}{p}\sqrt{\left( 1+p_{-k}w_{-k}\right) ^{2}\left( 1+p_{k}%
				\widetilde{w}_{k}\right) }.
		\end{eqnarray*}%
		The desired result now follows from putting all together.
	\end{proof}
\end{lemma}

\newpage

\section{Simulations\label{sec:4}}

\subsection{Data generation\label{sec:4.1}}

We investigate the finite sample performance of the proposed
iterative projection methods. We compare the
performances of initial estimator (IE), the projected estimator (PE),
iterative projected mode-wise PCA estimation (IPmoPCA) by \cite%
{zhang2022tucker}, and the Time series Outer-Product Unfolding Procedure
(TOPUP) and Time series Inner-Product Unfolding Procedure (TIPUP) with their
iteration procedure (iTOPUP and iTIPUP) by \cite{Chen2021Factor} in terms of
estimating the loading matrices, the common components and the number of
factors. The tensor observations are generated following the order 3 tensor
factor model:
\begin{equation*}
	\cX_{t}=\cF_{t}\times _{1}\Ab_{1}\times _{2}\Ab_{2}\times _{3}\Ab_{3}+\cE%
	_{t}.
\end{equation*}%
We set $r_{1}=r_{2}=r_{3}=3$, draw the entries of $\Ab_{1}$, $\Ab_{2}$ and $%
\Ab_{3}$ independently from a uniform distribution $\cU(-1,1)$, and let
\begin{equation*}
	\begin{gathered} \operatorname{Vec}\left(\cF_{t}\right)=\phi \times
		\operatorname{Vec}\left(\cF_{t-1}\right)+\sqrt{1-\phi^{2}} \times
		\epsilon_{t},\quad \epsilon_{t} \stackrel{i . i . d}{\sim}
		\cN\left(\mathbf{0}, \Ib_{r_1r_2r_3}\right) \\
		\operatorname{Vec}\left(\cE_{t}\right)=\psi \times
		\operatorname{Vec}\left(\cE_{t-1}\right)+\sqrt{1-\psi^{2}} \times
		\operatorname{Vec}\left(\cU_{t}\right) \end{gathered}
\end{equation*}%
where $\cU_{t}$ is drawn from a tensor normal distribution, i.e., $\cT\cN(\cM%
,\bSigma_{1},\bSigma_{2},\bSigma_{3})$, which is equivalent to saying that $%
\text{Vec}{(}\cU_{t})~{\sim }~\cN(\text{Vec}{(}\cM),\bSigma_{3}\otimes %
\bSigma_{2}\otimes \bSigma_{1})$. In our study, we set $\cM=0$, $\bSigma_{k}$
to be the matrix with 1 on the diagonal, and $1/p_{k}$ on the off-diagonal,
for $k=1,2,3$. The parameters $\phi $ and $\psi $ control for temporal
correlations of $\cF_{t}$ and $\cE_{t}$. By setting $\phi $ and $\psi $
unequal to zero, the common factors have cross-correlations, and the
idiosyncratic components have both cross-correlations and weak
autocorrelations.

In sections \ref{sec:4.2} and \ref{sec:4.3}, it is assumed that factor
numbers are known, and the performance of estimating the number of factors
is investigated in section \ref{sec:4.4}. All the following simulation
results are based on 1000 replications.

\begin{table}[!h]
	\caption{Averaged estimation errors and standard errors of loading spaces
		for Settings A, B and D over 1000 replications. ``PE": projection estimation
		method.``IE": initial estimation method. ``TOPUP": Time series Outer-Product
		Unfolding Procedure with $h_0=1$. ``TIPUP": Time series Inner-Product
		Unfolding Procedure with $h_0=1$. ``iTOPUP": Iterative Time series
		Outer-Product Unfolding Procedure with $h_0=1$. ``iTIPUP": Iterative Time
		series Inner-Product Unfolding Procedure with $h_0=1$. ``IPmoPCA": iterative
		projected mode-wise PCA estimation. }
	\label{tab:1}
	
	\footnotesize
	\renewcommand{\arraystretch}{0.5} 
	\centering
	\scalebox{1}{\begin{tabular}{cccccccccccc}
			\toprule[2pt]
			Evaluation	&	$p_1$	&	$p_2$	&	$p_3$	&	$T$	&	PE	&	IE	&	IPmoPCA	&	TOPUP	&	TIPUP	&	iTOPUP	&	iTIPUP	\\
			\midrule
			$\cD(\widehat{\Ab}_1,\Ab_1)$	&	10	&	10	&	10	&	20	&	0.0444	&	0.1970 	&	0.0476 	&	0.2065 	&	0.4711 	&	0.0677 	&	0.3564 	\\
			&		&		&		&	50	&	0.0322 	&	0.1783 	&	0.0286 	&	0.1851 	&	0.4191 	&	0.0544 	&	0.3039 	\\
			&		&		&		&	100	&	0.0248 	&	0.1728 	&	0.0219 	&	0.1776 	&	0.3426 	&	0.0502 	&	0.2357 	\\
			&		&		&		&	200	&	0.0202 	&	0.1713 	&	0.0175 	&	0.1732 	&	0.2695 	&	0.0477 	&	0.1519 	\\
			&	100	&	10	&	10	&	20	&	0.0424 	&	0.0410 	&	0.0401 	&	0.0486 	&	0.3541 	&	0.0584 	&	0.3495 	\\
			&		&		&		&	50	&	0.0266 	&	0.0259 	&	0.0252 	&	0.0348 	&	0.2963 	&	0.0476 	&	0.2941 	\\
			&		&		&		&	100	&	0.0189 	&	0.0186 	&	0.0179 	&	0.0291 	&	0.2328 	&	0.0435 	&	0.2350 	\\
			&		&		&		&	200	&	0.0133 	&	0.0133 	&	0.0126 	&	0.0257 	&	0.1526 	&	0.0405 	&	0.1548 	\\
			&	20	&	20	&	20	&	20	&	0.0203 	&	0.0512 	&	0.0199 	&	0.0536 	&	0.3391 	&	0.0294 	&	0.2551 	\\
			&		&		&		&	50	&	0.0127 	&	0.0456 	&	0.0125 	&	0.0478 	&	0.2656 	&	0.0237 	&	0.1927 	\\
			&		&		&		&	100	&	0.0091 	&	0.0450 	&	0.0090 	&	0.0467 	&	0.1817 	&	0.0218 	&	0.1363 	\\
			&		&		&		&	200	&	0.0064 	&	0.0445 	&	0.0063 	&	0.0452 	&	0.0978 	&	0.0204 	&	0.0758 	\\
			$\cD(\widehat{\Ab}_2,\Ab_2)$	&	10	&	10	&	10	&	20	&	0.0474 	&	0.1873 	&	0.0430	&	0.1943 	&	0.4634 	&	0.0647 	&	0.3456 	\\
			&		&		&		&	50	&	0.0327 	&	0.1843 	&	0.0290 	&	0.1899 	&	0.4198 	&	0.0554 	&	0.3066 	\\
			&		&		&		&	100	&	0.0250 	&	0.1736 	&	0.0221 	&	0.1781 	&	0.3388 	&	0.0505 	&	0.2344 	\\
			&		&		&		&	200	&	0.0205 	&	0.1773 	&	0.0177 	&	0.1782 	&	0.2744 	&	0.0483 	&	0.1530 	\\
			&	100	&	10	&	10	&	20	&	0.0129 	&	0.1778 	&	0.0125 	&	0.1821 	&	0.4170 	&	0.0192 	&	0.2039 	\\
			&		&		&		&	50	&	0.0082 	&	0.1746 	&	0.0079 	&	0.1763 	&	0.3518 	&	0.0159 	&	0.1510 	\\
			&		&		&		&	100	&	0.0059 	&	0.1728 	&	0.0057 	&	0.1732 	&	0.2797 	&	0.0148 	&	0.0958 	\\
			&		&		&		&	200	&	0.0042 	&	0.1642 	&	0.0040 	&	0.1633 	&	0.2162 	&	0.0134 	&	0.0492 	\\
			&	20	&	20	&	20	&	20	&	0.0203 	&	0.0514 	&	0.0199 	&	0.0542 	&	0.3388 	&	0.0293 	&	0.2551 	\\
			&		&		&		&	50	&	0.0127 	&	0.0461 	&	0.0125 	&	0.0481 	&	0.2752 	&	0.0237 	&	0.2007 	\\
			&		&		&		&	100	&	0.0091 	&	0.0450 	&	0.0089 	&	0.0466 	&	0.1812 	&	0.0220 	&	0.1375 	\\
			&		&		&		&	200	&	0.0064 	&	0.0430 	&	0.0063 	&	0.0439 	&	0.0999 	&	0.0205 	&	0.0783 	\\
			$\cD(\widehat{\Ab}_3,\Ab_3)$	&	10	&	10	&	10	&	20	&	0.0482 	&	0.1927 	&	0.0440	&	0.1992 	&	0.4642 	&	0.0654 	&	0.3539 	\\
			&		&		&		&	50	&	0.0325 	&	0.1821 	&	0.0289 	&	0.1880 	&	0.4208 	&	0.0550 	&	0.3049 	\\
			&		&		&		&	100	&	0.0252 	&	0.1716 	&	0.0222 	&	0.1757 	&	0.3443 	&	0.0507 	&	0.2350 	\\
			&		&		&		&	200	&	0.0204 	&	0.1691 	&	0.0175 	&	0.1712 	&	0.2587 	&	0.0478 	&	0.1515 	\\
			&	100	&	10	&	10	&	20	&	0.0128 	&	0.1748 	&	0.0125 	&	0.1797 	&	0.4146 	&	0.0189 	&	0.2050 	\\
			&		&		&		&	50	&	0.0081 	&	0.1685 	&	0.0078	&	0.1707 	&	0.3436 	&	0.0158 	&	0.1500 	\\
			&		&		&		&	100	&	0.0059 	&	0.1719 	&	0.0057 	&	0.1732 	&	0.2770 	&	0.0147 	&	0.0944 	\\
			&		&		&		&	200	&	0.0042 	&	0.1681 	&	0.0041 	&	0.1676 	&	0.2169 	&	0.0138 	&	0.0492 	\\
			&	20	&	20	&	20	&	20	&	0.0203 	&	0.0506 	&	0.0199 	&	0.0534 	&	0.3404 	&	0.0295 	&	0.2617 	\\
			&		&		&		&	50	&	0.0128 	&	0.0462 	&	0.0126 	&	0.0484 	&	0.2618 	&	0.0238 	&	0.1945 	\\
			&		&		&		&	100	&	0.0090 	&	0.0435 	&	0.0089 	&	0.0450 	&	0.1750 	&	0.0219 	&	0.1328 	\\
			&		&		&		&	200	&	0.0064 	&	0.0425 	&	0.0063 	&	0.0433 	&	0.0977 	&	0.0205 	&	0.0772 	\\

			\bottomrule[2pt]
	\end{tabular}}
\end{table}

\begin{figure}[htbp]
	%%ÍŒ,[htbp]ÊÇž¡¶¯žñÊœ
	\centering
	%\begin{minipage}{0.505\linewidth}      %ÍŒÆ¬ÕŒÓÃÒ»ÐÐ¿í¶ÈµÄ30%
	\centering
	\includegraphics[width=0.5\linewidth]{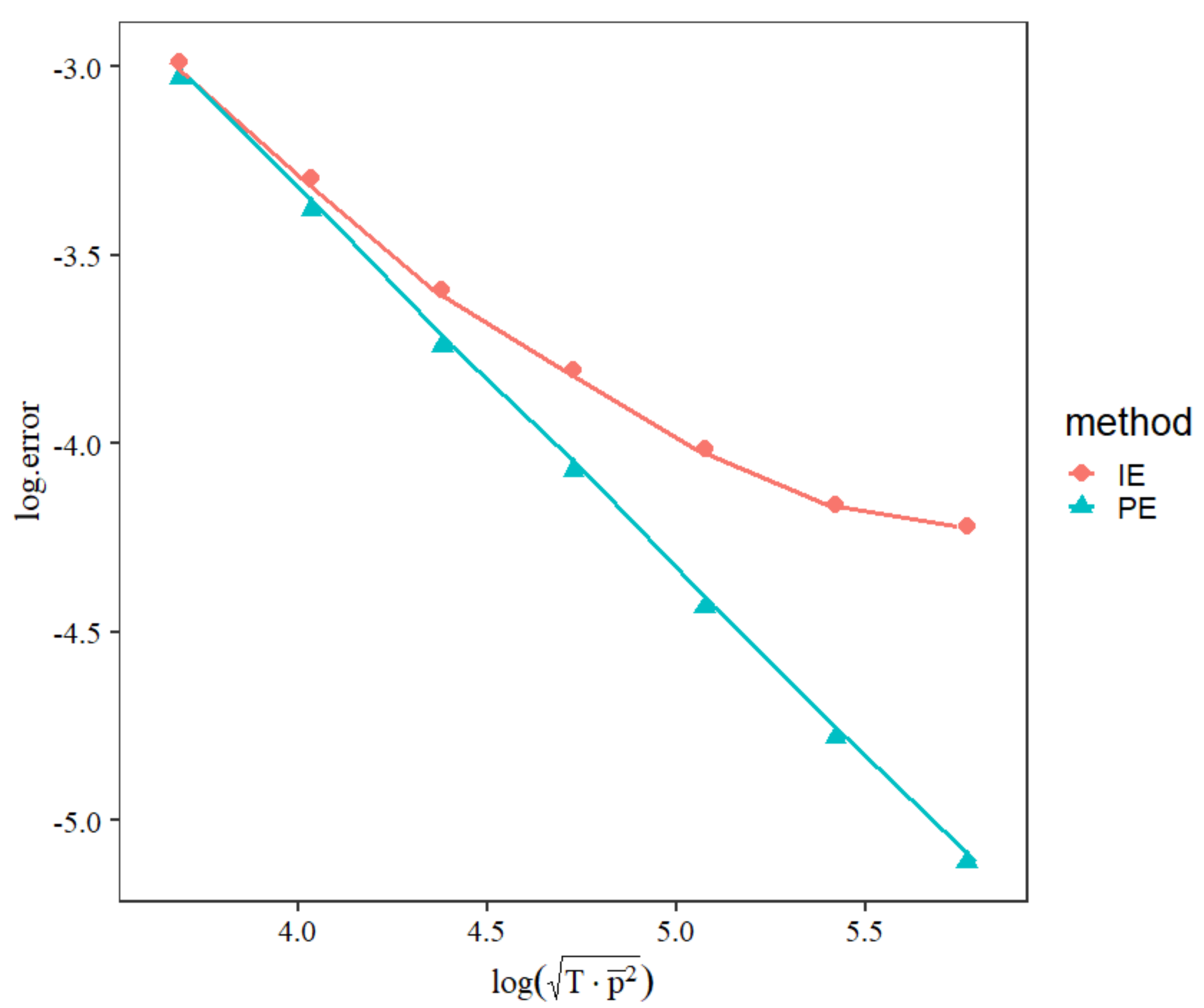}
	\caption{Mean log error for estimating the loading matrices over 1000 replications for $p_1=40,~p_2=p_3=\bar p=10,~T\in (16,32,64,128,256,512,1024)$.}\label{fig:1}
	%	\end{minipage}
\end{figure}

\begin{figure}[htbp]
%\begin{minipage}{0.485\linewidth}      %ÍŒÆ¬ÕŒÓÃÒ»ÐÐ¿í¶ÈµÄ30%
\centering
\includegraphics[width=0.5\linewidth]{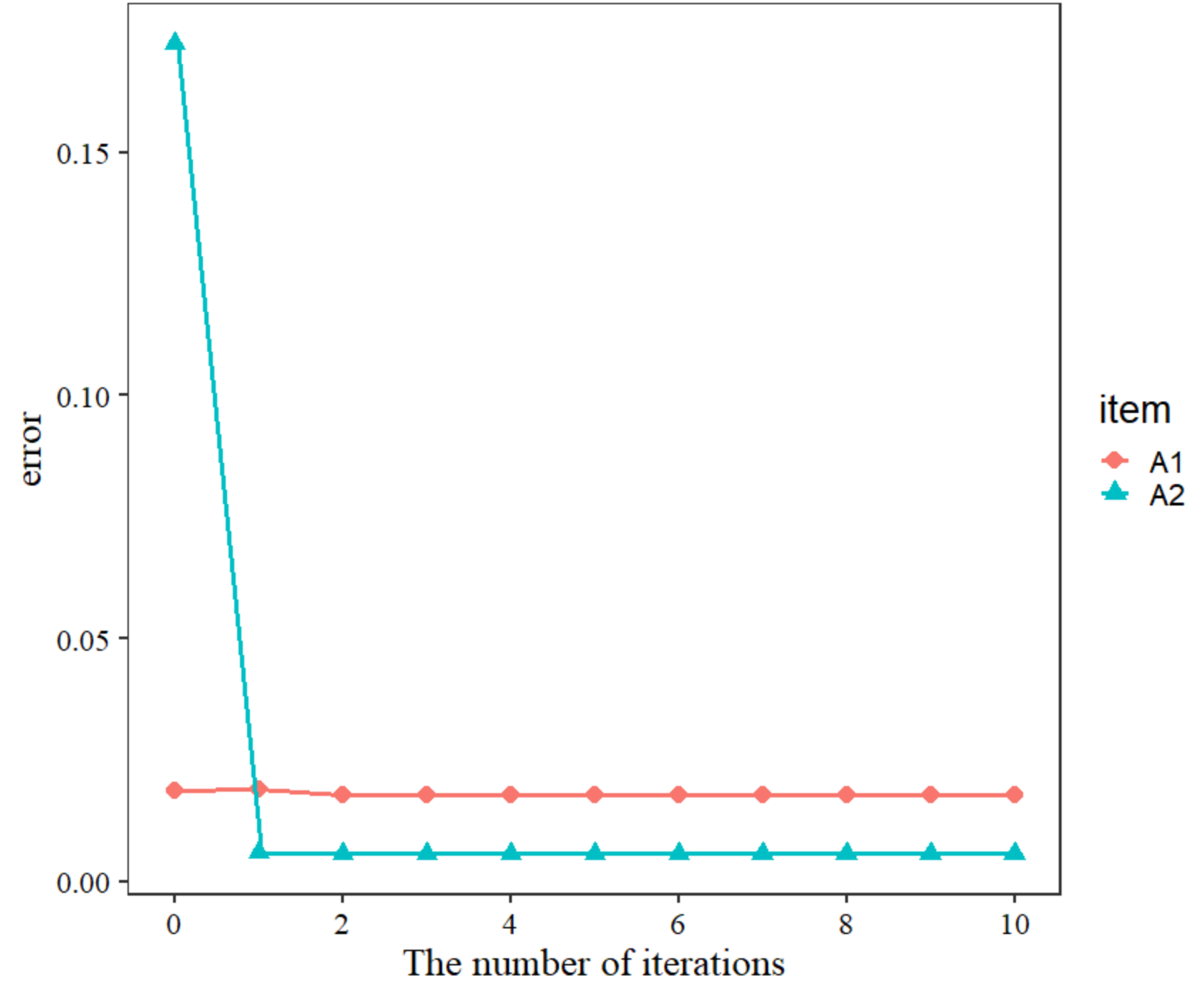}
\caption{Mean estimation error at each step of the recursive procedure over 1000 replications for $T=p_1=100,~p_2=p_3=10$.}\label{fig:2}
%	\end{minipage}
\end{figure}

\subsection{Verifying the convergence rates for loading spaces\label{sec:4.2}}

We compare the performance of PE, IE, IPmoPCA, TOPUP, TIPUP, iTOPUP, and
iTIPUP (with $h_{0}=1$ in their implementation) in estimating loading
spaces. We consider the following three settings:

\vspace{0.5em}

\textbf{Setting A:~} $p_1=p_2=p_3=10,~\phi=0.1,~\psi=0.1,~T \in
(20,50,100,200)$

\vspace{0.5em}

\textbf{Setting B:~} $p_1=100,~p_2=p_3=10,~\phi=0.1,~\psi=0.1,~T \in
(20,50,100,200)$

\vspace{0.5em}

\textbf{Setting C:~} $p_1=p_2=p_3=15,~\phi=0.1,~\psi=0.1,~T \in
(20,50,100,200)$

\vspace{0.5em}

\textbf{Setting D:~} $p_1=p_2=p_3=20,~\phi=0.1,~\psi=0.1,~T \in
(20,50,100,200)$

\vspace{0.5em}

\textbf{Setting E:~} $p_1=p_2=p_3=30,~\phi=0.1,~\psi=0.1,~T \in
(20,50,100,200)$

\vspace{0.5em}

\textbf{Setting F:~} $p_1=40,~p_2=p_3=10,~\phi=0.1,~\psi=0.1,~T \in
(16,32,64,128,256,512,1024)$

\vspace{0.5em}

Due to the identifiability issue of factor model, the performance of the
candidates methods is evaluated by comparing the distance between the
estimated loading space and the true loading space, which is
\begin{equation*}
\cD(\widehat{\Ab}_k, \Ab_k)=\left(1-\frac{1}{r_k} \mathop{\mathrm{Tr}}\left(%
\widehat{\Qb}_k\widehat{\Qb}_k^{\top} \Qb_k \Qb_k^{\top}\right)\right)^{1 /
2},~k=1,2,3,
\end{equation*}
where $\Qb_k$ and $\widehat{\Qb}_k$ are the left singular-vector matrices of
the true loading matrix $\Ab_k$ and its estimator $\widehat{\Ab}_k$. The
distance is always in the interval $[0,1]$. When $\Ab_k$ and $\widehat{\Ab}%
_k $ span the same space, the distance $\cD(\widehat{\Ab}_k, \Ab_k)$ is
equal to 0, while is equal to 1 when the two spaces are orthogonal.

Table \ref{tab:1} shows the averaged estimation errors under Settings A, B
and D for tensor normal distribution. Table \ref{tab:1} shows that all these
methods benefit from the increase in dimensions and sample size $T$. For
fixed sample size $T$, when $p_k$s are of the same order, PE performs better
than all other procedures in estimating loading except IPmoPCA (IPmoPCA
adopts a multi-step projection technique of PE). When the tensor dimensions
are not balanced, that is, one mode dimension is much higher than the other
mode dimensions and sample size $T$, the estimation effect of PE for this
dimension is comparable to that of IE, which is consistent with the
conclusion of Theorem \ref{th:1} and Corollary \ref{cor:1} that when $p_k$
is in high order and $p_j \asymp \bar p$ for all $j\neq k$, the convergence
rates of IE and PE are both $1/(T\bar p^2)$. Figure \ref{fig:1} plots mean
log error of PE and IE under setting F, which clearly shows that PE and IE
have different convergence rates. Figure \ref{fig:1} shows that the log
error of the PE method for estimating $\Ab_1$ is almost linear to $\log(%
\sqrt{T\bar p^2})$ with slope -1 while $p_2=p_3=\bar p$, which also matches
the rate in Corollary \ref{cor:1}. For the IE method, the log error
initially decreases as $\log(\sqrt{Tp_2p_3})$ increases and then tends to be
constant. According to Theorem \ref{th:1}, this is because when $\log(\sqrt{%
T\bar p^2})$ is large enough, the error of IE only depends on $p_1$. This
shows that PE is a more accurate estimation method compared IE.

It can be seen from Table \ref{tab:1} that, as one-step iteration and
multi-step iteration versions, PE and IPmoPCA have similar performance under
any setting. And when the dimensions $p_k$s are large enough, they behave
almost exactly the same. We consider whether the multi-step iteration will
further reduce the estimation error and propose the following iterative
estimation Algorithm \ref{alg7}. Figure \ref{fig:2} plots the estimation
error of each step of the 10-step iterative estimators of $\Ab_1$ and $\Ab_2$
under Setting B. Since $\Ab_2$ and $\Ab_3$ are symmetric, the estimation
result of $\Ab_3$ is omitted here. The estimation error of $\Ab_2$ decreases
significantly in the first step, but does not decrease further in the
subsequent multi-step iterations. The error in the estimation of $\Ab_1$ has
been almost constant, which is due to the fact that the size of estimation
error of $\Ab_k$ are all $O_p(1/\sqrt{T\bar p^2})$ for any step while $%
p_j=\bar p$ for all $j\neq k$ and $p^2=O(T\bar p^{2})$. This shows that our
one-step iterative algorithm \ref{alg1} can get comparable results with the
IPmoPCA by \cite{zhang2022tucker} while has much less computational burden.
\begin{algorithm}[htbp]
\caption{Iterative estimation for loading spaces}
\label{alg7}
\begin{algorithmic}[1]

\REQUIRE tensor datas $\{\cX_t,1\le t\le T\}$, factor numbers $r_1,\cdots,r_K$

\ENSURE factor loading matrices $\{\widehat{\Ab}_k^{(s)},1\le k\le K,0\le s\le S\}$

\STATE obtain the initial estimators $\{\widehat{\Ab}_k,1\le k\le K\}$, let $\widehat\Ab_k^{(0)}=\widehat{\Ab}_k$;

\STATE project to reduce dimensions by defining $\widehat{\Yb}_{k,t}^{(s)}:=\frac{1}{p_{-k}}\Xb_{k,t}\widehat{\Bb}_k^{(s)},\text{where~}\widehat{\Bb}_k^{(s)}=\otimes_{j\in[K]\slash \{k\}}\widehat{\Ab}_j^{(s-1)},1\le k\le K$;

\STATE given $\{\widehat{\Yb}_{k,t}^{(s)},1\le k\le K\}$, define $\widetilde{\Mb}_k^{(s)}:=(Tp_{k})^{-1}\sum_{t=1}^T \widehat{\Yb}_{k,t}^{(s)}(\widehat{\Yb}_{k,t}^{(s)})^{\top}$, set $\widehat{\Ab}_k^{(s)}$ as $\sqrt{p_k}$ times the matrix with columns being the first $r_k$ eigenvectors of $\widehat{\Mb}_k^{(s)}$;

\STATE Repeat steps 2 to 3 until up to the maximum number of iterations $s=S$.   Output the projection estimators as $\{\widehat{\Ab}_k^{(s)},1\le k\le K,0\le s\le S\}$.
\end{algorithmic} 	
\end{algorithm}

\begin{table}[!h]
\caption{Averaged estimation errors and standard errors of common components
for Settings A, B and D over 1000 replications.
``PE"``IE"``TOPUP"``TIPUP"``iTOPUP"``iTIPUP" are the same as in Table
\protect\ref{tab:2} }
\label{tab:2}

\footnotesize
\renewcommand{\arraystretch}{0.5} \centering
\scalebox{1}{\begin{tabular}{ccccccccccc}
\toprule[2pt]
$p_1$	&	$p_2$	&	$p_3$	&	$T$	&	PE	&	IE	&	IPmoPCA	&	TOPUP	&	TIPUP	&	iTOPUP	&	iTIPUP	\\
\midrule
10	&	10	&	10	&	20	&	0.032144 	&	0.082885 	&	0.031372 	&	0.088068 	&	0.463034 	&	0.036216 	&	0.287679 	\\
&		&		&	50	&	0.029511 	&	0.078503 	&	0.029042 	&	0.082402 	&	0.402929 	&	0.033516 	&	0.239876 	\\
&		&		&	100	&	0.028914 	&	0.075327 	&	0.028576 	&	0.077963 	&	0.282279 	&	0.032770 	&	0.161073 	\\
&		&		&	200	&	0.028364 	&	0.075043 	&	0.028087 	&	0.076372 	&	0.170435 	&	0.032172 	&	0.083502 	\\
100	&	10	&	10	&	20	&	0.004583 	&	0.032746 	&	0.004404 	&	0.035071 	&	0.377769 	&	0.006369 	&	0.190216 	\\
&		&		&	50	&	0.003459 	&	0.030847 	&	0.003388 	&	0.032154 	&	0.280331 	&	0.005202 	&	0.133925 	\\
&		&		&	100	&	0.003069 	&	0.030864 	&	0.003033 	&	0.031784 	&	0.171590 	&	0.004768 	&	0.079095 	\\
&		&		&	200	&	0.002887 	&	0.029564 	&	0.002869 	&	0.029921 	&	0.083551 	&	0.004511 	&	0.031586 	\\
20	&	20	&	20	&	20	&	0.004394 	&	0.009253 	&	0.004361 	&	0.009900 	&	0.308279 	&	0.005508 	&	0.182627 	\\
&		&		&	50	&	0.003792 	&	0.008399 	&	0.003779 	&	0.008847 	&	0.229848 	&	0.004817 	&	0.119517 	\\
&		&		&	100	&	0.003592 	&	0.008126 	&	0.003585 	&	0.008432 	&	0.119257 	&	0.004587 	&	0.062612 	\\
&		&		&	200	&	0.003486 	&	0.007941 	&	0.003483 	&	0.008077 	&	0.033438 	&	0.004446 	&	0.020716 	\\

\bottomrule[2pt]
\end{tabular}}
\end{table}

\subsection{Verifying the convergence rates for common components\label%
{sec:4.3}}

We compare the performance of PE, IE, TOPUP, TIPUP, iTOPUP and iTIPUP (with $%
h_{0}=1$ in their implementation) in common components under Setting A, B,
D. We use the mean squared error to evaluate the performance of different
estimated procedure, i.e.,
\begin{equation*}
\mathrm{MSE}=\frac{1}{Tp}\sum_{t=1}^{T}\left\Vert \widehat{\cS}_{t}-\cS%
_{t}\right\Vert _{F}^{2}.
\end{equation*}

Table \ref{tab:2} shows the averaged estimation errors under Settings A, B
and D. Table \ref{tab:2} shows that all these methods benefit from the
increase in dimensions and sample size $T$, and PE performs to IPmoPCA and
better than all other producers in estimating common components.

\begin{table}[!h]
\caption{ The frequencies of exact estimation of the numbers of factors
under Settings A, C, D and E over 1000 replications. ``$\backslash$" means
the data is too large to calculate.}
\label{tab:3}

\footnotesize
\renewcommand{\arraystretch}{0.5} \centering
\begin{tabular}{ccccccccc}
\toprule[2pt] $p_1=p_2=p_3$ & $T$ & PE-ER & IE-ER & TCorTh & TOP-ER & TIP-ER
& iTOP-ER & iTIP-ER \\
\midrule 10 & 20 & 0.395 & 0.049 & 0.000 & 0.048 & 0.002 & 0.776 & 0.046 \\
& 50 & 0.424 & 0.071 & 0.034 & 0.065 & 0.003 & 0.789 & 0.097 \\
& 100 & 0.450 & 0.069 & 0.139 & 0.058 & 0.013 & 0.811 & 0.302 \\
& 200 & 0.456 & 0.077 & 0.318 & 0.076 & 0.029 & 0.796 & 0.577 \\
&  &  &  &  &  &  &  &  \\
15 & 20 & 0.932 & 0.309 & 0.125 & 0.288 & 0.007 & 0.992 & 0.169 \\
& 50 & 0.947 & 0.343 & 0.596 & 0.329 & 0.023 & 0.993 & 0.321 \\
& 100 & 0.944 & 0.346 & 0.822 & 0.340 & 0.061 & 0.991 & 0.644 \\
& 200 & 0.948 & 0.362 & 0.920 & 0.347 & 0.134 & 0.999 & 0.940 \\
&  &  &  &  &  &  &  &  \\
20 & 20 & 0.997 & 0.675 & 0.628 & 0.662 & 0.022 & 1.000 & 0.215 \\
& 50 & 0.999 & 0.731 & 0.941 & 0.700 & 0.067 & 1.000 & 0.434 \\
& 100 & 1.000 & 0.728 & 0.990 & 0.705 & 0.153 & 1.000 & 0.766 \\
& 200 & 1.000 & 0.724 & 0.994 & 0.691 & 0.325 & 1.000 & 0.975 \\
&  &  &  &  &  &  &  &  \\
30 & 20 & 1.000 & 0.970 & 0.988 & $\backslash$ & 0.052 & $\backslash$ & 0.321
\\
& 50 & 1.000 & 0.982 & 1.000 & $\backslash$ & 0.138 & $\backslash$ & 0.573
\\
& 100 & 1.000 & 0.981 & 1.000 & $\backslash$ & 0.346 & $\backslash$ & 0.839
\\
& 200 & 1.000 & 0.981 & 1.000 & $\backslash$ & 0.650 & $\backslash$ & 0.991
\\
\bottomrule[2pt] &  &  &  &  &  &  &  &
\end{tabular}%
\end{table}

\subsection{Estimating the numbers of factors\label{sec:4.4}}

We investigate the empirical performances of the proposed projected iteration
procedure based on eigenvalue-ratio (PE-ER), along with the initial version
(IE-ER), total mode-$k$ correlation thresholding (TCorTh) in \cite{lam2021rank}%
, information-criterion/eigenvalue-ratio method that used TOPUP/ TIPUP/
iTOPUP/ iTIPUP in \cite{Han2021Rank}, abbreviated as TOP-IC/ TIP-IC/
iTOP-IC/ iTIP-IC/ TOP-ER/ TIP-ER/ iTOP-ER/ iTIP-ER respectively, while we
fixed $h_{0}=1$ in their implementation.

First of all, we note that the accuracy of IC method is zero under any
setting, so IC method is not an appropriate method to estimate the number of
factors. We will ignore the IC method in the following discussion. Table \ref%
{tab:3} presents the frequencies of exact estimation over 1000 replications
under Settings A, C, D and E by different methods. We set $r_{\max}=8$ for
all estimation procedures. The performance of all three methods improve as
the dimension and sample size increase. Among them, PE-ER, IE-ER, TOP-ER and
ITOP-ER are more sensitive to the increase of dimension, while TCorTH,
TIP-ER and iTIP-ER are more sensitive to the increase of sample size.

Table \ref{tab:3} also shows that the iterative algorithm generally
outperforms the non-iterative algorithm, and TOP-ER and iTOP-ER outperform
TIP-ER and iTIP-ER, respectively, under this setting. The accuracy of PE-ER
is second only to iTOP-ER. Note that iTOP-ER require much longer run time
and large storage space, since the TOPUP procedure needs to deal with a
large high-order tensor, such as $\RR^{p_1\times p_2 \times p_3 \times
p_1\times p_2 \times p_3 \times 1}$ while $h_0=1$. Therefore, iTOP-ER may
have some difficulties in practical calculation.

\end{document}